\documentclass{article}
\usepackage{microtype}
\usepackage[utf8]{inputenc}
\usepackage[T1]{fontenc}
\usepackage{graphicx,booktabs,array,calc,longtable}
\usepackage{xurl}
\usepackage[hidelinks]{hyperref}
\usepackage[preprint]{icml2026}
\makeatletter
\renewcommand{\Notice@String}{\textit{Preprint. October 27, 2026.}}
\makeatother
\hypersetup{pdfsubject={},pdfauthor={Zhong Han}}
\usepackage{balance}
\usepackage{needspace}
\usepackage{amsmath,amssymb,amsthm}
\newtheorem{theorem}{Theorem}
\newtheorem{lemma}{Lemma}
\newtheorem{proposition}{Proposition}
\newtheorem{corollary}{Corollary}
\theoremstyle{definition}\newtheorem{definition}{Definition}
\theoremstyle{remark}

\icmltitlerunning{Optimal Consistency for Interval Selection}
\begin{document}
\twocolumn[
\icmltitle{Online Interval Selection with Binary Predictions: Optimal Consistency and a Continuous Robustness Tradeoff}
\begin{icmlauthorlist}
\icmlauthor{Zhong Han}{nus}
\end{icmlauthorlist}
\icmlaffiliation{nus}{Department of Biomedical Engineering, College of Design and Engineering, National University of Singapore}
\icmlcorrespondingauthor{Zhong Han}{zhong.han.hz@u.nus.edu}
\vskip 0.3in
]
\printAffiliationsAndNotice{}
\begin{abstract}
We study online interval selection with binary predictions of membership in a final optimal solution. Requests arrive in arbitrary order; accepted requests may be revoked, but rejected or revoked requests cannot be recovered. With rewards proportional to interval length, the best consistency under finite strict robustness is $d=(3+\sqrt5)/2$ when both guarantees hold almost surely on each fixed input, including for randomized algorithms. A continuous lower bound forces consistency strictly above $d$ at every robustness budget below $4\varphi+3\approx9.472$, whereas our deterministic algorithm attains $d$ with exact robustness $\varphi^5\approx11.090$, where $\varphi=(1+\sqrt5)/2$. Each held interval inherits a ledger recording positive prediction mass lost through irrevocable decisions. A tunable replacement cost gives a continuous family with exact consistency throughout and exact robustness on two parameter intervals, including $(C,R)=(25/9,10)$ and $(4,7)$. Randomization permits expected consistency approaching one with finite expected robustness, whereas almost-sure robustness forces expected consistency strictly above two at every finite budget.
\end{abstract}

\section{Introduction}\label{sec:intro}
An online selection algorithm can revoke accepted requests but cannot recover rejected alternatives. Even perfect predictions leave a dilemma: a short request may belong to the final optimum because compatible requests arrive later, yet accepting it can be disastrous if the input stops. How much can perfect predictions help while preserving finite robustness?

In the binary optimum-membership model of \citet{karavasilis2025},
optimal consistency under two almost-sure guarantees is $d=\varphi^2$,
below the limit $3$ for every 1-increasing policy with finite robustness
\citep[Definition~4.8 and Lemma~4.9]{karavasilis2025}. Such a policy never
accepts an interval shorter than its longest conflicting holder. Shaped-BANK attains $d$ with robustness $\varphi^5$,
whereas every budget below $4\varphi+3$ forces consistency strictly above $d$.

\paragraph{Why per-execution protection?}
Consider provisional reservations for one resource: offers specify future
slots, decisions precede service, and declined offers expire. When accepted
reservations can still be cancelled and hourly values are equal, duration
is the selection objective. A guarantee for each completed batch protects
against a poor random outcome as well as incorrect advice.
Appendix~\ref{app:cancellation-costs} treats proportional cancellation costs.

Our lower bounds cover randomized algorithms with both ratios almost
surely on each fixed input; deterministic Shaped-BANK attains optimal
consistency. Mixing a prediction follower with a robust algorithm improves
expected ratios but can return zero reward with positive probability.
Almost-sure robustness with expected consistency has different necessary
bounds and an open optimum (Section~\ref{sec:randomized-main}).

\paragraph{Remembering missed predictions.}
When a positive request is rejected, BANK distributes its length among the
ledgers of conflicting holders. A replacement inherits those ledgers, so the
same held interval can require different future decisions depending on
what has already been missed. Bounding this inherited mass permits shorter
replacements while protecting both guarantees. Varying the replacement cost and both special thresholds
gives a continuous tradeoff; Section~\ref{sec:robustness} identifies its exact segments.

\subsection{Related work}
Consistency--robustness tradeoffs quantify the benefit of accurate advice
together with worst-case protection when it is wrong
\citep{mitzenmacher2020,lykouris2018,purohit2018,angelopoulos2020,wei2020tradeoffs}.
Here even exact optimum-membership advice cannot recover rejected or revoked
alternatives. Optimal consistency with a matching lower bound, an attainable budget curve and budget lower bounds quantify this irreversible decision constraint.
In our model, \citet[Observation~4.6]{karavasilis2025} excludes consistency
one by early stopping; Definition~4.8 and Lemma~4.9 give $C\ge3$ for
every finitely robust 1-increasing policy, including Revoke-Proportional (RP), and Theorem~4.10 certifies
$C(\mathrm{RP}_\lambda)\le3\lambda/(\lambda-1)$.
Revoke-Prop-Half permits shorter positive replacements without a matching
guarantee. BANK makes such replacements finitely robust below three and
matches a full-class consistency lower bound.

Replacement charging appears in any-order interval selection
\citep{borodin2023anyorder}. Karavasilis identified dependence between
direct and transferred charges as possible slack. Our joint geometric
bound determines exact robustness for two BANK parameter intervals and unchanged RP; source
measures control mass lost through rejections.

\citet{antoniadis2026switching} also use binary optimum advice and length
rewards, but require release-time order and irrevocable acceptance.
TrustGreedy~\citep{boyar2026} maximizes cardinality, receives a predicted
full input set before arrivals, and cannot revoke accepted intervals;
its replacements update the unaccepted part of a plan.
Other models permit fractional allocation, resumption, recall or unit
profits; Appendix~\ref{app:further-literature} distinguishes these scopes.

\section{Model and Main Guarantees}\label{sec:model}
A request $I=[s(I),f(I))$ has a unique identity, length $w(I)=f(I)-s(I)>0$, and a fixed prediction $p(I)\in\{0,1\}$ revealed on arrival. Positive-length overlap means conflict; endpoint contact is compatible. A feasible set contains pairwise compatible requests, with $w(X)=\sum_{I\in X}w(I)$. Distinct identities may have coincident intervals.

Requests arrive in arbitrary order. Each arrival is immediately accepted or
permanently rejected, maintaining a feasible held set $\mathcal S_t$.
Held requests may be revoked; rejected and revoked identities cannot be
recalled, but their information may be retained. Future arrivals, their
total number and termination time are unknown; any finite prefix may be complete.

\begin{definition}[Strict consistency and robustness]\label{def:performance}
For a nonempty finite labelled input $\sigma$, let $P$ be its positive
requests, $\mathrm{OPT}=\max_{X\text{ feasible}}w(X)$ and
$\mathrm{ALG}=w(\mathcal S_{\rm final})$. Predictions are perfect when
$P$ equals one final optimum, allowing ties. For a possibly randomized
algorithm and inputs fixed independently of its coins, define
\begin{equation}
 \begin{aligned}
 R_{\rm as}(A)&=\sup_\sigma\mathop{\rm ess\,sup}
                 \frac{\mathrm{OPT}(\sigma)}{\mathrm{ALG}_A(\sigma)},\\
 C_{\rm as}(A)&=\sup_{\sigma:P\text{ optimal}}\mathop{\rm ess\,sup}
                 \frac{\mathrm{OPT}(\sigma)}{\mathrm{ALG}_A(\sigma)}.
 \end{aligned}\label{eq:performance}
\end{equation}
The essential supremum is over the coins. A zero denominator gives
$+\infty$; there is no additive term. For deterministic algorithms
write $R,C$ instead of $R_{\rm as},C_{\rm as}$.
\end{definition}
Thus a finite budget $B$ gives, on every fixed labelled prefix $\tau$,
\begin{equation}
 \Pr\{\mathrm{OPT}(\tau)\le B\mathrm{ALG}_A(\tau)\}=1.
 \label{eq:finite-robust}
\end{equation}
The null set may depend on $\tau$. Predictions may be infeasible.

Write $\alpha=1/\varphi$, $d=\varphi^2$, and $M=\varphi^3$. Define
\[
 \begin{aligned}
 C^*(B)&=\inf_{A\text{ deterministic}:R(A)\le B}C(A),\\
 B_{\min}(d)&=\inf_{A\text{ deterministic}:C(A)=d}R(A),
 \end{aligned}
\]
with $\inf\varnothing=+\infty$. Define $C^*_{\rm as}(B)$ and
$B_{\min}^{\rm as}(d)$ analogously using possibly randomized algorithms
and $C_{\rm as},R_{\rm as}$.
The class defining $C^*(B)$ is empty below $M$: all-zero labels transfer the classical prediction-free lower bound \citep{tomkins1995}. The LR rule of \citet{garay1997}, in the interval formulation of \citet[Lemma~4.5]{karavasilis2025}, has $C=R=M$ because it ignores predictions.

\begin{theorem}[Optimal consistency under almost-sure guarantees]\label{thm:main}
Allowing randomized algorithms,
\begin{equation}
 \min_{A:R_{\rm as}(A)<\infty} C_{\rm as}(A)=d.
 \label{eq:main}
\end{equation}
The deterministic Shaped-BANK $T_*$ (Section~\ref{sec:algorithm}) attains the minimum, also under the stronger requirement that both guarantees hold on every execution.
\end{theorem}
\begin{corollary}[Bounds on its robustness cost]\label{cor:golden-endpoint}
The universal budget and the Shaped-BANK budget satisfy
\begin{equation}
 4\varphi+3\le B_{\min}^{\rm as}(d)\le B_{\min}(d)\le R(T_*)=B_*:=\varphi^5.
 \label{sum:eq:budget-main}
\end{equation}
The remaining factor is $\varphi^5/(4\varphi+3)\approx1.171$ for either class.
\end{corollary}
Table~\ref{tab:main-guarantees} and Figure~\ref{fig:frontier} distinguish universal
bounds from attainments. Theorems~\ref{thm:lower}, \ref{sg:thm:main}, and
\ref{sg:thm:robust-main} prove Theorem~\ref{thm:main};
Theorem~\ref{r10:thm:source-first} gives the corollary's lower endpoint.
Naive accepts a positive request exactly when it conflicts with no previously
accepted positive, and never revokes; it rejects all negative requests.
\begin{table}[t]
\caption{Principal guarantees. First block: also valid for randomized algorithms
with both guarantees almost surely, using $C_{\rm as},R_{\rm as}$.
Second: deterministic attainments. LR is due to \citet{garay1997};
RP's consistency bound to \citet{karavasilis2025}.
The last row instead uses expected reward on fixed inputs.}
\label{tab:main-guarantees}\centering\small\setlength{\tabcolsep}{3pt}
\begin{tabular}{@{}p{.45\columnwidth}p{.48\columnwidth}@{}}\toprule
Scope & Guarantee\\\midrule
Finite almost-sure robustness & $C_{\rm as}\ge d$\\
Budget $R_{\rm as}\le B$ & $C_{\rm as}\ge\max\{d,G(B),H(B)\}$\\
Optimal $C_{\rm as}=d$ & $R_{\rm as}\ge4\varphi+3$\\\midrule
LR (prediction-free) & $C=R=M$\\
SUM-BANK at $\lambda=M$ & $C=d$, $R=7\varphi+4$\\
Shaped-BANK & $C=d$, $R=\varphi^5$\\
Parametric $T_a$, $\alpha\le a\le1$ & $C=(1+a)^2$, $R\le K_a$; equality outside $(a_-,a_+)$ (Theorem~\ref{param:thm:curve})\\
RP, $\lambda\ge2$ & $C\le3\lambda/(\lambda-1)$; $R=2\lambda+2$\\\midrule
Initial Naive/LR mixture, $1<c\le M$ & $C_{\rm exp}=c$, $R_{\rm exp}=2\varphi c/(c-1)$\\\bottomrule
\end{tabular}\end{table}

\begin{figure*}[t]
\centering\includegraphics[width=\textwidth]{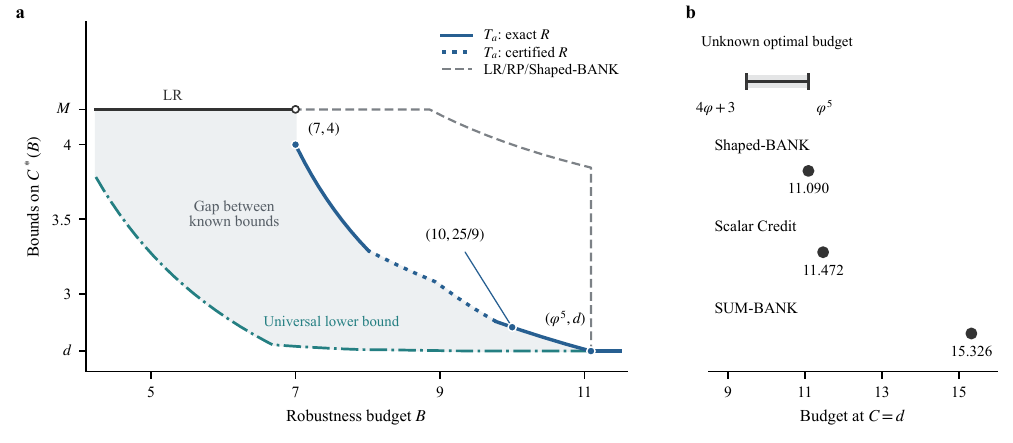}
\caption{Known performance bounds. (a) The universal lower envelope is
$\max\{d,G(B),H(B)\}$; the attainable envelope combines LR with $T_a$.
Solid $T_a$ segments have exact robustness, dotted segments only a certificate;
the shaded gap is unresolved. Near $B=6.65$, where $G$ meets $d+1/20$,
the lower envelope switches to $H$, which is proved only for $\gamma\le1/20$.
The dashed LR/RP/Shaped-BANK envelope excludes interior $T_a$ parameters;
its jump at $B=\varphi^5$ is Shaped-BANK attaining $d$.
(b) At $C=d$, the three algorithm budgets are exact, while the universal
minimum remains in the marked interval. Both panels concern almost-sure guarantees.}
\label{fig:frontier}
\end{figure*}

\section{Universal Lower Bounds}\label{sec:lower}
\begin{theorem}\label{thm:lower}
Every deterministic algorithm with finite robustness has $C(A)\ge d$. For each such algorithm and every $\epsilon>0$, a finite perfectly labelled input with a unique optimum gives ratio greater than $d-\epsilon$.
\end{theorem}
\begin{proof}
Perfect labels concern the final input; a stopping prefix may have incorrect labels. Fix a finite robustness bound $R$. Let $t=\alpha$ and choose
$0<\delta<\min\{t/2,t\epsilon/2\}$. Set $s=1-2\delta$, take an integer $m>Rs$, and partition $[\delta,1-\delta)$ into $m$ positive tiles of length $\ell=s/m<1/R$.

First offer the negative unit interval $J=[0,1)$, then the tiles. Robustness forces acceptance of $J$. At the first abandonment of $J$, only the arriving tile can remain: earlier tiles have been rejected and cannot be recalled. Stopping there gives $\mathrm{OPT}\ge1$ and $\mathrm{ALG}\le\ell<1/R$, a contradiction. Thus every tile is rejected and $J$ remains. This argument uses robustness, since the positive tiles alone are not yet optimal.

Offer the positive interval $L=[\delta-t,\delta)$. If it is accepted, stop. The positives are compatible and have weight $s+t>1$; every positive conflicts with the sole negative $J$. They therefore form the unique optimum. The algorithm holds at most $L$, giving ratio at least
\[
 (s+t)/t=d-2\delta/t>d-\epsilon.
\]
If $L$ is rejected, offer the positive unit interval $U=[1-\delta,2-\delta)$ and stop. The full positives again form the unique optimum, now of weight $s+t+1$. The algorithm can retain only $J$ or $U$, so its value is at most one and its ratio is at least $d-2\delta>d-\epsilon$. Emptying the held set only worsens either bound. 

Determinism fixes the decision on $L$ from the shared prefix, so the appropriate finite input can be chosen for $A$ without changing any label. Taking the supremum proves the result.
\end{proof}
The same argument applies to two almost-sure guarantees: fix the finite
tree of stopping and perfect terminal inputs, and intersect their
probability-one guarantee events. A tape in that intersection would
survive every branch, contradicting the deterministic argument.
All lower bounds below have such finite trees and therefore also hold
with $C_{\rm as},R_{\rm as},C^*_{\rm as}$ in place of $C,R,C^*$
(Appendix~\ref{r39:sec:finite-tree-transfer}). Tile rejection relies on unknown termination, strict ratios, and no recall.
Figure~\ref{fig:lower} illustrates the two terminal branches. Appendix~\ref{sec:witness} gives bounded-size rational witnesses for fixed $R$.

\begin{figure*}[t]
\centering
\includegraphics[width=\textwidth]{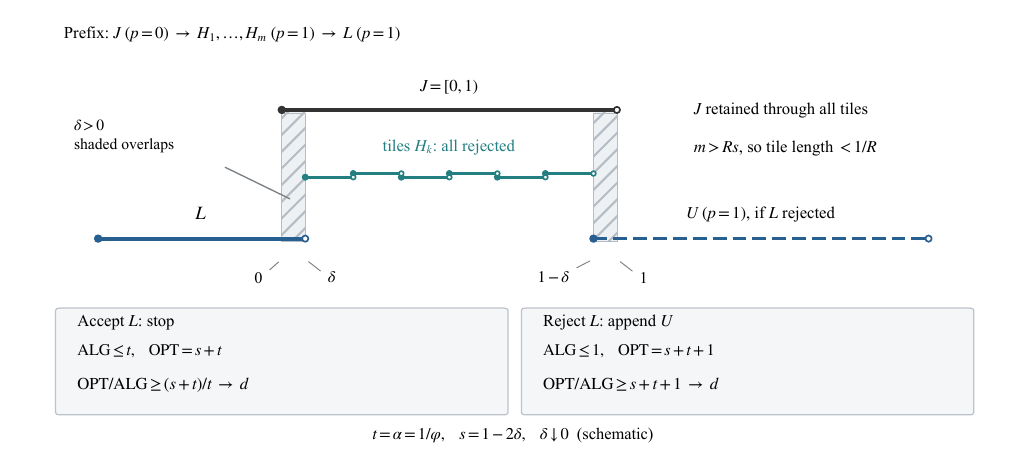}
\caption{Universal lower-bound construction. Filled left and open right endpoints indicate half-open intervals. The common prefix forces retention of $J$ and permanent rejection of every $H_i$; the decision on $L$ selects one of two complete inputs. Both have perfect predictions and a unique optimum. The geometry uses $s=1-2\delta$, $t=\varphi^{-1}$, and $\delta>0$; six central pieces are drawn schematically. }
\label{fig:lower}
\end{figure*}

\begin{theorem}[A budget-dependent lower curve]\label{r3:thm:budget-main}
For every deterministic algorithm with $R(A)\le B<\infty$ and $B\ge M$,
\begin{equation}
 C(A)\ge\max\{d,G(B)\},\qquad G(B)=\frac{B}{(\sqrt B-1)^2}.
 \label{r3:eq:budget-main}
\end{equation}
\end{theorem}
A finite binary-reserve tree amplifies missed positive mass under either
response to a negative reserve, until a perfect terminal violates
consistency (Appendix~\ref{r7:app:binary-reserve}). For example, $C^*(5)>3.27$ and $C^*(6)>2.85$.

The robustness cost remains separated from $d$ even when consistency only
approaches its optimum.
\begin{theorem}[A continuous cost barrier]\label{r10:thm:source-first}
Put $\beta=4\varphi+3$. Every deterministic algorithm with finite
robustness and $C(A)\le d$ has $R(A)\ge\beta$. More generally,
\begin{equation}
 C^*(B)\ge d+\min\left\{\frac1{100},
 \frac{\beta-B}{83\varphi+53}\right\}>d
 \quad(B<\beta).
 \label{r10:eq:budget}
\end{equation}
\end{theorem}
\paragraph{Proof sketch: missed sources and a longer reserve.}
Figure~\ref{fig:source-first-geometry} shows four forced decisions: retain the negative anchor $K$, reject the right reserve $R_0$, switch to a positive unit $J_+$ (called $J$ in the figure), then reject $I,L$. Tiny positive tiles force retention of $K$, of length near $\varphi$. Accepting $R_0$, of length near $\varphi^2$, permits a perfect completion violating consistency. To force $J_+$, a remote anchor and boundary requests repair full-history optimality after either response; analogous completions exclude $I,L$. The realized history contains a compatible packing of limiting weight $\varphi^3+1+\varphi+\varphi^2=\beta$ while only $J_+$ survives. Hypothetical completion requests are not counted.

For $C\le d+\gamma$, suppose $R(A)<D_{\rm l}(\gamma)$ and choose
$r$ strictly between them. A finite remote anchor of length $h$ consumes
consistency slack $\gamma h$. Choosing $h=(r-1)a+\varepsilon$ after
a hypothetical switch to held value at most $a$, and reducing the
initial anchor, preserves the four forcing steps.
Appendix~\ref{r11:app:continuous-source-first} supplies finite parameters
and strict full-history optimality certificates. Its implicit bound
$D_{\rm l}(\gamma)$ satisfies $D_{\rm l}(0)=\beta$ and
$D_{\rm l}(\gamma)\ge\beta-(83\varphi+53)\gamma$ on $[0,1/100]$;
approaching equality budgets from below proves~\eqref{r10:eq:budget}.

The stronger implicit bound gives $D_{\rm l}(.005)\approx8.680$ and $D_{\rm l}(.01)\approx8.088$. A shorter reserve gives
$D_{\rm s}$ on $0\le\gamma\le1/20$, useful farther from $d$: $D_{\rm s}(.02)>7.41$. Figure~\ref{fig:frontier} combines them as
$H(B)=\max\{H_{\rm s}(B),H_{\rm l}(B)\}$, where
\[
 H_j(B)=d+\sup\bigl(\{0\}\cup
 \{\gamma\in[0,\gamma_j]:B<D_j(\gamma)\}\bigr),
\]
with $\gamma_{\rm s}=1/20$ and $\gamma_{\rm l}=1/100$.

\section{The BANK Mechanism}\label{sec:algorithm}
Each held request has a nonnegative ledger $b(J)$. A positive arrival
adds its length once: acceptance inherits revoked ledgers, while
rejection distributes the length among conflicts. Negative rejection
adds nothing. Ledgers preserve prediction mass lost through irrevocable
decisions.

For conflicts $D$, write $W=w(D)$ and $b_D=\sum_{J\in D}\bar b_J$, where $\bar b_J$ is the snapshot before this arrival. The ordinary gate $\mathcal Q(D)$ is a sum of replacement costs. SUM-BANK uses cost $\lambda w(J)$ ($\lambda=M$ unless stated) for every conflict. Scalar Credit-BANK uses
$\min\{Mv,dv+cb\}$ for a negative conflict, with a fixed coefficient $c\ge0$. Shaped-BANK retains cost $Mv$ for a positive conflict of length $v$, but uses
\begin{equation}
 q_*(v,b)=\min\{2\varphi v,\max\{dv,\varphi v+\alpha b\}\}
 \label{sg:eq:main-cost}
\end{equation}
for a negative conflict. Thus $\mathcal Q(D)=\sum_{J\in D}\tau(J)$, with $\tau(J)=Mw(J)$ or $q_*(w(J),b(J))$ according to its label. The negative cost stays at $dv$ up to $b=\varphi v$, increases linearly up to $b=dv$, and then remains at $2\varphi v$.

After ordinary acceptance fails, a positive request can still replace all its conflicts if they are negative and
\[
 w\ge\alpha W,\qquad w\ge b_D/\varphi.
\]
The first test permits replacing longer conflicts, unlike 1-increasing policies. More generally, an $a$-increasing rule requires $w\ge a\max_{J\in D}w(J)$ when accepting with $D\ne\varnothing$; Shaped-BANK satisfies this for $a=\alpha$. The second prevents the same missed positive mass from being discarded repeatedly. Ordinary acceptance has priority; all thresholds accept equality.

For a rejected positive, put $Z=[\forall J\in D:p(J)=0]$. The simultaneous fees are
\begin{equation}
 g_J=\begin{cases}
 w\,w(J)/W,&Z,\ w<\alpha W,\\
 w\,\bar b_J/b_D,&Z,\ w\ge\alpha W,\ w<b_D/\varphi,\\
 w\,w(J)/W,&\neg Z.
 \end{cases}\label{eq:fees}
\end{equation}
These cases are exhaustive on rejection; $\neg Z$ means that a positive conflict exists.
\begin{algorithm}[t]
\caption{Shared ledger policy with ordinary gate $\mathcal Q$}
\label{alg:bank}
\begin{algorithmic}[1]
\STATE Initialize $\mathcal S\gets\varnothing$.
\STATE On arrival of $(I,p)$, with $w=w(I)$:
\STATE $D\gets\{J\in\mathcal S:J\text{ conflicts with }I\}$; $W\gets w(D)$.
\STATE Snapshot $\bar b_J\gets b(J)$ for every $J\in D$.
\STATE $b_D\gets\sum_{J\in D}\bar b_J$; $Z\gets[\forall J\in D:p(J)=0]$.
\STATE Test ordinary $w\ge\mathcal Q(D)$ first (true if $D=\varnothing$).
\STATE If it fails, test special $p=1$, $Z$, $w\ge\alpha W$, $w\ge b_D/\varphi$.
\IF{either test accepts}
  \STATE Accept $I$; permanently revoke all of $D$.
  \STATE $\mathcal S\gets(\mathcal S\setminus D)\cup\{I\}$; $b(I)\gets b_D+pw$.
\ELSE
  \STATE Permanently reject $I$.
  \IF{$p=1$}
    \STATE Compute every $g_J$ by \eqref{eq:fees} from snapshots.
    \STATE Simultaneously set $b(J)\gets\bar b_J+g_J$ for $J\in D$.
  \ENDIF
\ENDIF
\STATE All other existing ledgers are unchanged; output $\mathcal S$ at termination.
\end{algorithmic}
\end{algorithm}

The compatible-positive filter keeps a positive label only if its request is compatible with all earlier kept positives (Appendix~\ref{new:sec:wrapper}). For fixed parameters, this implementation uses $O(n^2)$ exact arithmetic
operations/comparisons and $O(n)$ stored scalars or interval records, including filtering (Appendix~\ref{r39:sec:arithmetic-resources}).

\begin{lemma}[Conservation]\label{lem:conservation}
For arbitrary labels and each time $t$,
$\sum_{J\in\mathcal S_t}b(J)=w(P_t)$, where $P_t$ is the positive set revealed by time $t$. All ledgers are finite and nonnegative.
\end{lemma}
\begin{proof}
Acceptance removes the sum $b_D$ and adds $b_D+pw$. A rejected negative adds zero; a rejected positive adds fees summing to $w$. On rejection $W>0$, and a ledger-test failure has $b_D>0$. Induction proves both assertions.
\end{proof}

\paragraph{Why length alone does not determine a replacement.}
For $0<\varepsilon<1/8$, start with negative $[0,1)$ and reveal
the compatible positives $[\varepsilon,1/2)$,
$[1/2,1-\varepsilon)$ and $[-\alpha+2\varepsilon,\varepsilon)$.
All three are shorter than $\alpha$ and fail both acceptance routes.
The holder remains $[0,1)$, but its ledger grows from zero to
$\varphi-3\varepsilon>1$.
Now positive $[1-\varepsilon,1-\varepsilon+\alpha)$ passes the
length test at equality. Its ledger test fails:
$b_D>1=\varphi w$. Removing that test would leave reward $\alpha$
and a perfect-label ratio approaching $\varphi+2>d$.
The complete five-request calculation is in
Appendix~\ref{app:ledger-ablation}.

\paragraph{The price of a special replacement.}
A negative length-$\varphi$ anchor has zero-ledger cost $M$.
A unit positive replaces it at the special threshold and acquires
cost $M$. Rejections around these successive holders yield a compatible
packing approaching $\varphi^5$ while only the unit survives
(Proposition~\ref{sg:prop:floor}).

\section{Why the Ledger Gives Optimal Consistency}\label{sec:attainment-main}
\begin{theorem}[A sufficient consistency interface]\label{sg:thm:main}
Use positive cost $Mv$, the summed ordinary gate, and the shared BANK rules. Suppose the negative cost is finite and satisfies $q(v,b)\ge dv$ for all $b\ge0$, and
\begin{equation}
 v+\max\{\alpha v,\alpha^2b\}\le\alpha q(v,b)
 \quad (0\le b\le dv).\label{sg:eq:interface-main}
\end{equation}
If the complete positive set $P$ is feasible, every held request satisfies $b(J)\le dw(J)$. Consequently $w(P)\le d\,\mathrm{ALG}$. Shaped-BANK and SUM-BANK with cost $Mv$ satisfy this interface; their finite robustness and Theorem~\ref{thm:lower} imply $C=d$.
\end{theorem}
\paragraph{Proof sketch: full source shares and directional cuts.}
Fix a finite execution with feasible complete positive set $P$.
At snapshot $n$, conserved fractions $a_{X,n}(I)$ assign every arrived
positive source $I$ among held identities $X$. The measure
\begin{equation}
 \mu_{X,n}(A)=\sum_{I\in P}a_{X,n}(I)|I\cap A|
 \label{sg:eq:measure-main}
\end{equation}
has mass $b_n(X)$ and density at most one; unarrived sources have zero
coefficients. Its support is not clipped to $X$: replacing positive
$[0,1)$ by negative $[.9,6)$ inherits mass one, not $.1$.
Let $E_L,E_R$ denote left and right exterior mass and
$L_{n,X}(s)=\sum_{I\in P:f(I)\le s}a_{X,n}(I)w(I)$ the \emph{full} left shares.

\begin{lemma}[Side and directional-cut invariant]
\label{sg:lem:main-invariant}
Under the cost and BANK-rule hypotheses of
Theorem~\ref{sg:thm:main}, fix a finite execution with feasible complete
positive set $P$. Use the conserved source measure and full shares
defined above. Write $\nu_X$ for birth and $\delta_X$ for revocation,
or one beyond the final arrival if $X$ survives.
For every negative $V=[x,x+v)$ and live snapshot
$\nu_V\le n<\delta_V$, each side $S\in\{L,R\}$ satisfies
\[
 \begin{aligned}
 E_{S,\nu_V}(V)&\le\alpha v,\qquad b_n(V)\le dv,\\
 E_{S,n}(V)&\le\max\{\alpha v,\alpha^2b_n(V)\}.
 \end{aligned}
\]
For each actual positive $T=[s,t)\in P$ arriving after $\nu_V$, if
$s<z\le t$, $z\le x$ and $L_{\nu_V,V}(s)>0$, then
\[
 L_{\nu_V,V}(s)+z-s\le\alpha v.
\]
At each live snapshot $n$, a positive $Q$ of length $u$ has
$b_n(Q)\le du$ and, for every
actual positive $T=[s,t)$ arriving after its live snapshot $n$,
\[
 L_{n,Q}(s)+z-s\le du
 \quad(s<z\le t\le s(Q),\ L_{n,Q}(s)>0).
\]
Both cuts include their reflections under $y\mapsto-y$.
The complete proof is in Appendix~\ref{app:consistency-interface}.
\end{lemma}

\paragraph{Why the negative capacity follows.}
Induct on accepted birth order, establishing every live snapshot of
each earlier identity before using it at a parent's birth.
Earlier children's cuts and the cost interface pay the new negative
birth bounds. The positive cut uses the three cases below. Independently, replacement-tree geometry bounds subtree length by $2v$ at a negative root and $Mv$ at a positive root (Appendix~\ref{app:shaped}).
Density bounds the mass inside $V$ by $v$.
A side with nonzero exterior mass at birth stays bounded by $\alpha v$:
an existing crossing source excludes another, while wholly exterior
birth sources and a later crossing source share the birth-cut budget.
Each initially empty side has at most one later crossing source.
Its fee obeys $g<\alpha v$ after length failure, or
$g\le\alpha b_{\rm pre}$ after ledger failure, which implies
$g\le\alpha^2(b_{\rm pre}+g)$.
These are event-time bounds.

With two occupied birth sides, $b\le(1+2\alpha)v<dv$.
With one, the other side's first crossing has
$b_{\rm pre}\le(1+\alpha)v$, so its fee is at most
$\alpha(1+\alpha)v=v$. Without that crossing, density already suffices.
Thus
\[
 b\le\underbrace{v}_{\text{inside}}+
       \underbrace{\alpha v}_{\text{protected side}}+
       \underbrace{v}_{\text{other side}}=dv.
\]
With no occupied birth side, the first crossing has
$b_{\rm pre}\le v$ and fee at most $\alpha v$, reducing to the
preceding case. A source crossing both endpoints pays one fee and
excludes every later crossing. Thus event-time bounds close capacity.

At a positive birth, let $B_Q,W_Q$ be its negative children's total
ledger and length. Ordinary acceptance gives $B_Q\le u$, $W_Q\le u/d$;
special acceptance gives $B_Q,W_Q\le\varphi u$. Adding $Q$ gives
capacity, and compatibility excludes later positive fees.
For the cut, set $e=z-s$, $G=s(Q)-s\ge e$. Trace a left source share
backward to the first support starting below $s$. There are three cases.

(i) An intermediate positive $H$ occurs before that crossing.
Shrinking paths give $w(H)\le u/d$ and $e<dw(H)\le u$;
the cut is at most $B_Q+u\le du$.

(ii) Otherwise the negative path starts with child $V$, of length $v$.
For ordinary $Q$, $G<\varphi v\le\alpha u$.
For special $Q$, a nonleftmost $V$ gives $G<\alpha v\le u$;
if another child $K$ carries left shares, protrusion gives
$\varphi G<v+w(K)\le\varphi u$. Again $B_Q+u$ pays the cut.

(iii) Only the leftmost child $V$ carries shares and $Q$ is special.
If $s(V)\ge s$, its inherited birth cut plus the remaining distance
gives a bound below $\varphi v$. Otherwise density gives
\[
 L_{n,Q}(s)+e<v+E_L(V)
 \le\max\{\varphi v,v+\alpha^2b(V)\}\le du.
\]
Here the length gate pays for $v\le\varphi u$ and the ledger gate
for $b(V)\le\varphi u$; neither budget alone suffices for this bound.
Reflection and source conservation finish
$w(P)=\sum_Xb(X)\le d\,\mathrm{ALG}$.
Appendix~\ref{app:shaped} proves the geometric and boundary-case details.

At the golden endpoint, $q_*$ is the smallest cost nondecreasing in
$b$ that satisfies this sufficient interface
(Proposition~\ref{sg:prop:minimal-main}). The interface itself is not necessary: setting $q(v,0)=\varphi v$ and keeping $q_*$ for $b>0$
also preserves $d$-consistency, by pruning source-free children
(Appendix~\ref{r29:sec:necessary-cost}).
The relaxed policy still satisfies
the nonshrinking hypothesis below, so its robustness is at least $\varphi^5$.

\section{Exact Robustness and the Remaining Gap}\label{sec:robustness}
\begin{theorem}[Exact Shaped-BANK robustness]\label{sg:thm:robust-main}
For arbitrary finite labelled inputs, $R(T_*)=\varphi^5$.
\end{theorem}
\begin{proof}[Proof sketch: joint endpoint charges]
In the replacement forest, revoked requests are children of their replacements. Give each accepted identity its ordinary cost just before revocation or final termination. Costs do not decrease during a lifetime. An ordinary parent of length $h$ pays at most $h$ in child costs. At a special parent, its negative children have total length and ledger $W,b_D\le\varphi h$, and
\[
 \begin{aligned}
 q_*(v,b)&\le dv+\alpha^3b,\\
 \sum_iq_i&\le dW+\alpha^3b_D\le Ah,\quad A=\varphi+3.
 \end{aligned}
\]
For analysis, split each rejected optimal request among its current conflicts in proportion to their costs. These receivers form an ancestry antichain: no receiver is an ancestor of another. The comparator is a fixed offline optimum. Passing the charges to final roots preserves their sum. At most two disjoint optimal requests cross the boundary of a root's connected accepted hull; the others fit between their cuts. Bounding the two exterior charges jointly with the space remaining between them avoids counting the same interval twice.

Let $H_J$ be the connected accepted hull of a root $J$. For a cut $x$, define $L_J(x)$ as the largest sum of terminal costs over an ancestry antichain in this subtree whose intervals start strictly left of $x$; define $U_J(y)$ by ends strictly right of $y$. The cuts $x,y$ are the inward endpoints of the crossing comparator intervals, using the hull endpoints when a crossing is absent. Unless one comparator covers the entire hull, feasibility gives $x\le y$ and the root charge is bounded by
\[
 F_J(x,y)=|H_J\cap[x,y]|+L_J(x)+U_J(y).
\]
The two maxima need not use the same antichain. A comparator covering the entire hull is treated separately: it is the only contributor there and is charged at most $Ah$.

At an ordinary node, its own cost dominates every descendant antichain's total terminal cost, so the node itself maximizes an endpoint whenever eligible. At a special node $J=[a,b)$, each endpoint maximum is either the root's singleton cost or the sum of child maxima. Orient a single singleton to the left. If $y<b$, the gap lies inside $J$; if $y\ge b$, at most one compatible child extends beyond $y$. With both maxima supplied by children, at most one child bridges both cuts. These alternatives exhaust the geometry, including coincident cuts and equality at an interval endpoint. The coefficients are in Table~\ref{sg:tab:coefficients} (Appendix~\ref{curve:sec:robustness}).

The largest case is a special root with one endpoint charged at the root and the other at an outward child. Write $\rho=A/M$. Monotonicity and the fixed-ledger slope inequality $q_*(W,b_D)-q_*(v,b_D)\ge\varphi(W-v)$ absorb that child's cost and excess length. Since $\rho-1<\varphi$, the total charge is at most
\begin{align}
 &(M+1)h+W+q_*(W,b_D)\notag\\
 &\quad\le(M+1+\varphi)h+q_*(\varphi h,\varphi h)
 =\varphi^5h.\label{sg:eq:compression-main}
\end{align}
The other endpoint choices give smaller coefficients; Appendix~\ref{app:shaped} lists and proves them. Summing over final roots proves the upper bound. The five-request construction in Proposition~\ref{sg:prop:floor} approaches $\varphi^5$, proving equality.
\end{proof}

The endpoint cannot be improved merely by reshaping nonshrinking negative costs ($q_-(v,b)\ge v$) within the same special tests: $d$-consistency then
forces $R\ge\varphi^5$, even without cost continuity or monotonicity
(Appendix~\ref{r29:sec:necessary-cost}). To trade consistency for a
smaller budget, we change both special thresholds below.

\subsection{A continuous attainable tradeoff}
The optimal-consistency endpoint is part of a tunable family. Fix
$a\in[\alpha,1]$ and put
\[
 m_a=\max\{(1+a)/a^2,1+a+a^{-1}\}.
\]
Policy $T_a$ uses positive cost $m_av$, negative cost
\begin{equation}
 \begin{aligned}
 q_a(v,b)=\min\bigl\{&(1+a+a^{-1})v,\\
             &\max\{(1+a^{-1})v,v/a+b/(1+a)\}\bigr\},
 \end{aligned}
 \label{param:eq:cost-main}
\end{equation}
and special tests $w\ge aW$, $w\ge ab_D$. In all fee-case guards,
replace $\alpha W$ by $aW$ and $b_D/\varphi$ by $ab_D$;
the proportional fees, priority and inheritance are unchanged.

\begin{theorem}[Continuous attainable guarantees]\label{param:thm:curve}
For every $a\in[\alpha,1]$, $C(T_a)=(1+a)^2$ and
$R(T_a)\le K_a=\max\{O_a,N_a\}$, where the outward and no-bridge
coefficients are defined in~\eqref{curve:eq:two-coefficients}: $O_a$ comes from a special root with one outward child; $N_a$ from child charges with no child bridging both cuts.
The bound is exact on $[\alpha,a_-]\cup[a_+,1]$, where
$a_-\approx0.677821$ and $a_+\approx0.813824$ are the unique roots
specified by~\eqref{curve:eq:transition-polynomials}
in Proposition~\ref{curve:prop:exact-arcs}. In particular,
\[
 \begin{array}{c|ccc}
 a&\alpha&2/3&1\\\hline
 C(T_a)&d&25/9&4\\
 R(T_a)&\varphi^5&10&7
 \end{array}
\]
All three pairs are exact, and the guarantees hold after the
compatible-positive filter.
\end{theorem}
\paragraph{Proof sketch: mechanism of the tradeoff.}
The source induction separates negative capacity $(1+a)^2v$ from
positive capacity $am_av$. The exterior-mass bounds become $av$
at birth and $\max\{av,ab/(1+a)\}$ during a negative lifetime.
Their weighted cut is paid by $q_a$. In the positive cut's multichild
case, the accepted-history bounds change branch at $a^3+a^2=1$;
both branches pay the same directional gap. The joint-cut analysis
reduces to the two continuous coefficients $O_a,N_a$.
Five-request families match the consistency throughout the range
and the robustness on the two stated intervals.
Appendix~\ref{app:shaped} gives the induction, coefficient comparison
and matching constructions.
At budget $7$, $T_1$ improves LR's $M\approx4.236$ and RP's certificate
$5$; continuity extends this improvement to an interval. At budget
$10$, $T_{2/3}$ has exact consistency $25/9<3$. For $a_-<a<a_+$,
only $O_a\le R(T_a)\le N_a$ is known.

The envelope combines LR's $C=R=M$, RP's
$C\le3(B-2)/(B-4)$ for $B\ge6$, and $T_a$.
RP has exact robustness $2\lambda+2$ for $\lambda\ge2$
(Appendix~\ref{sum:app:robustness}); its consistency certificate beats LR
above $4+3\varphi\approx8.854$. Appendix~\ref{app:additional-policies} compares SUM and scalar Credit.
For $M\le B<7$, LR remains the best guarantee among these policies.

\section{Guarantees and Learning Interfaces}\label{sec:randomized-main}
\subsection{Randomization and the guarantee being protected}
For fixed inputs, measure $\mathrm{OPT}/\mathbb E[\mathrm{ALG}]$.
An initial coin selects \textsc{Naive} with probability $\pi$, otherwise LR. Its exact tradeoff is
\begin{equation}
 \overline R=\frac{2\varphi c}{c-1},\quad 1<c\le M.
 \label{eq:mixture-main}
\end{equation}
Here $\pi=(M/c-1)/(M-1)$. At $c=d$, $\overline R=2d$;
at $\overline R=\varphi^5$, $c\approx1.412$. On all-negative inputs
the Naive branch gains zero. The finite-expected-robustness consistency
infimum is one, unattained; the full frontier remains open
(Appendix~\ref{r12:sec:randomized-lower}).

With only robustness almost surely and consistency in expectation,
every finite budget $B>1$ instead requires
(Proposition~\ref{r38:prop:hybrid-simple})
\[
 C_{\rm exp}\ge2+\frac1{B-1}>2.
\]
Robustness on a fixed shorter input forces rejection on the shared
prefix; a fixed longer perfect input then gives the bound.

A half-length special gate with ordinary factor four allows pathwise
robustness at most $18$, a sharp uniform bound over this action family,
but every fixed independent acceptance probability
has $C_{\rm exp}\ge d$ (Appendix~\ref{r37:sec:hybrid}).

For $0<\pi<1$, adapting inputs to observed actions makes the mixture's
expected robustness infinite (Appendix~\ref{r12:sec:adaptive-scope}; \citealp{bendavid1994}). This further motivates protection of individual executions rather than relying solely on a mixture's fixed-input expected reward.

\raggedbottom
\balance
\subsection{Workload and learned-label diagnostics}\label{sec:empirical}
NASA-iPSC and CTC-SP2 give 18,066 and 77,205 intervals.
Twenty shared reveal/label repetitions flip members (FN) or nonmembers
(FP) of one optimum; FP 1 is all-positive, FN 1 all-negative.

Filtering preserves both exact guarantees and adds
$\mathrm{OPT}\le d\,\mathrm{ALG}+\eta_K$. Here $\eta_K$ is missed optimal
weight plus each false positive's nonnegative excess of conflicting
optimal weight over its own length (Appendix~\ref{new:sec:wrapper}).
Zero conflict error need not make positives feasible, so raw SG lacks
this guarantee.

Table~\ref{tab:main-mixtures} (Appendix~\ref{r5:app:current-comparison}) reports
means, intervals and individual-run fifth percentiles $q_{.05}$.
Intervals use 10,000 shared whole-input bootstrap resamples
(Appendix~\ref{r39:sec:saved-quantiles}).

LR leads FP $.1$ and $1$ in mean and $q_{.05}$. In the two low-FP
rows of Table~\ref{tab:main-mixtures}, SG$^F$ gains $.245$ and $.532$ points over LR
(paired intervals $[.076,.401]$ and $[.365,.693]$), and raw SG gains a further $.527$ and $.916$ (Appendix~\ref{r5:app:current-comparison}), so filtering is not the source. Both rows use the smallest rate of the prespecified low-FP grid (Appendix~\ref{sec:experiment-protocol}). Perfect-label Naive gets $100\%$ but lacks finite robustness.
Filtered SG exceeds filtered RP in all 20 cell means; its differences
from best scalar Credit range from about $-.001$ to $.04$ points (Appendix~\ref{r5:app:current-comparison}).

Chronological logistic advice on 7 NASA and 31 CTC instances gives
SG$^F$ $94.58/89.66\%$, below LR's $96.86/91.25\%$.
Reward selection on reused blocks mainly chooses all-negative advice;
the nonconstant CTC cost score remains below LR
(Appendices~\ref{sec:r029-extensions}--\ref{sec:r030-decision-alignment}).

\subsection{What should a learned advice bit predict?}
The same rejection can leave the same held interval but different ledgers,
changing a later replacement; thus the useful prediction can concern
continuation reward rather than immediate reward.
Appendix~\ref{app:rollout-interface} lets the bit estimate its rollout effect on final reward. Its transfer bound uses the reference policy's membership error and deployment-state rollout error; the latter is not estimated here. An independent cohort favors lookahead but does not confirm additional learning gains.

\section{Discussion}\label{sec:scope}
The universal budget for attaining $d$ remains in
$[4\varphi+3,\varphi^5]$. The optimal universal frontier and the tradeoff
with almost-sure robustness but expected consistency are open; within
$T_a$, only $(a_-,a_+)$ lacks exact robustness.

Improving the budget upper bound $\varphi^5$ requires leaving policies
that always reject negative arrivals shorter than $d$ times a sole
negative holder, even with arbitrary history-dependent ledgers and
positive rules (Appendix~\ref{app:history-barrier}).

A filter-free policy with $C=d$, finite robustness and a smooth
conflict-error guarantee remains open; weighted precision and recall
need not control this error (Appendix~\ref{r12:sec:precision}).

\label{page:main-end}
\clearpage
\nobalance
\label{page:refs-start}
\section*{Impact Statement}
The results concern retained length under unrestricted revocation.
Proportional fees account for one cancellation cost, but not reliability
obligations or unequal consequences across users. Training labels inherit
their offline objective's biases; neither optimal membership nor these
competitive ratios ensure fair allocation. The public job traces are
used only to construct selection instances, without inferring individual
behavior.

\bibliographystyle{icml2026}
\bibliography{references}
\label{page:refs-end}
\clearpage
\onecolumn
\raggedbottom
\makeatletter\setlength{\@fptop}{0pt}\makeatother
\appendix
\numberwithin{equation}{section}
\numberwithin{theorem}{section}
\numberwithin{lemma}{section}
\numberwithin{proposition}{section}
\numberwithin{corollary}{section}
\numberwithin{definition}{section}
\section*{Technical Appendix}

\subsection*{Guide to the appendix}
\begin{center}\small
\begin{tabular}{@{}p{.47\textwidth}p{.45\textwidth}@{}}\toprule
Question & Reading route\\\midrule
Universal consistency and budget lower bounds & Appendices~\ref{sec:proof-lower}--\ref{r11:app:continuous-source-first}\\
Parametric guarantees, exact intervals and cost limits & Appendix~\ref{app:shaped}, especially Proposition~\ref{curve:prop:exact-arcs}; penalties in Appendix~\ref{app:cancellation-costs} and history-dependent restriction in Appendix~\ref{app:history-barrier}\\
Credit, SUM and RP guarantees & Overview in Appendix~\ref{app:additional-policies}; proofs in Appendices~\ref{app:credit}--\ref{sum:app:robustness}\\
Filtering and randomized guarantees & Appendix~\ref{new:sec:wrapper}; randomized guarantees in Appendix~\ref{w02r4:sec:randomized-scope}, including the sharp safe action family in Proposition~\ref{r38:prop:random-action-envelope}\\
Mechanism experiments and protocol & Results in Appendices~\ref{app:ledger-ablation}--\ref{r5:app:current-comparison} and~\ref{app:randomized-experiment}; shared protocol in Appendix~\ref{sec:experiment-protocol}\\
Learned labels and rollout advice & Appendix~\ref{sec:r030-decision-alignment}; rollout and its expected reference certificate in Appendix~\ref{app:rollout-interface}\\
Model and literature comparisons & Appendix~\ref{sec:comparison-appendix}\\\bottomrule
\end{tabular}\end{center}

\subsection*{Notation guide}
The algorithm stores only requests, labels and ledgers. Source measures, cuts, capacities and replacement histories are proof quantities.
\begin{center}\small
\begin{tabular}{@{}p{.25\textwidth}p{.67\textwidth}@{}}\toprule
Notation & Meaning\\\midrule
\multicolumn{2}{@{}l}{\emph{Core policies and error}}\\
$\varphi,\alpha,d,M$ & Golden ratio, $\alpha=1/\varphi$, optimal consistency $d=\varphi^2$, and positive cost coefficient $M=\varphi^3$.\\
$T_*,q_*(v,b)$ & Shaped-BANK and its negative cost for length $v$ and ledger $b$, defined in~\eqref{sg:eq:cost}.\\
$T_a,q_a,C_a,K_a$ & Parametric policy, negative cost, exact consistency $(1+a)^2$, and robustness certificate for $a\in[\alpha,1]$.\\
$O_a,N_a,a_-,a_+$ & Outward and no-bridge coefficients, with $K_a=\max\{O_a,N_a\}$; exact robustness holds on $[\alpha,a_-]\cup[a_+,1]$.\\
$\tau(J),q_J$ & A held request's current ordinary cost and its terminal cost at revocation or final termination; subscripts specify the policy.\\
$W,b_D,b(J)$ & Total conflicting length, total conflicting ledger, and the ledger of a current request; some appendix proofs denote the total ledger by $B$.\\
$\eta_K(P,O)$ & Conflict-sensitive error of original labels relative to a fixed optimum; subscript $K$ names the metric.\\
\multicolumn{2}{@{}l}{\emph{Lower bounds and attained budgets}}\\
$B_*,B_{C*},B_S$ & Exact Shaped-BANK, best scalar Credit and SUM endpoints $\varphi^5$, $4\varphi+5$ and $7\varphi+4$.\\
$G(B),D_{\rm s},D_{\rm l},H$ & Binary-reserve lower curve, short/long-reserve necessary budgets, and $H=\max\{H_{\rm s},H_{\rm l}\}$ from their inverse consistency bounds.\\
$\beta$ & Local notation: $4\varphi+3$ in Appendix~\ref{r11:app:continuous-source-first}; $\varphi$ in Appendix~\ref{sharp:sec:threshold}; a free transfer parameter in Appendix~\ref{sum:app:robustness}.\\
$\pi,r_{\rm LR}$ & Initial-mixture probability of Naive and the LR ratio $r_{\rm LR}=M$.\\
\multicolumn{2}{@{}l}{\emph{Appendix variants and proof quantities}}\\
$\beta_a,q_{0,a},Q_a,m_a,k_a$ & Appendix~\ref{app:shaped}'s parameter-dependent cost and capacity coefficients; $Q_a$ is distinct from a request $Q$.\\
$h_{0,a},h_{1,a},\theta_a,c_a,A_a,\rho_a$ & Parameter-dependent history, source-share and terminal-cost coefficients in~\eqref{curve:eq:constants} and~\eqref{curve:eq:robust-constants}; $a_c^3+a_c^2=1$ marks the history branch change.\\
$A_\lambda$ & Descriptive MAX-BANK comparator with $\alpha=1/\varphi$, $\beta=\varphi$ and ordinary maximum threshold $\lambda$.\\
$S_\lambda$ & SUM-BANK with the same special tests and ledgers, and ordinary total-conflict threshold $\lambda$.\\
$T_c,c_*=\alpha^3$ & Scalar Credit-BANK and its best consistent coefficient; $T=T_{\varphi/2}$ denotes the fixed $c=\varphi/2$ comparison.\\
$a_{X,n}(I),\mu_{X,n},E_L,E_R$ & Conserved share of positive source $I$ in $X$, its source measure, and the mass exterior to $X$ on each side.\\
$P,P'$ & Original positive set and permanent feasible positive history produced by the filter.\\
\bottomrule
\end{tabular}\end{center}

\section{Finite universal witnesses}\label{sec:proof-lower}\label{sec:witness}
Theorem~\ref{thm:lower} gives the complete stopping-prefix proof.
For its bounded-size refinement, fix a finite robustness bound $R\ge1$ with
\begin{equation}
 \mathrm{OPT}(\tau)\le R\mathrm{ALG}_A(\tau)
 \quad\text{on every nonempty labelled prefix }\tau.
 \label{eq:old-L1}
\end{equation}
Write $J=[0,1)$, $s=1-2\delta$ and $\ell=s/m<1/R$. The positive requests are
\begin{align}
 H_k&=[\delta+(k-1)\ell,\delta+k\ell),\quad k=1,\ldots,m,\label{eq:old-L3}\\
 L&=[\delta-t,\delta),\label{eq:old-L4}\\
 U&=[1-\delta,2-\delta).\label{eq:old-L5}
\end{align}
The negative anchor $J$ is followed by the $H_k$ and $L$; $U$ arrives
only if $L$ is rejected. Thus $\sigma^-=(J,H_1,\ldots,H_m,L)$ and
$\sigma^+=(J,H_1,\ldots,H_m,L,U)$. Figure~\ref{fig:lower} shows both branches.

\begin{corollary}[Bounded-size witnesses with a prescribed ratio]\label{cor:witness}

Fix a deterministic algorithm \(A\) in the model in the main text and any valid finite robustness bound \(R\ge1\) satisfying \eqref{eq:old-L1}. For every real number \(\kappa\in(2,d)\), there is a finite input \(\sigma_\kappa\) with at most \(\lceil R\rceil+3\) requests such that its positive requests form its unique optimal solution and \[
\frac{\mathrm{OPT}(\sigma_\kappa)}{\mathrm{ALG}_A(\sigma_\kappa)}=\kappa.
\] If \(\kappa\) is rational, all interval endpoints can be chosen rational. In particular, for this same \(A\) and \(R\), every \(\varepsilon>0\) admits a rational-endpoint input with at most \(\lceil R\rceil+3\) requests, unique optimum, perfect predictions, and ratio greater than \(d-\varepsilon\).

\end{corollary}

\begin{proof}

Fix \(\kappa\in(2,d)\), and set \begin{equation}
t=\frac{\kappa-1}{\kappa},\qquad
s=\frac{(\kappa-1)^2}{\kappa},\qquad
\delta=\frac{1-s}{2}=\frac{3-\kappa-\kappa^{-1}}2,
\qquad m=\lceil R\rceil,\qquad \ell=\frac{s}{m}.
\label{eq:old-A1}
\end{equation} Since \(2<\kappa<d\), we have \(\kappa^2-3\kappa+1<0\), and hence \(0<s<1\) and \(\delta>0\). Moreover, \begin{equation}
0<t<1,\qquad t-2\delta=\kappa-2>0,\qquad
s+t=\kappa-1>1.
\label{eq:old-A2}
\end{equation} In particular, \(0<\delta<t/2<1/2\). Also, \begin{equation}
m\ge R>Rs,\qquad
0<\ell=\frac{s}{m}\le\frac{s}{R}<\frac1R.
\label{eq:old-A3}
\end{equation} Thus the same \(m\), chosen solely from \(R\), works for every \(\kappa\) in the stated range.

Use the requests \(J,H_1,\ldots,H_m,L,U\) defined in \eqref{eq:old-L3}--\eqref{eq:old-L5} and the preceding definition of \(J\), now with the parameters \eqref{eq:old-A1}. The prefix argument of Theorem~\ref{thm:lower} applies because \(\ell<1/R\): the algorithm must retain \(J\) and permanently reject every \(H_k\). Explicitly, a first abandonment at \(H_k\) would leave value at most \(\ell\), while stopping at that prefix leaves the offline request \(J\) of weight one, contradicting \eqref{eq:old-L1}.

The geometry is unchanged: the positive requests are mutually compatible, each conflicts with the sole negative request \(J\), and the positive total weights in the two complete inputs are \[
\mathrm{OPT}(\sigma^-)=s+t=\kappa-1>1,
\qquad
\mathrm{OPT}(\sigma^+)=s+t+1=\kappa>1.
\] They are therefore the unique optima, so both inputs have perfect predictions.

If \(A\) accepts \(L\), choose \(\sigma^-\). Its only possible nonempty final set is \(\{L\}\), and \eqref{eq:old-L1} excludes an empty final set. Hence its value is exactly \(t\), giving \[
\frac{\mathrm{OPT}(\sigma^-)}{\mathrm{ALG}_A(\sigma^-)}
=\frac{\kappa-1}{(\kappa-1)/\kappa}=\kappa.
\] If \(A\) rejects \(L\), choose \(\sigma^+\). Its only possible nonempty final sets are \(\{J\}\) and \(\{U\}\), each of weight one. Again \eqref{eq:old-L1} excludes the empty set, so the ratio is exactly \(\kappa\). As in Theorem~\ref{thm:lower}, determinism and unknown termination allow the appropriate complete input to be fixed for \(A\), with no change to its labels.

The two branches use \(m+2\) and \(m+3\) requests, respectively. If \(\kappa\) is rational, then \(t,s,\delta\), and \(\ell=s/m\) are rational, so every endpoint in the construction is rational. Finally, for any \(\varepsilon>0\), choose a rational \[
\kappa\in\bigl(\max\{2,d-\varepsilon\},d\bigr).
\] The construction gives the asserted rational witness with the same request bound.

\end{proof}

\paragraph{Request count and the limiting boundary.}
\label{remark-request-count-and-endpoint-precision}
The request bound $\lceil R\rceil+3$ is uniform in $\kappa$ and
$\varepsilon$ for fixed $R$, but not over all algorithms. Rational
endpoints give finite descriptions without a uniform encoding-length
bound as $\kappa\to d$.
\label{remark-why-the-gap-must-remain-positive}
The gap $\delta>0$ is essential: at $\delta=0$, $L=[-t,0)$ and
$U=[1,2)$ are compatible with $J=[0,1)$. The two branch bounds
$\mathrm{ALG}\le t$ and $\mathrm{ALG}\le1$ no longer follow, and replacing
the $H_k$ by $J$ gives a second optimum of equal weight.
The witnesses therefore keep $\delta>0$ and only approach this limit.

\section{An all-algorithm tradeoff by binary reserve amplification}
\label{r7:app:binary-reserve}
\label{r3:app:tradeoff}
\label{r3:app:veto}
\label{r3:app:composed}

Put $\varphi=(1+\sqrt5)/2$, $d=\varphi^2$, and
\begin{equation}
 G(t)=\frac{t}{(\sqrt t-1)^2},\qquad t>1.
 \label{r7:eq:G}
\end{equation}
The following bound applies to every deterministic algorithm in the original model.

\begin{theorem}[Universal consistency--robustness tradeoff]
\label{r7:thm:binary-reserve}
\label{r3:thm:tradeoff}
For every finite $B\ge\varphi^3$, an algorithm with $R(A)\le B$
satisfies
\begin{equation}
 C(A)\ge\max\{d,G(B)\}.
 \label{r3:eq:frontier-lower}
\end{equation}
In particular,
\begin{equation}
 R(A)\le5\quad\Longrightarrow\quad
 C(A)\ge\frac{15+5\sqrt5}{8}>3,
 \label{r3:eq:budget-five}
\end{equation}
and $R(A)\le6$ implies $C(A)\ge(42+12\sqrt6)/25$.

\end{theorem}

\begin{corollary}[Endpoint of the binary-reserve bound]
\label{r3:thm:veto}
An algorithm with $C(A)=d$ has
\begin{equation}
 R(A)\ge\varphi^4=3\varphi+2.
 \label{r3:eq:d-budget}
\end{equation}
\end{corollary}

Theorem~\ref{r10:thm:source-first} strengthens this endpoint to $4\varphi+3$; the continuous bound above remains useful away from $C=d$.

\begin{proof}[Proof of Theorem~\ref{r7:thm:binary-reserve} and its corollary]
We prove the theorem by constructing a finite adaptive input for every
strict pair
\begin{equation}
 \varphi^3\le r<\varphi^4,\qquad d\le c<G(r).
 \label{r7:eq:strict-pair}
\end{equation}
All requests have positive length and all stated conflicts have positive
overlap. Labels are fixed at arrival. At every prefix at which consistency
is invoked, the complete positive set is a feasible unique optimum.

\paragraph{Exchange certificates.}
For a feasible positive set $P$ and a set $S$ of negative requests, let
$P(S)$ be the positive requests intersecting at least one member of $S$.
If
\begin{equation}
 w(P(S))>w(S)
 \quad\text{for every nonempty feasible negative set }S,
 \label{r7:eq:exchange}
\end{equation}
then $P$ is the unique optimum. Indeed, any feasible solution using
exactly the negatives $S$ has weight at most
$w(S)+w(P\setminus P(S))<w(P)$. A proper subset of $P$ also has smaller
weight. In the induction below, the certificate for old negatives uses
old positives and that for newly added negatives uses a disjoint group
of new positives. Thus their weights can be added even when a single
new positive intersects several new negatives.

\paragraph{The two ideal transitions.}
Write
\[
 k=c-2,\qquad A=(r-1)(c-1)-1,
\]
and define
\begin{equation}
 F(q)=\frac{c}{A-rq},\qquad
 T(q)=\frac{1+q}{k+1-q}.
 \label{r7:eq:binary-maps}
\end{equation}
Since $r\ge\varphi^3>4$ and $G$ decreases with $G(4)=4$, we have
$0<k<2$. Also $G(r)<r$ for $r>4$, so $r-c>0$. On $[0,k]$,
\[
 A-rq\ge A-rk=r-c>0,\qquad k+1-q\ge1.
\]
The sign of $F(q)-q$ is the sign of $rq^2-Aq+c$. Its discriminant,
\[
 \Delta=A^2-4rc=\bigl((r-1)c-r\bigr)^2-4rc,
\]
has, as a quadratic in $c$, roots
$r/(\sqrt r+1)^2<1$ and $G(r)$. Hence
\eqref{r7:eq:strict-pair} gives $\Delta<0$ and $F(q)>q$.
Furthermore,
\[
 T(q)-q=\frac{q^2-kq+1}{k+1-q}>0,
 \qquad q^2-kq+1=(q-k/2)^2+1-k^2/4>0.
\]
Consequently
\begin{equation}
 \mu=\min_{0\le q\le k}\bigl(\min\{F(q),T(q)\}-q\bigr)>0.
 \label{r7:eq:uniform-progress}
\end{equation}
Starting at $q=0$, every ideal sequence of the two transitions exceeds
$k$ within $\lfloor k/\mu\rfloor+1$ steps. This gives a uniformly
finite binary response tree.

\paragraph{A strict initial state.}
Choose a sufficiently small $\delta_0>0$. Present negative
$J=[0,1)$, which robustness forces the algorithm to accept. Present
positive pieces partitioning $[\delta_0,1+2\delta_0)$: prescribe
the final piece $[1-\delta_0,1+2\delta_0)$ and partition
$[\delta_0,1-\delta_0)$ into pieces shorter than $1/r$.
Take $\delta_0<1/2$ and $3\delta_0<1/r$. Every piece strictly
intersects $J$. At the first departure from $J$, only the incoming
piece could be held, contradicting robustness because historical
$\mathrm{OPT}\ge1$. Thus all pieces are missed. Their total weight is
$1+\delta_0>1$, and each conflicts with the sole negative $J$, so they
form a unique optimum. The algorithm holds $J$ alone.

\paragraph{The inductive state and terminal tail.}
At a live state the algorithm holds a sole negative anchor
$Q=[b,b+x)$. Every positive has been missed; their feasible set $P$
is the unique optimum and has weight
\[
 p=x+w,\qquad q=w/x>0.
\]
The leftmost positive endpoint is $b+\delta$, where $0<\delta<x$,
and every positive starts at or to the right of that point. Initially
$(x,w,\delta)=(1,\delta_0,\delta_0)$.

If $q>k$, append a positive tail of length $x$ ending at $b+\delta$.
It is compatible with $P$ and overlaps $Q$ by $\delta$. The old
exchange certificates persist after adding a compatible positive,
so the full positive set remains uniquely optimal. Either nonempty
response has gain $x$, but its optimum has weight
$p+x=2x+w>cx$. The empty response violates robustness. This terminates
the branch.

\paragraph{A reserve and the accepting branch.}
Suppose $q\le k$. Set
\begin{equation}
 z=(c-1)x-w,\qquad \ell=z+\delta,
 \qquad Z=[b+\delta/4-z,b+\delta/4).
 \label{r7:eq:reserve}
\end{equation}
Then $z\ge x>0$. Present $Z$ as a negative request. It overlaps $Q$
by $\delta/4$ and is separated from $P$ by a gap of $3\delta/4$.

If $Z$ is accepted, put $\zeta=\delta/16$ and present a positive
partition of
\[
 [\operatorname{left}(Z)+\zeta,
   \operatorname{right}(Z)+2\zeta).
\]
Prescribe the final crossing piece
$[\operatorname{right}(Z)-\zeta,
\operatorname{right}(Z)+2\zeta)$ and partition the preceding
interval into pieces shorter than $x/r$. We choose the initial
perturbation so that $3\zeta<x/r$ and $z>2\zeta$ at every such state.
Every piece strictly intersects $Z$. Historical $Q$ has length $x$,
so robustness forces every piece to be missed. All new positives end
at or before $b+3\delta/8<b+\delta$, hence are compatible with $P$.

The new positive group has weight $z+\zeta>z$ and every member
intersects $Z$. For any feasible subset of historical negatives,
use its old conflict certificate inside $P$; if it also contains $Z$,
add this new group. The groups contain distinct positive requests,
so \eqref{r7:eq:exchange} holds for every subset and the augmented
positive set is uniquely optimal. The new state has
\begin{equation}
 x'=z,\qquad w'=x+w+\delta/16,\qquad
 \delta'=\delta/16.
 \label{r7:eq:accept-update}
\end{equation}
Its density tends to $T(q)$ as $\delta/x\to0$. The small right-crossing
piece repairs the reserve while leaving a positive free margin at its
left endpoint.

\paragraph{The rejecting branch and a forced reset.}
If $Z$ is rejected, present positive
\[
 L=[b-z,b+\delta),\qquad w(L)=\ell.
\]
It contains $Z$ with margins $\delta/4$ and $3\delta/4$, is compatible
with $P$, and overlaps $Q$ by $\delta$. Old $P$ certifies every old
negative subset, while $L$ alone strictly dominates $Z$. Thus $P+L$
is uniquely optimal. Its weight is
\[
 p+\ell=cx+\delta,
\]
so retaining $Q$ violates consistency. A surviving algorithm accepts
$L$.

Present negative $X$, ending at $\operatorname{left}(Z)$, of length
\begin{equation}
 x'=r\ell-p-z+\delta
    =(A-rq)x+(r+1)\delta
    \ge(r-c)x+(r+1)\delta>0.
 \label{r7:eq:reset-length}
\end{equation}
It overlaps $L$ by $\delta/4$ and is compatible with $Z$ and every
member of $P$. The historical feasible set $P+Z+X$ has weight
$r\ell+\delta$, forcing replacement of $L$ by $X$ through robustness.

Fill $[\operatorname{left}(X)+\delta/4,\operatorname{left}(L))$
with positive pieces shorter than $x'/r$. Its length is
$x'-\delta/2>0$ by \eqref{r7:eq:reset-length}. Each piece intersects
$X$, and historical $X$ forces every piece to be missed. The new
positive group, consisting of $L$ and this fill, has weight
\[
 \ell+x'-\delta/2=x'+z+\delta/2.
\]
For the newly added negatives, the full certificates are
\[
\begin{array}{c|c|c}
 \text{negative subset}&\text{positive conflict group}&
     \text{strict surplus}\\ \hline
 \{Z\}&L&\delta\\
 \{X\}&L+\text{fill}&z+\delta/2\\
 \{Z,X\}&L+\text{fill}&\delta/2.
\end{array}
\]
In the joint row the positive group is counted once. Combining a row
with the old certificate inside $P$ establishes
\eqref{r7:eq:exchange} for every feasible historical negative subset.
Hence the full positive set is uniquely optimal. All positives have
been missed, $X$ is the sole held request, and the new state is
\begin{equation}
 x'=Ax-rw+(r+1)\delta,\qquad
 w'=cx+\delta/2,\qquad \delta'=\delta/4.
 \label{r7:eq:reject-update}
\end{equation}
Its limiting density is $F(q)$.

\paragraph{One finite perturbation for the complete tree.}
The ideal binary tree from \eqref{r7:eq:uniform-progress} has finitely
many nodes. All its nonterminal anchor lengths are positive, since
the transition denominators are bounded below by $r-c$ and $1$;
every terminal density strictly exceeds $k$. For each branch, the
initial state and the exact updates
\eqref{r7:eq:accept-update}, \eqref{r7:eq:reject-update} depend
continuously on $\delta_0$ near zero and converge to the corresponding
ideal states.

Choose one sufficiently small $\delta_0>0$ simultaneously for all
nodes. Then all ideal terminal excesses remain positive, all anchor
lengths remain positive, and every required ratio $\delta/x$ is small
enough to ensure $\delta<x$, $z>2\zeta$, and $3\zeta<x/r$, as well
as the initial seed inequalities. The margin is divided by $16$ on
acceptance and by $4$ on rejection. Anchor lengths need not increase:
continuity at the finitely many positive ideal lengths controls every
ratio separately. If a perturbed state crosses $k$ earlier, terminate
there. Otherwise it reaches a terminal level of the finite ideal tree
with its strict terminal excess preserved.

All overlaps and forcing margins that vanish in the ideal limit have
explicit positive finite identities. These are the overlaps
$\delta/4$, $\delta$, $\delta/4$ at the reserve, positive repair,
and reset; the accepting repair surplus $\delta/16$; the two forcing
surpluses $\delta$; and the joint reset certificate surplus
$\delta/2$. The seed has positive surplus $\delta_0$ and a crossing
piece of length $3\delta_0$.

\paragraph{Causality, all responses, and termination.}
Each partition has finitely many pieces. Along a surviving history
the held set is a singleton and the incoming request intersects it.
The legal responses are therefore to retain that singleton, replace
it by the incoming request, or hold nothing. Empty responses violate
robustness. Both nonempty reserve responses are continued explicitly.
At all other steps the alternative response has a strict violation
at the current prefix, where the adversary may stop. Voluntary
revocation creates no additional available request, and rejected or
revoked requests cannot be recalled. Thus the construction covers
every deterministic decision and uses no assumption about the
termination time. All consistency invocations have the complete
positive optimum certified above. The finite response tree proves
the contradiction for every strict pair
\eqref{r7:eq:strict-pair}.

For $\varphi^3\le B<\varphi^4$ and every $d\le c<G(B)$, the construction gives $C(A)>c$. Therefore $C(A)\ge G(B)$ by taking the supremum over those $c$. Since $G$ is decreasing and
$G(\varphi^4)=d$, the established universal lower bound $C(A)\ge d$
proves Theorem~\ref{r7:thm:binary-reserve} for all stated budgets.
The corollary follows at the crossing. More generally,
$G(G(t))=t$, so for $d\le c<G(\varphi^3)$ the theorem equivalently
requires $R(A)\ge G(c)$ whenever $C(A)\le c$.
\end{proof}

\paragraph{Scope of the lower bound.}
The proof provides no matching algorithm. At $c=d$ and $r=\varphi^4$,
the reset map $F$ has a double fixed point at $q=k$; its ideal orbit
from zero approaches $k$ from below. For larger $r$ its smaller fixed
point is below $k$. Always rejecting the reserve therefore prevents
this ideal construction from crossing $k$. This explains the limit of the present mechanism.

\begin{corollary}[An explicit finite request bound]\label{r12:cor:binary-size}
For a strict pair $(r,c)$ in \eqref{r7:eq:strict-pair}, put
\[
 k=c-2,\quad A=(r-1)(c-1)-1,\quad \Delta=A^2-4rc,
 \qquad
 m_0=\min\left\{\frac{-\Delta}{4rA},
                    \frac{1-k^2/4}{k+1}\right\}>0.
\]
Set $N=\lfloor k/m_0\rfloor+1$ and
$Q=\max\{\lfloor r(c-1)\rfloor+3,\lfloor r\rfloor+4\}$.
The binary-reserve adversary can use at most
\[
 \lfloor r\rfloor+4+NQ
\]
requests on every branch. The bound depends on the strict target
ratios, not on the algorithm or the choice of sufficiently small
perturbation and allowed subdivisions.
\end{corollary}
\begin{proof}
On $[0,k]$, completing the square gives
$rq^2-Aq+c\ge-\Delta/(4r)$, while $0<A-rq\le A$.
Thus $F(q)-q\ge-\Delta/(4rA)$.
Likewise $q^2-kq+1\ge1-k^2/4$ and
$0<k+1-q\le k+1$, so $T(q)-q\ge(1-k^2/4)/(k+1)$.
Every ideal branch therefore terminates in at most $N$ transitions.
The common finite-tree perturbation in the proof preserves this
depth; a perturbed branch reaching $q>k$ sooner terminates sooner.

Partition a segment of length $L$ into $\lfloor L/u\rfloor+1$
equal pieces whenever every piece must be strictly shorter than $u$.
The seed uses the anchor, at most $\lfloor r\rfloor+1$ interior
pieces, and one prescribed crossing piece. On an accepting-reserve
transition, $z/x=c-1-q\le c-1$: the reserve, at most
$\lfloor r(c-1)\rfloor+1$ interior pieces, and the crossing
piece use at most $\lfloor r(c-1)\rfloor+3$ requests.
On a rejecting-reserve transition, the reserve, positive repair and
negative reset are three requests. The fill has length less than
$x'$, with pieces shorter than $x'/r$, so at most
$\lfloor r\rfloor+1$ pieces suffice. Finally one terminal tail
is enough. Adding these counts gives the displayed bound. Early
violations or voluntary deletion only shorten a branch.
\end{proof}

\section{Source-first budget bounds with two reserve lengths}
\label{r11:app:continuous-source-first}\label{r10:app:source-first}

Put $c=d+\gamma$, $k=d+\gamma r$, and $\beta=4\varphi+3$.
We use two constructions, distinguished by the length of a rejected right
reserve. Their limiting packing values are
\begin{align}
 F_{\rm s}(\gamma,r)&=2V_{\rm s}
 +(1+V_{\rm s})\left(\frac1k+\frac1{k-2}\right),
 &V_{\rm s}&=\frac1{c-2},\notag\\
 F_{\rm l}(\gamma,r)&=V_{\rm l}\left(1+\frac1{c-2}\right)
 +(1+V_{\rm l})\left(\frac1k+\frac1{k-2}\right),
 &V_{\rm l}&=\frac1{k-2}.
 \label{r11:eq:implicit-map}
\end{align}
The short-reserve construction allows $0\le\gamma\le1/20$;
the long-reserve construction allows $0\le\gamma\le1/100$.
Let $D_{\rm s}(\gamma)$ and $D_{\rm l}(\gamma)$ solve
\begin{equation}
 D_j(\gamma)=F_j(\gamma,D_j(\gamma)),\qquad j\in\{{\rm s},{\rm l}\}.
 \label{r11:eq:continuous-budget}
\end{equation}

\begin{theorem}[Necessary robustness budgets]
\label{r11:thm:continuous-source-first}
The roots above are unique in $[6,2\varphi^3]$ and $[8,\beta]$,
respectively. Each is continuous and strictly decreasing on its stated
domain, with $D_{\rm s}(0)=2\varphi^3$ and $D_{\rm l}(0)=\beta$.
In the model of Section~\ref{sec:model}, every deterministic algorithm
with finite robustness satisfies
\[
 C(A)\le d+\gamma\quad\Longrightarrow\quad R(A)\ge D_j(\gamma)
\]
on the corresponding domain. Every perfect terminal used below has a
unique full-history optimum equal to the complete positive family.
\end{theorem}

For fixed $\gamma$, each $F_j$ is nonincreasing in $r$, so $F_j-r$
is strictly decreasing. Each $F_j$ strictly decreases with $\gamma$
at fixed positive $r$. At the upper root endpoints its value is at most
that endpoint, with equality when $\gamma=0$. For the short construction,
$13/5<d<21/8$ gives
\[
 F_{\rm s}(\gamma,6)>
 \frac{80}{27}+\frac{67}{27}
 \left(\frac{40}{117}+\frac{40}{37}\right)>6.
\]
For the long construction, $d<131/50$ gives $k<27/10$, $c<263/100$
at $r=8$, and hence
\[
 F_{\rm l}(\gamma,8)>
 \frac{10}{7}\left(1+\frac{100}{63}\right)
 +\frac{17}{7}\left(\frac{10}{27}+\frac{10}{7}\right)
 =\frac{10670}{1323}>8.
\]
Continuity proves existence and uniqueness. The strict parameter decrease
and uniqueness on a compact interval prove the root properties.

\subsection{A finite completion}

For a feasible positive family $\mathcal P$, let
$\Gamma_{\mathcal P}(\mathcal S)$ be the positive identities conflicting
with a negative family $\mathcal S$. The condition
\begin{equation}
 w(\Gamma_{\mathcal P}(\mathcal S))>w(\mathcal S)
 \quad\text{for every nonempty feasible negative family }\mathcal S
 \label{r11:eq:strict-exchange}
\end{equation}
makes $\mathcal P$ the unique full-history optimum: a packing using
$\mathcal S$ has weight at most
$w(\mathcal P)-w(\Gamma_{\mathcal P}(\mathcal S))+w(\mathcal S)$,
and a proper subset of $\mathcal P$ also has smaller weight.

\begin{lemma}[Completion with finite consistency slack]
\label{r11:lem:quantitative-completion}
Let $c=d+\gamma$, $\gamma\ge0$, and $r>1$. Suppose the complete positive
family $\mathcal P$ is feasible, ends at or before $p$, and has weight
$P_0$. The negatives are $\mathcal N$ and an unavailable
$R_0=[p,p+z)$, $z>0$. Suppose~\eqref{r11:eq:strict-exchange} holds for
every nonempty feasible subset of $\mathcal N$. The held weight is at
most $a>0$, and a historical feasible family has weight at least $a$.
If finite $h,\zeta>0$ satisfy
\begin{align}
 h&>(r-1)a,&h&>2\zeta,&\alpha h&>z+2\zeta,\notag\\
 p+\alpha h-\zeta&>\text{every old right endpoint},&&&
 P_0-ca-\gamma h-2\zeta&>0,
 \label{r11:eq:completion-conditions}
\end{align}
there is a finite continuation causing an $r$-robustness violation or a
$c$-consistency violation at a unique-perfect terminal.
\end{lemma}
\begin{proof}
Set $b=p+\alpha h-\zeta$ and offer negative $H=[b,b+h)$, disjoint
from all old requests. Omitting it leaves gain at most
$a<(a+h)/r$, so it must be retained. Offer a finite positive partition
$\mathcal F$ of $[b+\zeta,b+h-\zeta)$, each tile shorter than
$(a+h)/r-a$. A first departure from $H$ leaves at most $a$ plus the
incoming tile, so robustness forces every tile to be missed.

Offer positive $U=[p,p+\alpha h)$. If retained, stop. The gain is at
most $a+\alpha h$, while the complete positives weigh
$P_0+\varphi h-2\zeta$. Their consistency excess is
$P_0-ca-\alpha\gamma h-2\zeta>0$. If $U$ is rejected and $H$
discarded, stop for a robustness violation. Otherwise offer positive
$V_1=[b+h-\zeta,b+2h-\zeta)$ and stop. Every response has gain at
most $a+h$, whereas the positives weigh $P_0+dh-2\zeta$, with excess
$P_0-ca-\gamma h-2\zeta>0$.

The old positive group strictly pays every feasible subset of
$\mathcal N$. The disjoint new group $\{U\}\cup\mathcal F$ strictly
pays $R_0$, since $\alpha h>z$; it pays $H$ or $\{R_0,H\}$, since
its weight $\varphi h-2\zeta$ exceeds $h+z$. This covers every mixed
negative family without double counting. Adding $V_1$ preserves the
certificate. All partitions are finite, no old request is recalled, and
voluntary deletion only reduces the gain bounds.
\end{proof}

\subsection{One construction with short or long reserves}
Figure~\ref{fig:source-first-geometry} shows the surviving long-reserve branch
and the compatible packing used at its terminal state.

\begin{proof}[Proof of Theorem~\ref{r11:thm:continuous-source-first}]
Choose a construction $j$ in its stated domain and suppose
$C(A)\le c$ and $R(A)<D_j(\gamma)$. Choose $r>R(A)$ strictly below
the root, with $r\ge6$ for $j={\rm s}$ or $r\ge8$ for $j={\rm l}$.
Then $F_j(\gamma,r)>r$. Write $V=V_j$, $P=1+V$, and set
\begin{align}
 \Xi_{\rm s}&=\frac{10}{k-2}+\frac6k+12,&
 \Xi_{\rm l}&=\frac{10}{k-2}+\frac6k+\frac4{c-2}+17,\notag\\
 0<\varepsilon&<\min\left\{\frac1{10000},
 \frac{F_j(\gamma,r)-r}{\Xi_j}\right\}.&&
 \label{r31:eq:finite-slack}
\end{align}
Use common parameters
\begin{align}
 v&=V-4\varepsilon,&p&=1+v-2\varepsilon=P-6\varepsilon,\notag\\
 t&=p/k-\varepsilon,&\ell&=(p-4\varepsilon)/(k-2),&s&=\varepsilon-\ell,
 \label{r11:eq:finite-parameters}
\end{align}
and reserve length
\[
 z=\begin{cases}v-\varepsilon,&j={\rm s},\\
 v/(c-2)-10\varepsilon,&j={\rm l}.
 \end{cases}
\]
Both cases satisfy $v>1$, $p>2$, $4/5<t<1-2\varepsilon$,
$\ell>2$, and $\varepsilon<v/(4r)$. In the short case $z<v<\varphi$;
in the long case $v<z<d$. To check the uniform bounds, in the short case
$V>40/27$, $p>67/27-6/10000$ and $k<61/20$ imply $t>4/5$;
in the long case $V>25/18$, $p>2.38$ and $k<68/25$ do so.
In both cases $p\le d-6\varepsilon$, so
$t\le1-(6/k+1)\varepsilon$ and
$p-v-t\ge(6/k-1)\varepsilon>0$.

\paragraph{Missed sources and a rejected right reserve.}
Offer negative $K=[1-\varepsilon,1-\varepsilon+v)$ and a positive
partition $\mathcal A_0$ of $[1,p)$, each tile shorter than $v/r$.
Finite robustness first forces $K$ to be retained; at a first departure
from it, only the incoming tile could remain. Historical $K$ would then
give a ratio exceeding $r$. Thus every tile is missed and $K$ remains alone.

Offer negative $R_0=[p,p+z)$, overlapping $K$ by $\varepsilon$.
If retained, offer positive $S=[1-\varepsilon,1)$, then a positive
partition $\mathcal F$ of $[p,p+z-\varepsilon)$ whose first tile is
$[p,p+2\varepsilon)$ and all tiles are shorter than $v/(2r)$.
Discarding $R_0$ on $S$ leaves at most $\varepsilon<v/r$; a first
departure on a tile leaves at most $\varepsilon$ plus that tile, also
below $v/r$. Hence $R_0$ is retained throughout. Offer positive
$T_R=[p+z-\varepsilon,p+2z-\varepsilon)$ and stop.
Every response has gain at most $z+\varepsilon$, and the complete
positive weight is $v+2z-2\varepsilon$. Its excess over $c$ times the
gain bound is
\[
 \begin{cases}
 (3-c)v-4\varepsilon>0,&j={\rm s},\\
 (9c-22)\varepsilon>0,&j={\rm l}.
 \end{cases}
\]
The positives are compatible. The group $S,\mathcal A_0$ and the first
new tile conflicts with $K$ and weighs $v+\varepsilon>v$.
The group $\mathcal F,T_R$ conflicts with $R_0$ and weighs
$2z-\varepsilon>z$. Since $K,R_0$ conflict, these checks cover every
feasible negative family. The terminal is unique-perfect, contradicting
consistency. Consequently $R_0$ is rejected and $K$ retained. The
hypothetical positive continuation is absent from this rejection history.

\paragraph{Forcing a positive switch.}
Offer positive $J=[0,1)$ and write
$\mathcal A=\{J\}\cup\mathcal A_0$, of weight $p$.
Suppose $J$ is rejected. Offer positive $T=[p,p+v)$ and put $E=p+v$.
The only possible local holders are $K$ and $T$, which conflict and
each weigh $v$. For the short reserve, stop: $T$ strictly pays $R_0$
because $v>z$, and $\mathcal A$ pays $K$. The positive set tiles
$[0,E)$, is uniquely optimal, and has consistency excess
\[
 E-cv=\bigl(4(c-2)-2\bigr)\varepsilon>0.
\]

For the long reserve, $T$ alone does not pay $R_0$. Put
$\chi=2\alpha-1>0$, $\zeta=\chi\varepsilon/4$ and
$h=(r-1)v+\varepsilon$. With $b=E+\alpha h-\zeta$, offer negative
$H=[b,b+h)$, a positive partition $\mathcal F_H$ of
$[b+\zeta,b+h-\zeta)$ with tiles shorter than $(v+h)/r-v$, and then
positive $U=[E,E+\alpha h)$. Here $\alpha h>7\alpha>d+2\zeta>z+2\zeta$,
so $H$ is beyond all old requests. Historical $\{K,H\}$ forces $H$
to be retained while its tiles are missed. If $U$ is retained, stop;
otherwise $H$ must remain, and we offer
$V_1=[b+h-\zeta,b+2h-\zeta)$ and stop. Gains are at most
$v+\alpha h$ and $v+h$, respectively. The second consistency excess,
which is no larger than the first, is
\[
 E-cv-\gamma h-2\zeta
 =\bigl(4(k-2)-2-\gamma-\chi/2\bigr)\varepsilon>0,
\]
since its coefficient is at least $3\chi/2-1/100>0$.

For full-history optimality, the negative families are subsets of
$\{K,H\}$ or $\{R_0,H\}$. The following payment groups have positive
surplus:
\begin{center}
\begin{tabular}{lll}\toprule
Negative family & Conflicting positive group & Surplus\\\midrule
$K$ & $\mathcal A$ & $p-v$\\
$R_0$ & $T,U$ & $v+\alpha h-z$\\
$H$ & $U,\mathcal F_H$ & $\alpha h-2\zeta$\\
$R_0,H$ & $T,U,\mathcal F_H$ & $v+\alpha h-z-2\zeta$\\\bottomrule
\end{tabular}
\end{center}
The groups for $K$ and $H$ are disjoint, covering the remaining joint
family. The shared $U$ is counted only once for $R_0,H$. Adding $V_1$
preserves all certificates. Thus both terminals are unique-perfect.
In either reserve construction, $J$ must be accepted and retained alone.

\paragraph{Forcing rejection of the interior and left negatives.}
Offer negative $I=[\varepsilon,\varepsilon+t)$, disjoint from $K,R_0$
and conflicting with $J$. If accepted, apply
Lemma~\ref{r11:lem:quantitative-completion} with positives $\mathcal A$,
held bound $a=t$, local negatives $\{K,I\}$, boundary reserve $R_0$
and historical witness $I$. The positive $J$ strictly pays $I$ because
$t<1$, and $\mathcal A$ pays $K$ and $\{K,I\}$ because $p-v-t>0$.
Take $h=(r-1)t+\varepsilon$ and $\zeta=\varepsilon/4$.
In the short case $\alpha h>(3/5)5(4/5)=12/5>z+2\zeta$;
in the long case $\alpha h>(3/5)7(4/5)=84/25>z+2\zeta$.
The anchor therefore lies beyond all old endpoints, and the deficit is
\[
 p-ct-\gamma h-2\zeta=(k-\gamma-1/2)\varepsilon>0.
\]
The lemma rules out acceptance of $I$.

Now offer negative $L=[s,\varepsilon)$, touching $I$, disjoint from
$K,R_0$, and conflicting with $J$. If accepted, set $\eta=\varepsilon/4$,
offer a positive partition of $[s+\eta,0)$ with tiles shorter than
$\ell/r$, and then positive $B=[s+\eta-\ell,s+\eta)$.
Historical $L$ forces the tiles to be missed. Their union with $B$ is a
feasible positive family $\mathcal Z$ of weight
$2\ell-\varepsilon-\eta>\ell$, all conflicting with $L$.
Every response to $B$ leaves gain at most $\ell$.

Apply the lemma with positives $\mathcal A\cup\mathcal Z$, $a=\ell$,
local negatives $\{K,I,L\}$, boundary reserve $R_0$ and witness $L$.
For $\{I\}$, the conflicting group is $\Gamma(I)=\{J\}$ and
$w(J)=1>t$. For $\{K\}$ and $\{K,I\}$, $\mathcal A$ pays with
surplus at least $p-v-t>0$. The disjoint group $\mathcal Z$ pays $L$.
Combining these certificates covers $\{I,L\}$ and every other nonempty
subset of the three compatible local negatives without double counting. Take
\begin{equation}
 h=(r-1)\ell+\varepsilon,\qquad \zeta=\varepsilon/4.
 \label{r11:eq:finite-completion-scale}
\end{equation}
Since $\ell>2$, the geometric requirements are stronger than those just
checked for $I$. The deficit is
\[
 p+2\ell-\varepsilon-\eta-c\ell-\gamma h-2\zeta
 =(9/4-\gamma)\varepsilon>0.
\]
The lemma rules out acceptance of $L$, including the joint $R_0,H$
payment without reusing a boundary positive.

\paragraph{The actual robustness terminal.}
Stop after $L$ is rejected. The algorithm holds only $J$, of value one.
The actual requests $L,I,\mathcal A_0,R_0$ form a compatible packing of
weight
\[
 \ell+t+(p-1)+z=F_j(\gamma,r)-\Xi_j\varepsilon>r.
\]
This packing is uniquely optimal. Without $K$, choosing $L,I$ beats
$J$ and permits every member of $\mathcal A_0$ and $R_0$; a packing
containing $K$ has weight at most $\ell+t+v$, which is smaller because
$(p-1)+z>v$. This violates robustness. All partitions are finite; each
branch retains its original labels and complete negative history.
Voluntary deletion either reduces the displayed gain or violates robustness
at an earlier prefix. Determinism fixes which finite branch is taken.
\end{proof}

\begin{figure}[t]
\centering\includegraphics[width=\textwidth]{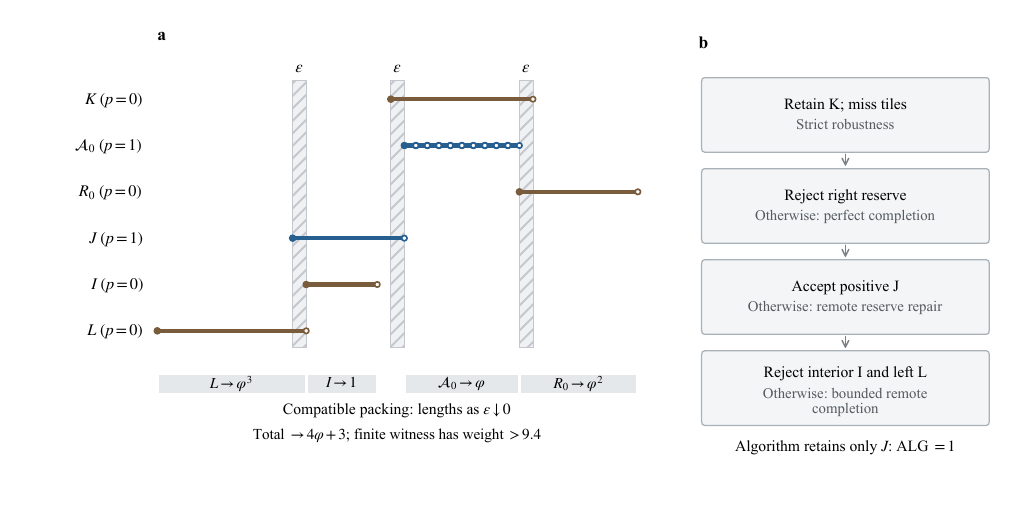}
\caption{Surviving branch with a long right reserve, at $\gamma=0$,
$r=9.4$ and $\varepsilon=1/20000$. Endpoint gaps are exaggerated:
the schematic is not to scale, and hatching marks positive-length overlaps.
Positive tiles arrive before the reserve. The forced responses retain $K$,
reject $R_0$, switch to $J$, and reject $I,L$. The actual packing
$L,I,\mathcal A_0,R_0$ approaches $4\varphi+3$, while the algorithm
retains value one. If $J$ is rejected, the additional positive $T$ and
remote completion pay the entire longer reserve; this branch is separate
from the displayed history. Packing labels give the limiting lengths
as $\varepsilon\downarrow0$, not the finite instance's exact lengths.}
\label{fig:source-first-geometry}
\end{figure}

\subsection{Consequences for the consistency infimum}

For $j\in\{{\rm s},{\rm l}\}$ put $\gamma_{\rm s}=1/20$,
$\gamma_{\rm l}=1/100$, and define
\[
 H_j(B)=\begin{cases}
 d+\gamma_j,&B\le D_j(\gamma_j),\\
 d+D_j^{-1}(B),&D_j(\gamma_j)<B<D_j(0),\\
 d,&B\ge D_j(0).
 \end{cases}
\]
Taking parameters approaching equality from below in the theorem gives
\begin{equation}
 C^*(B)\ge\max\{d,G(B),H(B)\},\qquad H(B)=\max\{H_{\rm s}(B),H_{\rm l}(B)\}.
 \label{r11:eq:strict-infimum}
\end{equation}
The $G$ term is used on its domain $B\ge M$. Since $D_{\rm l}$ is
continuous at zero, every $B<\beta$ has $C^*(B)>d$.

An explicit weaker version follows by fixing $r=\beta$ and writing
\[
 F_{\rm l}(\gamma,\beta)=2V+V^2+V/q+1/k+V/k,
 \quad V=(d+\gamma\beta-2)^{-1},\quad q=d+\gamma-2.
\]
Every term is a positive product of reciprocals of increasing affine
functions. Its negative derivative is nonincreasing, and at zero their
sum is
\[
 -\partial_\gamma F_{\rm l}(0,\beta)
 =\beta(7\varphi+8)+\varphi^3=83\varphi+53.
\]
As $D_{\rm l}(\gamma)\le\beta$,
$D_{\rm l}(\gamma)\ge F_{\rm l}(\gamma,\beta)
\ge\beta-(83\varphi+53)\gamma$. Thus, by taking smaller parameters
and then their supremum,
\[
 C^*(B)\ge d+\min\left\{\frac1{100},
 \frac{\beta-B}{83\varphi+53}\right\}>d\qquad(B<\beta).
\]
The short-reserve curve remains useful farther from $d$ (Section~\ref{sec:lower}). Its corresponding explicit bound
$D_{\rm s}(\gamma)\ge2\varphi^3-(32\varphi+20)\gamma$ follows by the
same derivative argument, and is retained by~\eqref{r11:eq:strict-infimum}.

At $\gamma=0$ the long construction permits $v\to\varphi$,
$z\to\varphi^2$, $t\to1$ and $\ell\to\varphi^3$.
The larger reserve is possible because a separate completion pays it
when $J$ is rejected. These four parameter inequalities give the limiting value $\beta$. Further
improvement requires additional compatible missed mass on the continuing
history or different forcing decisions. A larger counterfactual completion
cannot itself supply that mass.

\section{Source measures and shaped replacement costs}\label{app:shaped}

Write $\alpha=1/\varphi$, $d=\varphi^2$, and $M=\varphi^3$.
We prove a common capacity theorem for $a\in[\alpha,1]$;
$a=\alpha$ recovers the interface used by Shaped-BANK, SUM-BANK,
and Credit-BANK. Set
\begin{equation}
 \begin{gathered}
 \beta_a=a^{-1},\quad C_a=(1+a)^2,\quad q_{0,a}=1+\beta_a,\quad Q_a=1+a+\beta_a,\\
 m_a=\max\{(1+a)/a^2,Q_a\},\qquad k_a=am_a=\max\{1+\beta_a,1+a+a^2\},\\
 \theta_a=\frac{a}{1+a},\qquad
 h_{0,a}=\frac{m_a+1}{m_a-\beta_a},\qquad h_{1,a}=1+\beta_a h_{0,a}.
 \label{curve:eq:constants}
 \end{gathered}
\end{equation}
These constants are fixed for an entire execution. The ordinary gate
adds the actual costs of all current conflicts and has priority.
The special gate accepts a positive arrival only if every conflict
is negative and $w\ge aW$, $w\ge ab_D$.
The simultaneous fees are those in~\eqref{eq:fees}, with $\alpha$
replaced by $a$ and the ledger-test condition replaced by $w<ab_D$.
Full inheritance and equality acceptance are unchanged.

\subsection{A common consistency interface}\label{app:consistency-interface}

\begin{theorem}[Parametric feasible-positive capacity]
\label{curve:thm:interface}
Let the fixed finite nonnegative ordinary costs satisfy
\begin{align}
 q_+(v,b)&\ge m_av &&(v>0,\ b\ge0),\notag\\
 q_-(v,b)&\ge q_{0,a}v &&(v>0,\ b\ge0),\label{curve:eq:global}\\
 q_-(v,b)&\ge\max\{q_{0,a}v,v/a+b/(1+a)\}
      &&(v>0,\ 0\le b\le C_av).\label{curve:eq:capacity}
\end{align}
For every finite execution whose complete positive set $P$ is feasible,
each negative held identity satisfies $b(J)\le C_aw(J)$ and each positive
one satisfies $b(J)\le k_aw(J)$ throughout its lifetime. Consequently
$w(P)\le C_a\,\mathrm{ALG}$.
\end{theorem}

The complete set includes future positive requests and previously
rejected or revoked identities. The theorem assumes no monotonicity or upper bound on the costs; the robustness proof imposes these on the capped cost.

\paragraph{Event times and conserved sources.}
Fix one finite execution with arrivals numbered $1,\ldots,N$ and
feasible complete positive set $P$. Write $\mathcal S_n$ for the
held identities immediately after arrival $n$. An accepted identity
$X$ has birth index $\nu_X$ and first revocation index $\delta_X$;
set $\delta_X=N+1$ if it survives. Its live snapshots are precisely
$\nu_X\le n<\delta_X$. A rejection fee at arrival $n$ uses the
old state $\mathcal S_{n-1}$; an accepted parent's inherited ledger
uses this same state, before adding its own positive source.

For an arrived positive source $T$, let $a_{X,n}(T)$ be its conserved
coefficient in the ledger of held identity $X$. Acceptance merges
coefficients and adds a unit coefficient for a positive arrival;
rejection distributes a unit coefficient by its fees. Thus
\[
 a_{X,n}(T)\ge0,\qquad
 \sum_{X\in\mathcal S_n}a_{X,n}(T)=1,\qquad
 b_n(X)=\sum_{T\in P}a_{X,n}(T)w(T).
\]
The unit-sum identity is for sources arriving by $n$; coefficients
of future sources are zero. Define
\begin{equation}
 \mu_{X,n}(E)=\sum_{T\in P}a_{X,n}(T)|T\cap E|.
 \label{sg:eq:measure}
\end{equation}
Compatibility of the sources and coefficients at most one imply that
this measure has density at most one. For $X=[x,x+v)$ put
$E_{L,n}(X)=\mu_{X,n}(( -\infty,x))$ and
$E_{R,n}(X)=\mu_{X,n}([x+v,\infty))$. For a cut $s$, the
\emph{whole-source} mass is
\[
 L_{n,X}(s)=\sum_{T\in P:\,f(T)\le s}a_{X,n}(T)w(T).
\]
It counts complete shares of sources ending by $s$, whereas the
measure of $(-\infty,s)$ may include a partial share of a source
crossing $s$. In particular, the former is at most the latter.
A subscript $0$ below denotes the birth snapshot $\nu_X$; omitted
time subscripts refer to the single live snapshot being considered.
Every left-side assertion also has its reflected right-side form.

\paragraph{Exterior mass and one-event fees.}
Internal mass is at most the holder's length. A fixed endpoint is
crossed by at most one positive source; only a post-birth source
crossing that endpoint can increase exterior mass on that side.
Feasibility excludes positive held conflicts of a positive arrival.
A negative receiver of length $v$ receives $g<av$ when the special
length test fails, or $g\le ab_{\rm pre}$ when the ledger test fails.
The latter gives
\begin{equation}
 g\le\theta_a(b_{\rm pre}+g).\label{sg:eq:fee-after}
\end{equation}
The exterior contribution is bounded both by its fee and by the
length of the source's exterior segment.

\paragraph{Geometry independent of capacity.}
Let $a_c\in(\alpha,1)$ solve $a_c^3+a_c^2=1$. Directly from
\eqref{curve:eq:constants},
\[
 \begin{array}{c|cc}
 &a\le a_c&a\ge a_c\\\hline
 m_a&(1+a)/a^2&Q_a\\
 k_a&1+\beta_a&1+a+a^2\\
 h_{0,a}&1+a+a^2&1+\beta_a\\
 h_{1,a}&2+a+\beta_a&1+\beta_a+\beta_a^2
 \end{array}
\]
The branches agree at $a_c$. Throughout the range,
\begin{equation}
 \begin{gathered}
 1+\beta_a\le k_a\le C_a,\quad m_a\ge Q_a>q_{0,a}\ge2,\quad
 h_{0,a}\ge1+a,\\
 \frac{h_{1,a}}{m_a}=h_{0,a}-1,\quad \frac{h_{0,a}}{q_{0,a}}\le h_{0,a}-1,\quad
 \frac{\beta_a}{m_a}<1,\quad \frac{1+a}{a^2m_a}\le1.
 \label{curve:eq:identities}
 \end{gathered}
\end{equation}
For $k_a\le C_a$, the first branch is equivalent to
$a^3+2a^2-1\ge0$, which vanishes at $\alpha$ and increases
thereafter; the other branch has slack $a$.

Let $H(X)$ be the sum of lengths of all accepted identities in the
replacement subtree rooted at $X$. For arbitrary labels,
\begin{equation}
 H(X)\le h_{0,a}w(X)\quad\text{at an ordinary birth},\qquad
 H(X)\le h_{1,a}w(X)\quad\text{at a special birth}.
 \label{curve:eq:history}
\end{equation}
Indeed each child's history divided by its actual cost is at most
$\max\{h_{0,a}/q_{0,a},h_{1,a}/m_a\}=h_{0,a}-1$. An ordinary parent pays the sum
of child costs, at most its own length; a special parent has only
negative children, of total length at most $\beta_a$ times its own.
These recurrences prove~\eqref{curve:eq:history} from empty births,
without a ledger-capacity assumption. Negative nodes are always
ordinary; an ordinary positive also satisfies the larger $h_{1,a}$ bound.
The accepted union is connected, and a negative subtree protrudes
beyond either root endpoint by at most $(h_{0,a}-1)$ times the root length.

Negative-to-negative downward edges shrink length by at least $q_{0,a}$;
a positive child of a negative node costs at least $m_a$ times its length.
An edge from a positive parent to a negative child expands by at most
$\beta_a$. Thus successive positive nodes on a downward path have
length ratio at most $\beta_a/m_a<1$. Before the first node starting
below a fixed cut, a negative tail from length $v$ displaces by less
than $v/(q_{0,a}-1)=av$. A tail from a positive root of length $u$, with
only negative intermediate nodes, displaces by less than
$\beta_a q_{0,a}u/(q_{0,a}-1)=(1+a)u/a$. Each nonzero source has an accepted
entry support overlapping it. At the golden endpoint these estimates
also give the original accepted-history bounds
\begin{equation}
 H(X)\le2w(X)\ (p(X)=0),\qquad
 H(X)\le Mw(X)\ (p(X)=1),\qquad H(X)\le\tau(X)
 \quad(a=\alpha),\label{sg:eq:history}
\end{equation}
where $\tau(X)$ is the actual ordinary cost; $2\le d$ and the
positive cost lower bound is $M$.

\paragraph{The birth-order induction.}
The induction order is the arrival index of accepted identities.
Each identity's claim covers \emph{all} its live snapshots, including
the last snapshot before a later parent's birth. For a negative $V$
of length $v$, prove capacity $b_n(V)\le C_av$ and
\begin{align}
 E_{L,\nu_V}(V),E_{R,\nu_V}(V)&\le av,
 \label{sg:eq:birth-exterior}\\
 E_{L,n}(V),E_{R,n}(V)&\le\max\{av,\theta_a b_n(V)\},
 \label{sg:eq:lifetime-exterior}\\
 L_{\nu_V,V}(s)+z-s&\le av
 \quad\text{if }s<z\le t,\ z\le s(V),\ L_{\nu_V,V}(s)>0.
 \label{sg:eq:birth-cut}
\end{align}
In the last assertion $T_0=[s,t)$ ranges over actual positive
requests arriving after $\nu_V$. It need not overlap $V$.
For a positive $Q$ of length $u$, prove capacity $b_n(Q)\le k_au$ and
\begin{equation}
 L_{n,Q}(s)+z-s\le k_au
 \quad\text{if }s<z\le t\le s(Q),\ L_{n,Q}(s)>0,
 \label{sg:eq:positive-cut}
\end{equation}
where $T_0=[s,t)$ is an actual positive request arriving after the
live snapshot $n$. Positive ledgers are constant during a lifetime.
All cut assertions include their reflections, counting full shares
of sources starting beyond the right cut.

Earlier capacities and cuts first establish a new negative node's
birth bounds. Its birth bounds and subsequent one-event fees then
establish its entire-lifetime bounds. A positive node uses earlier
negative assertions and the independent geometry. In particular, a
new node's capacity is never used to establish its own birth cuts.

\begin{lemma}[Weighted directional cut at a live snapshot]
\label{sg:lem:weighted}
Let a negative $X=[x,x+v)$ be live at snapshot $n$ and satisfy its
capacity, birth exterior, lifetime exterior, and birth-cut assertions.
For an actual positive $T_0=[s,t)$ arriving after $n$, if
$s<z\le t$, $z<f(X)$, and $L_{n,X}(s)>0$, then
\begin{equation}
 L_{n,X}(s)+z-s\le a q_-(v,b_n(X)).\label{sg:eq:weighted}
\end{equation}
The reflected assertion holds. No later parent's capacity or cut
is an assumption.
\end{lemma}
\begin{proof}
At this snapshot the cost hypothesis gives
\begin{equation}
 v+\max\{av,\theta_a b\}\le a q_-(v,b).
 \label{sg:eq:cost-identity}
\end{equation}
If $x<s$, density gives $L_X(s)\le E_L(X)+s-x$; adding $z-s$
leaves a quantity smaller than $E_L(X)+v$, paid by the lifetime
bound and~\eqref{sg:eq:cost-identity}. If $x\ge s$, a source
ending by $s$ cannot newly overlap $X$ during its lifetime, so
$L_X(s)=L_{0,X}(s)$. Set $z_0=\min\{z,x\}$.
For $z_0>s$ use the birth cut; for $z_0=s$ necessarily $x=s$,
and use the birth exterior bound. Thus
$L_X(s)+z_0-s\le av$. The remaining distance is zero or less
than $v$, giving $(1+a)v=a q_{0,a}v\le a q_-(v,b)$.
\end{proof}

\begin{lemma}[Negative birth bounds]\label{sg:lem:negative-birth}
Assume the full induction assertions for every earlier accepted
identity at all its live snapshots. A new negative identity satisfies
the birth exterior bounds and birth cuts above.
\end{lemma}
\begin{proof}
Let the ordinary parent be $V$ of length $\ell$. Its children have
old-state costs summing to at most $\ell$. A negative child
$X=[x,x+v)$ contributes mass to the left of $s(V)$ at most
$E_L(X)+\max\{s(V)-x,0\}\le E_L(X)+v\le a q_-(v,b(X))$.
A positive child's entire ledger is at most
$k_av=am_av\le a q_+(v,b)$. Sum to obtain the left birth exterior
bound, and reflect for the right.

For the left birth cut, every child with a nonzero left share overlaps
$V$ and ends strictly beyond $z$. A negative child obeys
Lemma~\ref{sg:lem:weighted}. A positive child is compatible with
$T_0$ and overlaps $V$ to its right, forcing its start at least $t$;
its earlier positive cut pays its share and the gap with
$k_aw(X)\le a\tau(X)$. Sum over all children with nonzero shares.
There is at least one, so the gap is paid at least once and
$L_{0,V}(s)+z-s\le a\sum_X\tau(X)\le a\ell$.
Zero-share children need no whole-ledger estimate. Reflection and
empty births finish the proof.
\end{proof}

\begin{lemma}[Negative lifetime and capacity]\label{sg:lem:negative-lifetime}
The negative birth assertions and the subsequent rejection fees imply
the lifetime exterior bounds and $b_n(V)\le C_aw(V)$ at every live
snapshot, without assuming this capacity in advance.
\end{lemma}
\begin{proof}
Write $V=[x,x+\ell)$. A side with positive exterior mass at birth
stays bounded by $a\ell$. A birth source crossing that endpoint
excludes any later positive crosser. Otherwise all birth sources
contributing to that side lie wholly beyond the endpoint. If a later
source $T=[s,t)$ crosses the left endpoint, those sources end by $s$,
so $E_{L,0}(V)=L_{0,V}(s)$. The birth cut bounds
$E_{L,0}(V)+x-s\le a\ell$, and the new exterior contribution is
at most $x-s$. No second boundary source can cross that side.
Reflect for the right.

A side with zero birth exterior receives exterior mass only from its
unique later boundary source. Its fee is bounded by $a\ell$ after
length failure, or by $\theta_a b(V)$ after ledger failure, using
\eqref{sg:eq:fee-after}. Ledger monotonicity preserves this bound
thereafter. This proves the lifetime exterior assertion.

For capacity, two birth-nonzero sides give
$b\le(1+2a)\ell\le C_a\ell$ by density. With exactly one such
side, before a first boundary on the other side
$b_{\rm pre}\le(1+a)\ell$, hence its fee is at most
$a(1+a)\ell$. Subsequently
$b\le\ell+a\ell+a(1+a)\ell=C_a\ell$.
With neither side nonzero at birth, the first boundary has
$b_{\rm pre}\le\ell$ and fee at most $a\ell$. One side is then
protected from any later crosser, and the preceding argument applies
to the other side. If that first source crosses both endpoints, no
later boundary source is possible; with no boundary, density suffices.
These are event-time fee bounds.
\end{proof}

\begin{lemma}[Positive birth and cut]\label{sg:lem:positive-birth}
Assume the full induction assertions for every earlier identity.
A positive accepted $Q$ satisfies capacity $b_n(Q)\le k_aw(Q)$ and
the positive cuts throughout its live snapshots.
\end{lemma}
\begin{proof}
Let $u=w(Q)$ and evaluate all child quantities immediately before its
birth. All children are negative. Ordinary acceptance gives
$W_Q\le u/q_{0,a}$ and $B_Q\le C_aW_Q\le a(1+a)u$; special
acceptance gives $B_Q,W_Q\le\beta_a u$. Thus
$b(Q)=u+B_Q\le k_au$. No later positive request overlaps $Q$,
so its ledger and shares remain constant.

Fix a live snapshot and an actual later positive $T_0=[s,t)$ with
$s<z\le t\le s(Q)$ and $L_Q(s)>0$. Set $e=z-s$ and
$G=s(Q)-s$, so $e\le G$. The source $Q$ is not counted in
$L_Q(s)$, which therefore equals $\sum_XL_X(s)\le B_Q$.
Trace one nonzero share of a source ending by $s$ down the replacement
tree to its accepted entry support. Whether the source was accepted
or rejected there, that support overlaps it and starts strictly below
$s$. Stop at the first such start. Every positive node before this
crossing starts at least at $t$, by compatibility with $T_0$.

\emph{Case 1: a positive node occurs strictly below $Q$ before the
crossing.} Let $H$ be the last one. The independent path estimates
give $w(H)\le\beta_a u/m_a$. Its subsequent tail has only negative
intermediate nodes, hence
\[
 e\le s(H)-s<\frac{1+a}{a}w(H)
 \le\frac{1+a}{a^2m_a}u\le u.
\]
This also permits a positive crossing node, whose cost is at least
$q_{0,a}$ times its length. Thus $L_Q(s)+e\le B_Q+u\le k_au$.

\emph{Case 2: no such positive node occurs, and either $Q$ is ordinary
or more than one child is relevant.} Let $V$ be the first child on
the selected path and $v=w(V)$. The negative-tail displacement gives
$s(V)-s<av$ if $s(V)\ge s$; if $s(V)<s$ the inequality is
immediate. Since $V$ overlaps $Q$, $G<(1+a)v$.
For ordinary $Q$, $v\le u/q_{0,a}$ yields $e<au$, and
\[
 B_Q+e<(a^2+2a)u\le(1+a+a^2)u\le k_au,
\]
where $a\le1$. For special $Q$, if $V$ is not leftmost, a
compatible child to its left ends after $s(Q)$; then
$s(V)>s(Q)$ and $G<av\le u$.
If $V$ is leftmost but another child $K$ has $L_K(s)>0$, that
child starts to the right of $s(Q)$. Its entry support lies below
$s$, so the independent protrusion bound gives
$G<(h_{0,a}-1)w(K)$. Now
\[
 \frac1{1+a}+\frac1{h_{0,a}-1}
 =\begin{cases}\beta_a,&a\le a_c,\\
                 (1+a)^{-1}+a\ge\beta_a,&a\ge a_c,
   \end{cases}
\]
with the last inequality equivalent to $a^3+a^2\ge1$.
Consequently $\beta_a G<v+w(K)\le W_Q\le\beta_a u$.
In either special subcase, $e\le G<u$ and
$B_Q+e<(1+\beta_a)u\le k_au$.

\emph{Case 3: $Q$ is special and only its leftmost child $V$ carries
a left share.} Then $L_Q(s)=L_V(s)$ and
$z\le t\le s(Q)<f(V)$. If $s(V)\ge s$, no source ending by
$s$ can charge $V$ after birth, so $L_V(s)=L_{0,V}(s)$.
Put $z_0=\min\{z,s(V)\}$. For $z_0>s$, use the earlier birth
cut. For $z_0=s$, necessarily $s(V)=s$, and use
$L_{0,V}(s)\le E_{L,0}(V)\le av$ instead. Thus
$L_V(s)+z_0-s\le av$, while $0\le z-z_0<v$, giving
$L_Q(s)+e<(1+a)v\le(1+\beta_a)u\le k_au$.
If $s(V)<s$, density and the earlier lifetime exterior bound give
\[
 \begin{aligned}
 L_Q(s)+e&<v+E_L(V)\\
 &\le\max\{(1+a)v,v+\theta_a b(V)\}\\
 &\le\max\{1+\beta_a,\beta_a+(1+a)^{-1}\}u
 =(1+\beta_a)u\le k_au.
 \end{aligned}
\]
The last inequality uses the two separate special budgets
$v\le W_Q\le\beta_a u$ and $b(V)\le B_Q\le\beta_a u$.
These three cases exhaust the path and child-share alternatives.
Reflection proves the right cut.
\end{proof}

\begin{proof}[Completion of Theorem~\ref{curve:thm:interface}]
A new negative identity obtains its birth assertions from earlier
children, then its lifetime assertions from its own actual fees.
Its weighted cut is used only for later births. A new positive
identity uses earlier negative assertions and independent geometry.
This closes the birth induction. Source conservation and $k_a\le C_a$
give $w(P)=\sum_Xb(X)\le C_a\sum_Xw(X)$.
\end{proof}

\begin{theorem}[Golden endpoint interface]\label{sg:thm:interface}
For the special tests and fees at $a=\alpha$, suppose
$q_+(v,b)\ge Mv$ and
\begin{align}
 q_-(v,b)&\ge dv &&(v>0,\ b\ge0),\label{sg:eq:interface-global}\\
 q_-(v,b)&\ge\max\{dv,\varphi v+\alpha b\}
       &&(v>0,\ 0\le b\le dv).\label{sg:eq:interface-capacity}
\end{align}
For every feasible complete positive set, each held identity satisfies
$b(J)\le dw(J)$ throughout its lifetime and $w(P)\le d\,\mathrm{ALG}$.
\end{theorem}
\begin{proof}
This is the $a=\alpha$ specialization of
Theorem~\ref{curve:thm:interface}: $C_a=k_a=q_{0,a}=d$, $m_a=M$, and
$v/a+b/(1+a)=\varphi v+\alpha b$.
\end{proof}

\subsection{The capped family and exact consistency}

Define $T_a$ by the preceding $a$-special tests and fees, positive
cost $m_av$, and negative cost
\begin{equation}
 q_a(v,b)=\min\{Q_av,\max\{q_{0,a}v,v/a+b/(1+a)\}\}.
 \label{curve:eq:cost}
\end{equation}
Its filtered form is $T_a^F$. The cost is $q_{0,a}v$ up to
$b=(1+a)v$, then linear up to $b=C_av$, and then equal to $Q_av$.
The cap is inactive on the capacity range, so
Theorem~\ref{curve:thm:interface} gives consistency at most $C_a$.

\begin{proposition}[Exact consistency of the family]
\label{curve:prop:exact-consistency}
For every $a\in[\alpha,1]$, both $T_a$ and $T_a^F$ have
$C(T_a)=C(T_a^F)=(1+a)^2$.
\end{proposition}
\begin{proof}
Choose $0<\epsilon<\min\{1/100,a^2/[4(3a+1)]\}$ and offer
\[
 \begin{array}{ccl}
 J&0&[0,1),\\
 H_1&1&[\epsilon,1/2),\\
 H_2&1&[1/2,1-\epsilon),\\
 L&1&[2\epsilon-a,\epsilon),\\
 R&1&[1-\epsilon,1-\epsilon+z),
 \end{array}
 \qquad z=a(1+a-3\epsilon)-\epsilon.
\]
The two seeds and $L$ are shorter than $a$ and fail ordinary
acceptance, since the negative cost is at least $q_{0,a}\ge2$.
They fail the special length test and leave ledger
$b_J=1+a-3\epsilon$. This is below the first breakpoint, so
$q_a(1,b_J)=q_{0,a}$. The final request has $z>a$ but
$z=ab_J-\epsilon$, hence passes the length test and fails the
ledger test. It also fails ordinary acceptance because
\[
 z<a+a^2\le1+1/a=q_{0,a},
\]
where the last inequality is $(a-1)(a+1)^2\le0$ after
multiplication by $a$. Thus $\mathrm{ALG}=1$.
The positives are consecutive in spatial order $L,H_1,H_2,R$.
Their union contains the anchor and every offered request; they
therefore form an optimum of weight
$(1+a)^2-(3a+4)\epsilon$. Any feasible set containing the
anchor contains none of these positives, so this optimum is unique.
Letting $\epsilon$ decrease to zero proves the lower bound.
The capacity theorem proves the upper bound. Filtering leaves
these perfect labels unchanged and inherits the upper bound.
\end{proof}

At $a=\alpha$, the family recovers Shaped-BANK $T_*$ with
\begin{equation}
 q_*(v,b)=\min\{2\varphi v,\max\{dv,\varphi v+\alpha b\}\}
 =\begin{cases}
 dv,&b\le\varphi v,\\
 \varphi v+\alpha b,&\varphi v\le b\le dv,\\
 2\varphi v,&b\ge dv.
 \end{cases}\label{sg:eq:cost}
\end{equation}
\begin{proposition}[Smallest monotone sufficient cost]\label{sg:prop:minimal-main}
Among costs satisfying the interface of Theorem~\ref{sg:thm:main} and nondecreasing in $b$ for fixed $v$, $q_*$ is pointwise smallest.
\end{proposition}
\begin{proof}
For $b\le dv$, the two interface inequalities require $q(v,b)\ge\max\{dv,\varphi v+\alpha b\}$, exactly $q_*$ on this range. At $b=dv$ this value is $2\varphi v$, so monotonicity requires at least $2\varphi v$ for larger $b$. The cost $q_*$ meets all requirements.
\end{proof}

The branches of~\eqref{sg:eq:cost} agree at both boundaries. Minimality is relative to this sufficient interface, which is not necessary (Appendix~\ref{r29:sec:necessary-cost}).

\subsection{Joint-cut robustness of the family}\label{curve:sec:robustness}

The following common transfer lemma is also used in
Appendices~\ref{app:credit} and~\ref{sum:app:robustness}.

\begin{lemma}[Comparator-to-root transfer]\label{sg:lem:charge-transfer}
Consider a finite execution on positive-length intervals, starting
from an empty held set. Acceptance revokes precisely the current
conflicts and retains all nonconflicts;
rejection leaves the held set unchanged, and there are no other
revocations. Give each accepted identity $V$ a weight
$\omega_V\ge w(V)$. Suppose every rejected $I$ has a nonempty
current conflict set $D(I)$ and pieces $z_{I,V}$ satisfying
\[
 0\le z_{I,V}\le\omega_V\quad(V\in D(I)),\qquad
 \sum_{V\in D(I)}z_{I,V}=w(I).
\]

Link each revoked identity to its replacement. In this forest,
let $H_J$ be the closure of the union in $J$'s subtree, including $J$.
An ancestry antichain contains no node and its strict ancestor;
let $\Omega_J$ be its maximum weight in this subtree.
Let $L_J(x)$ and $U_J(y)$ be
the same maximum restricted respectively to starts below $x$
and finishes above $y$, both strictly; empty antichains are allowed.
Set $F_J(x,y)=|H_J\cap[x,y]|+L_J(x)+U_J(y)$.
If each final root satisfies
\[
 \Omega_J\le K w(J),\qquad F_J(x,y)\le K w(J)
 \quad\text{for all real }x\le y,
\]
then every feasible subset $O$ of the input requests satisfies
$w(O)\le K\mathrm{ALG}$.
\end{lemma}
\begin{proof}
An accepted member of $O$ charges its length to its own identity;
a rejected member uses the stipulated pieces. Pass each piece to
its receiver's final root. Its charge there is at most its own
length and is positive only if it intersects that root's hull
in positive length. Each hull is connected along replacement edges.

Fix $H_J=[A,B]$. If one comparator covers its whole interior,
allowing endpoint equality, it is the sole positive-charge
contributor. Its receivers form an ancestry antichain:
they were simultaneously held, or consist of
its own accepted identity. Thus its charge is at most $\Omega_J$.
Otherwise feasibility permits at most one comparator strictly
crossing each hull endpoint. Let $x$ be the left crosser's finish
(or $A$ if absent), and $y$ the right crosser's start (or $B$).
Then $A\le x\le y\le B$. The left crosser's receivers form an
ancestry antichain with starts below $x$, giving charge at most
$L_J(x)$; the right side gives $U_J(y)$. An absent side has zero
maximum at its hull endpoint. All remaining contributors are
contained in $[x,y]$, and their disjoint lengths bound their
total charge by $|H_J\cap[x,y]|$. The root charge is therefore
at most $F_J(x,y)$. Summing over roots preserves every initial
piece, even when their hulls overlap, and proves the claim.
\end{proof}

\paragraph{Cost shape and terminal budgets.}
Let
\begin{equation}
 c_a=\frac{a}{(1+a)^2},\qquad
 A_a=\beta_a+\beta_a^2+(1+a)^{-2},\qquad
 \rho_a=\max\{1,A_a/m_a\}.
 \label{curve:eq:robust-constants}
\end{equation}
The three branches of~\eqref{curve:eq:cost} give
\begin{align}
 q_{0,a}v\le q_a(v,b)&\le Q_av\le m_av,\notag\\
 q_a(v,b)&\le q_{0,a}v+c_ab,\label{curve:eq:affine}\\
 q_a(W,B)-q_a(v,B)&\ge\beta_a(W-v)\quad(0\le v\le W).
 \label{curve:eq:slope}
\end{align}
For the affine majorant the middle-branch slack is $v-b/C_a\ge0$;
the cap branch uses $b\ge C_av$ and $c_aC_a=a=Q_a-q_{0,a}$.
For fixed $B$, the successive slopes in $v$ are $Q_a,\beta_a,q_{0,a}$;
the branches are continuous and $q_a(0,B)=0$. The cost is
nondecreasing in both variables.

Give each accepted identity $X$ its ordinary cost $q_X$ at the
last live snapshot, immediately before revocation or at termination.
These costs are nondecreasing during each lifetime. An ordinary parent of
length $h$ has $\sum_iq_i\le h$. A special parent has only
negative children and
\begin{equation}
 W,B\le\beta_a h,\qquad
 \sum_iq_i\le q_{0,a}W+c_aB\le A_ah.
 \label{curve:eq:terminal}
\end{equation}
Since $m_a\ge q_{0,a}\beta_a$, $c_a\le1/4$, and $q_{0,a}\ge2$,
\[
 \rho_a\le1+c_a/q_{0,a}\le9/8<q_{0,a},\qquad \rho_a-1<\beta_a.
\]
An ancestry antichain has total weight at most $q_J$ in an
ordinary-root subtree and at most $\max\{m_a,A_a\}w(J)=\rho_a m_a w(J)$
in a special-root subtree. Indeed, omitting an ordinary root gives
at most $\rho_a\sum_iq_i\le\rho_a h<q_J$. A special root's negative
children are ordinary and their weights sum to at most $A_ah$;
the singleton costs $m_ah$. This simultaneous induction uses no
feasible-positive capacity. It also shows that an eligible ordinary
root attains its endpoint maximum, while an ineligible side has
weight at most $\rho_a h$.

\begin{theorem}[Parametric robustness certificate]
\label{curve:thm:robustness}
Let $K_a$ be the maximum of the general coefficients in
Table~\ref{sg:tab:coefficients}. For every finite arbitrarily
labelled input, both $T_a$ and $T_a^F$ satisfy
$\mathrm{OPT}\le K_a\,\mathrm{ALG}$. In fact,
\begin{equation}
 \begin{gathered}
 O_a=m_a+1+2\beta_a+\beta_a^2,\qquad
 N_a=h_{1,a}+A_a+\rho_a\beta_a,\\
 K_a=\max\{O_a,N_a\}.
 \label{curve:eq:two-coefficients}
 \end{gathered}
\end{equation}
The map $a\mapsto K_a$ is continuous on $[\alpha,1]$.
\end{theorem}
\begin{proof}
Use Lemma~\ref{sg:lem:charge-transfer} with terminal cost weights
$\omega_X=q_X$. The independent history bounds give hull length
at most $h_{0,a}h$ at an ordinary root and $h_{1,a}h$ at a special root.
For $x\le y$, write
$F_J(x,y)=|H_J\cap[x,y]|+L_J(x)+U_J(y)$ as in that lemma.
The two endpoint maxima need not coexist.

An antichain either contains the root singleton or omits it and
decomposes over child subtrees. Thus
\[
 L_J(x)=\max\left\{q_J\mathbf1_{\{s(J)<x\}},
                         \sum_iL_{C_i}(x)\right\},
\]
with the analogous right recurrence. A zero or ineligible root
term is assigned to the child choice. Ties of positive eligible
values can use either attaining choice. There is no mixed root
and descendant choice, since they are comparable in ancestry.
The descendants of an antichain need not be spatially compatible;
only immediate children are known to have been held together.

At an ordinary root, two eligible endpoints put the gap in the
root and give $1+2m_a$. One eligible endpoint gives
$h_{0,a}+m_a+\rho_a$, and none gives $h_{0,a}+2\rho_a$, using the full hull.

At a special root $J=[l,r)$, the two endpoints choose exactly two,
one, or zero singletons. Two singletons give $1+2m_a$.
For exactly one singleton, reflect to put it on the left, so $l<x$.
If $y<r$, the gap lies in $J$ and the coefficient is $1+m_a+A_a$.
Otherwise $r\le y$. Every child starts strictly before $r$;
compatibility permits at most one child $V$ ending beyond $y$.
If this outward child exists, the gap is contained in $J\cup V$,
of length at most $h+v$. Its right maximum is $q_V$, while
all other child roots are right-ineligible. Hence
\begin{align}
 F_J&\le(m_a+1)h+(1-\rho_a)v+q_a(v,b_V)+\rho_a W\notag\\
 &\le(m_a+1)h+(1-\rho_a)v+q_a(v,B)+\rho_a W\notag\\
 &\le(m_a+1)h+W+q_a(W,B)\notag\\
 &\le(m_a+1+\beta_a+q_{0,a}\beta_a)h.
 \label{curve:eq:outward}
\end{align}
The third line follows from the fixed-ledger slope
$\beta_a>\rho_a-1$; the last uses monotonicity, both special budgets,
and $q_a(\beta_a h,\beta_a h)=q_{0,a}\beta_a h$.
If there is no outward child, all child roots are right-ineligible.
The full hull then gives $h_{1,a}+m_a+\rho_a\beta_a$.

If both endpoint choices use children, call a child $V$ a bridge
when $s(V)<x\le y<f(V)$. There is at most one compatible bridge,
including at $x=y$. Its gap and endpoint weights contribute at
most $v+2q_V$. Every other child has at most one eligible endpoint,
so its endpoint sum is at most $q_i+\rho_a v_i$; this also covers
two ineligible endpoints since $q_i\ge q_{0,a}v_i\ge\rho_a v_i$.
Consequently
\[
 F_J\le v+2q_V+\sum_{i\ne V}(q_i+\rho_a v_i)
 \le W+2\sum_iq_i\le(\beta_a+2A_a)h,
\]
where the middle inequality uses $(\rho_a-1)v_i\le q_i$.
With no bridge, use the full hull and all endpoint sums to obtain
$h_{1,a}+A_a+\rho_a\beta_a$. These cases exhaust the table: root versus
children at each endpoint, the unique outward child in the mixed
case, and the presence or absence of a bridge in the child-only case.
Strict eligibility assigns endpoint equality to ineligibility.

The full antichain bound $\max\{m_a,A_a\}h$ is at most
$(1+m_a+A_a)h\le K_ah$. A rejected request $I$ has a nonempty
conflict set and fails ordinary acceptance, so assign analytical
pieces $w(I)\tau(X)/\sum_D\tau$ to its current conflicts.
They sum to $w(I)$ and are each below its receiver's current cost,
hence below $q_X$. All weights exceed the corresponding interval
lengths. These analytical pieces are not ledger fees.
Every hypothesis of Lemma~\ref{sg:lem:charge-transfer} is now
verified, giving the robustness bound. The filter produces another
causal labelled execution and is covered by the same bound.

It remains to reduce the table to~\eqref{curve:eq:two-coefficients}.
Put $e_a=(1+a)^{-2}$ and $m_{0,a}=\beta_a+\beta_a^2$, so
$A_a=m_{0,a}+e_a$ and $m_a\ge m_{0,a}$.
The branch formulas give $h_{0,a}\le1+\beta_a$; also
$e_a\le\alpha^2<1/2$ and $1\le\rho_a\le9/8<m_a$.
The gap $O_a-(1+2m_a)$ is $\beta_a$ for $a\le a_c$ and
$\beta_a+\beta_a^2-1-a\ge0$ for $a\ge a_c$. Further,
\[
 \begin{aligned}
 O_a-(h_{0,a}+m_a+\rho_a)
   &\ge\beta_a+\beta_a^2-\rho_a\ge7/8,\\
 O_a-(1+m_a+A_a)&=\beta_a-e_a>0,\\
 O_a-(\beta_a+2A_a)&=m_a-m_{0,a}+1-2e_a>0.
 \end{aligned}
\]
The ordinary no-endpoint row is smaller than its one-endpoint row.
The no-outward row $h_{1,a}+m_a+\rho_a\beta_a$ is at most $N_a$
when $A_a\ge m_a$. Otherwise $a>a_c$, since $A_a=m_a+e_a$ on
the first branch; then $h_{1,a}=1+\beta_a+\beta_a^2$ and
$\rho_a=1$, making this row equal to $O_a$. Since both $O_a,N_a$
occur in the table, the reduction follows. Continuity follows from
the defining functions and maxima.
\end{proof}

\begin{table}[t]
\caption{Complete joint-cut case list. The special root is $[l,r)$
and mixed choices are oriented with the singleton on the left.
$K_a$ is the maximum of the general coefficients, with the repeated
$1+2m_a$ counted once. The endpoint column preserves the specialized
Shaped-BANK calculation; the largest entry there is $5\varphi+3$.}
\label{sg:tab:coefficients}
\centering
\footnotesize
\setlength{\tabcolsep}{2.5pt}
\begin{tabular}{lll}
\toprule
Root and endpoint choices & General coefficient & At $a=\alpha$\\
\midrule
Ordinary: two eligible & $1+2m_a$ & $4\varphi+3$\\
Ordinary: one eligible & $h_{0,a}+m_a+\rho_a$ & $7\varphi-4$\\
Ordinary: none eligible & $h_{0,a}+2\rho_a$ & $10\varphi-12$\\
\midrule
Special: two singletons & $1+2m_a$ & $4\varphi+3$\\
Special: one, $y<r$ & $1+m_a+A_a$ & $3\varphi+5$\\
Special: one, outward child & $m_a+1+\beta_a+q_{0,a}\beta_a$ & $5\varphi+3$\\
Special: one, no outward child & $h_{1,a}+m_a+\rho_a\beta_a$ & $2\varphi+7$\\
Special: children, bridge & $\beta_a+2A_a$ & $3\varphi+6$\\
Special: children, no bridge & $h_{1,a}+A_a+\rho_a\beta_a$ & $\varphi+9$\\
\bottomrule
\end{tabular}

\end{table}

\paragraph{The golden robustness endpoint.}
At $a=\alpha$, $c_a=\alpha^3$, $h_{0,a}=2$, $h_{1,a}=m_a=M$,
$A_a=\varphi+3$, and $\rho_a=A_a/M$. The cost identities specialize to
\begin{equation}
 dv\le q_*(v,b)\le2\varphi v<Mv,\qquad
 q_*(v,b)\le dv+\alpha^3b,\label{sg:eq:affine}
\end{equation}
\begin{equation}
 q_*(W,B)-q_*(v,B)\ge\varphi(W-v),\label{sg:eq:slope}
\end{equation}
and the special terminal budgets become
\begin{equation}
 W,B\le\varphi h,\qquad
 \sum_iq_i\le dW+\alpha^3B\le(\varphi+3)h.
 \label{sg:eq:terminal}
\end{equation}
In particular, the critical outward case is
\begin{equation}
 F_J\le(M+1)h+W+q_*(W,B)
 \le(M+1+\varphi)h+q_*(\varphi h,\varphi h)
 =(5\varphi+3)h.\label{sg:eq:outward}
\end{equation}
All other endpoint coefficients in the table are smaller: their
positive slacks follow directly from $3/2<\varphi<2$.
This gives $K_\alpha=\varphi^5$.

\begin{theorem}[Exact shaped-cost guarantees]\label{sg:thm:exact}
Both $T_*$ and $T_*^F$ have consistency $d$ and exact strict
robustness $\varphi^5=5\varphi+3$.
\end{theorem}
\begin{proof}
The capacity theorem and Proposition~\ref{curve:prop:exact-consistency}
give exact consistency. The specialized joint-cut calculation gives
the robustness upper bound; Proposition~\ref{sg:prop:floor} below
supplies the matching family.
\end{proof}

\begin{proposition}[A matching outward family]
\label{curve:prop:outward-lower}
For every $a\in[\alpha,1]$, raw and filtered $T_a$ satisfy
$R\ge O_a$. At $a=1$ the exact pair is
$(C,R)=(4,7)$.
\end{proposition}
\begin{proof}
For $0<\epsilon<1/10$ offer
\[
 \begin{array}{ccl}
 K&0&[1-\epsilon,1+\beta_a-\epsilon),\\
 I&0&[\epsilon,1+\beta_a-2\epsilon),\\
 R&0&[1+\beta_a-2\epsilon,1+\beta_a-2\epsilon+q_{0,a}\beta_a-\epsilon),\\
 J&1&[0,1),\\
 L&0&[2\epsilon-m_a,\epsilon).
 \end{array}
\]
The anchor $K$ has length $\beta_a$, zero ledger, and cost
$q_{0,a}\beta_a$. Both $I$ and $R$ reject; their lengths are
$q_{0,a}-3\epsilon<q_{0,a}\beta_a$ and $q_{0,a}\beta_a-\epsilon$.
The unit positive $J$ fails ordinary acceptance and passes the
special length test at equality with zero ledger. After it replaces
$K$, the final negative $L$ rejects at length $m_a-\epsilon<m_a$.
Thus $\mathrm{ALG}=1$, while the compatible $L,I,R$ tile the
entire offered hull and give
$\mathrm{OPT}=m_a+1+\beta_a+q_{0,a}\beta_a-5\epsilon$.
There is only one positive, so filtering changes nothing.
Let $\epsilon$ decrease to zero. At $a=1$ the constants are
$q_{0,a}=2$, $m_a=Q_a=3$, $A_a=9/4$, $\rho_a=1$, $h_{0,a}=2$, $h_{1,a}=3$;
the table's maximum is $7$, matching this lower bound, and the
exact consistency proposition gives $4$.
\end{proof}

\begin{proposition}[Exact robustness on two parameter intervals]
\label{curve:prop:exact-arcs}
Let $a_-\in(\alpha,a_c)$ and $a_+\in(a_c,1)$ be the unique zeros
in these respective intervals of
\begin{equation}
 \begin{aligned}
 p_-(x)&=1+3x+x^2-5x^3\\
       &\hspace{1em}-6x^4-4x^5-x^6,\\
 p_+(x)&=1+2x+x^2-3x^3-3x^4-x^5.
 \end{aligned}
 \label{curve:eq:transition-polynomials}
\end{equation}
For both $T_a$ and $T_a^F$, the exact robustness is
\[
 R=
 \begin{cases}
 1+3/a+2/a^2,&\alpha\le a\le a_-,\\
 a+2+3/a+1/a^2,&a_+\le a\le1.
 \end{cases}
\]
On $(a_-,a_+)$, the bounds are $O_a\le R\le N_a$ with $O_a<N_a$.
In particular, $a=2/3$ gives the exact pair $(C,R)=(25/9,10)$.
\end{proposition}
\begin{proof}
On $[\alpha,a_c]$, substitution gives
\[
 \begin{gathered}
 \Delta_a:=O_a-N_a=\frac1{a^2}-1-a-\frac{1+2a}{(1+a)^3},\\
 a^2(1+a)^3\Delta_a=p_-(a).
 \end{gathered}
\]
Here $d\Delta_a/da=-2/a^3-1+(1+4a)/(1+a)^4<0$, since the last
fraction is below one. The signs $\Delta_\alpha=2\sqrt5-4>0$ and
$\Delta_{a_c}=-(1+2a_c)/(1+a_c)^3<0$ give the unique crossing $a_-$.
On $[a_c,1]$, put
\[
 \begin{gathered}
 g_a:=A_a-m_a=\frac1{a^2}+\frac1{(1+a)^2}-1-a,\\
 a^2(1+a)^2g_a=p_+(a).
 \end{gathered}
\]
Now $dg_a/da=-2/a^3-2/(1+a)^3-1<0$,
$g_{a_c}=(1+a_c)^{-2}>0$, and $g_1=-3/4$.
Thus its unique zero is $a_+$, where $\rho_a$ reaches one, and
\[
 O_a-N_a=
 \begin{cases}
 (m_a-A_a)(1+\beta_a/m_a),&a_c\le a\le a_+,\\
 m_a-A_a,&a_+\le a\le1.
 \end{cases}
\]
The signs and Proposition~\ref{curve:prop:outward-lower} prove
the two exact intervals and the strict middle certificate gap;
substituting the two branches of $m_a$ gives the displayed $R$.
Finally, $\Delta_{2/3}=119/1500>0$ and $2/3<a_c$, with $O_{2/3}=10$.
Proposition~\ref{curve:prop:exact-consistency} gives $C=25/9$.
\end{proof}

Numerically, $a_-\approx0.677821$, $a_c\approx0.754878$, and
$a_+\approx0.813824$.

\subsection{A matching family and the architectural limitation}

\begin{proposition}[A floor for fixed special rules]\label{sg:prop:floor}
Fix $a=\alpha$. Keep positive ordinary cost $Mv$, the shared special tests,
ordinary priority, fees, and ledger inheritance. If the negative
cost is finite and satisfies $q(v,0)\ge dv$, the resulting raw
and filtered policies have strict robustness at least $\varphi^5$.
\end{proposition}
\begin{proof}
For $0<\epsilon<1/100$ offer the following five requests in order:
\[
\begin{array}{ccl}
 K&0&[1-\epsilon,d-\epsilon),\\
 I&0&[\epsilon,d-2\epsilon),\\
 R&0&[d-2\epsilon,d-2\epsilon+M-\epsilon),\\
 J&1&[0,1),\\
 L&0&[2\epsilon-M,\epsilon).
\end{array}
\]
$K$ is accepted empty, has length $\varphi$, and carries zero
ledger. Both $I$ and $R$ overlap it and have length below
$q(\varphi,0)\ge M$, so they are rejected. The positive $J$
fails ordinary acceptance, passes the special length test at
equality and the zero-ledger test, and replaces $K$.
Finally $L$ overlaps $J$ and has length $M-\epsilon<M$, so it
is rejected. Thus $\mathrm{ALG}=1$. The compatible intervals
$L,I,R$ tile the entire offered hull; all five requests lie in
that hull. Therefore
\[
 \mathrm{OPT}=2(M-\epsilon)+(d-3\epsilon)
 =5\varphi+3-5\epsilon.
\]
There is only one positive, so filtering changes nothing.
Taking $\epsilon\downarrow0$ proves the floor and the matching
lower bound in Theorem~\ref{sg:thm:exact}.
\end{proof}

With positive cost $Mv$, these results make Shaped-BANK robustness-optimal
within the sufficient-cost class of Theorem~\ref{sg:thm:interface}.
The zero-ledger relaxation and the nonshrinking-cost floor are analyzed next.

\section{Cancellation penalties and replacement-cost limits}\label{app:cancellation-costs}
The buyback objective subtracts a proportional payment for each revoked
request from the value finally retained \citep{babaioff2009buyback}.
For the unchanged Shaped-BANK policy, its accepted-history bound also
controls this payment. The following result applies with or without the
permanent-history filter and for arbitrary raw predictions.

\begin{proposition}[Sharp cumulative revocation factor]
\label{r32:prop:cancellation}
Let $T$ be the total length of accepted identities subsequently revoked,
each counted once, and let $N,P$ be the final lengths held with effective
labels zero and one. Then
\[
 T\le N+2\varphi P\le2\varphi\,\mathrm{ALG}.
\]
The supremum of $T/\mathrm{ALG}$ over finite nonempty inputs is
$2\varphi$, also for the filtered policy.
\end{proposition}
\begin{proof}
The replacement forest partitions all accepted identities among the
final held roots. Equation~\eqref{sg:eq:history} bounds the accepted
history of a negative root by twice its length and that of a positive
root by $M$ times its length. Subtracting the final roots gives
$T\le N+(M-1)P=N+2\varphi P$.

For sharpness, put $n_k=d^k$ and $a_k=\alpha^2(d^k-1)$.
For $k=0,\ldots,K$, present in order
\[
 N_k=[0,n_k),\ p(N_k)=0,\qquad
 P_k=[a_k,a_{k+1}),\ p(P_k)=1.
\]
The positives are mutually compatible, have lengths $\alpha n_k$,
and lie inside their respective $N_k$, since
$a_{k+1}=n_k-\alpha^2$. Thus the permanent filter preserves their
labels. Before $P_k$, the negative holder has ledger
$\sum_{j<k}\alpha n_j=a_k\le n_k$.
The special length test accepts $P_k$ at equality, and its ledger
test passes. Then $N_{k+1}$ passes the positive holder's ordinary
test at equality, since $n_{k+1}=M\alpha n_k$.
The final reward is $\alpha d^K$, and
\[
 T=\sum_{k=0}^{K}d^k+\alpha\sum_{k=0}^{K-1}d^k
   =2\varphi\,\mathrm{ALG}-1.
\]
Taking $K\to\infty$ proves the supremum.
\end{proof}

Charge $c\,w(I)$ for revoking $I$, in the same units as retained
length, and define $\mathrm{ALG}_{c}=\mathrm{ALG}-cT$.
Finite net robustness means a positive constant $R_c$ satisfying
$\mathrm{OPT}\le R_c\mathrm{ALG}_{c}$ on every input.
An offline algorithm can select its final optimum without cancellations,
so its net optimum is still $\mathrm{OPT}$. For
$0\le c<1/(2\varphi)$, Proposition~\ref{r32:prop:cancellation} gives
\begin{equation}
 \mathrm{ALG}_{c}\ge(1-2\varphi c)\mathrm{ALG},\qquad
 R_c\le\frac{\varphi^5}{1-2\varphi c}.
 \label{r32:eq:net-robustness}
\end{equation}
Under perfect membership predictions the corresponding sufficient
consistency bound is $d/(1-2\varphi c)$. For arbitrary predictions
with the permanent filter, Corollary~\ref{new:cor:wrapper} instead gives
\[
 \mathrm{ALG}_{c}\ge(1-2\varphi c)
 \max\left\{\frac{\mathrm{OPT}}{\varphi^5},
             \frac{\mathrm{OPT}-\eta_K}{d}\right\}.
\]
These ratio bounds are sufficient.

The fee range for finite robustness is exact for this unchanged
algorithm. In the sharpness family all intervals lie in $N_K$,
itself an offered interval, so $\mathrm{OPT}=d^K$ and
\[
 \mathrm{ALG}_{c}=\alpha d^K(1-2\varphi c)+c.
\]
At $c=1/(2\varphi)$ the net reward is constant while the optimum
diverges; above that threshold it is eventually negative.
This threshold concerns the unmodified policy; a monetary refund rate must first be converted to retained-duration units to give $c$.

\subsection{Which replacement-cost assumptions are necessary?}
\label{r29:sec:necessary-cost}

The zero-ledger condition in the sufficient interface of Theorem~\ref{sg:thm:interface} is not
necessary. What matters in that proof is the cost paid along paths that
actually carry positive source mass.

\begin{theorem}[Consistency after relaxing the zero-ledger cost]
\label{r29:thm:zero-ledger}
In the shared ledger policy, suppose positive costs are at least $Mv$.
For negative costs require
\[
 q(v,b)\ge dv\quad(b>0),\qquad
 q(v,b)\ge\varphi v+\alpha b\quad(0<b\le dv).
\]
The value $q(v,0)$ may be any finite nonnegative number.
For every finite execution with feasible complete positive set $P$,
every held interval satisfies $b(X)\le d w(X)$ and
$w(P)\le d\,\mathrm{ALG}$.
\end{theorem}
\begin{proof}
Use the source coefficients and birth-order induction of
Theorem~\ref{sg:thm:interface}. A zero-ledger child has zero
coefficient for every positive source, so it contributes nothing to an
inherited measure, exterior mass, or whole-source cut. It can be omitted
from every such sum; the omitted ordinary costs are nonnegative.

We justify the geometric estimates separately, since they need not hold
for source-free accepted history. Follow any nonzero source share from
its accepted entry support to its current holder. At each replacement
edge on this path, the child has positive ledger immediately before
revocation. A negative child therefore costs at least $dv$, and a
positive child at least $Mv$. The original edge contractions apply:
an ordinary parent has a negative child of length at most $1/d$ times
its own, or a positive child of length at most $1/M$ times its own;
a special positive parent has negative children of total length at most
$\varphi$ times its own. Consecutive positive nodes on a source path
contract by at least $d$. The negative-tail displacement bounds remain
$\alpha v$ from a negative root and $du$ from a positive root.

Equivalently, retain the root and only replacement edges whose child
had positive ledger at revocation. Every nonzero source has its entry
support in this retained tree. Its interval union is connected.
Induction gives retained history at most $2v$ at a negative root and
$Mv$ at a positive root: an ordinary parent pays all retained child
costs, and a special parent has only negative children of total length
at most $\varphi$ times its own. Thus the negative-subtree protrusion
bound also holds for all source-bearing paths used by the proof.

The weighted cut has a nonzero left share, hence positive ledger, so
its cost inequalities are unchanged. At a negative birth, omit the
zero-ledger children and apply the original weighted or positive cuts
to the others. Their costs sum to at most the actual ordinary gate.
Birth exterior and cut bounds follow, with the zero-birth-ledger case
immediate. The lifetime argument and its three capacity cases use no
cost restriction at zero ledger.

At an ordinary positive birth of length $u$, every child satisfies
$b_i\le q(v_i,b_i)$, also when $b_i=0$, so the inherited ledger is at
most $u$. One must not infer that the total child length is at most
$u/d$: zero-ledger children can violate that inequality. The positive
cut proof only needs this length bound for the selected child carrying
a nonzero share, for which its individual cost gives $v_i\le u/d$.
At a special birth the actual tests still give both child length and
inherited ledger at most $\varphi u$. All remaining path and sibling
arguments carrying nonzero shares use the geometry established above.
The positive cuts and the birth induction therefore close as before.
Finally source conservation gives the claimed bound for $w(P)$.
\end{proof}

For example, define
\[
 q_0(v,b)=
 \begin{cases}\varphi v,&b=0,\\q_*(v,b),&b>0.\end{cases}
\]
With positive cost $Mv$, this policy is $d$-consistent and has finite
robustness, although $q_0(v,0)<dv$. It is homogeneous and nondecreasing
in ledger, with an explicit jump at zero. A finite robustness certificate follows. Total accepted subtree length is at most
$dv$ at a negative root and $(M+1)v=2dv$ at a positive root. At an
ordinary parent, child history divided by child cost is at most
$\varphi$; at a special parent all children are negative and their
total length is at most $\varphi v$. This proves both bounds
simultaneously. Thus total accepted history $H$ is at most
$2d\,\mathrm{ALG}$.

Distribute each rejected offline-comparator interval's length among its
then-held conflicts proportionally to their lengths. Ordinary rejection
and costs at most $Mv$ bound a receiver's charge by $Mv$. For a fixed
accepted identity, at most two compatible comparator intervals extend
past its endpoints; those contained inside it have total length at most
its own. Its total rejected-comparator charge is therefore at most
$(2M+1)v$. Accepted comparator identities contribute at most $H$.
Consequently $\mathrm{OPT}\le(2M+2)H\le8d^2\mathrm{ALG}$.
The universal lower bound then makes its consistency exactly $d$.

\begin{theorem}[A floor under nonshrinking negative costs]
\label{r29:thm:nonshrinking-floor}
Keep the ordinary-first sum gate, both special tests, fees, and full
ledger inheritance. Allow arbitrary fixed finite nonnegative cost
functions $q_-(v,b)$ and $q_+(v,b)$, but assume
\begin{equation}
 q_-(v,b)\ge v\qquad(v>0,\ b\ge0).
 \label{r29:eq:nonshrinking}
\end{equation}
If the raw or permanently filtered policy is $d$-consistent, then
\[
 q_-(v,v)\ge dv,\qquad q_+(v,dv)\ge Mv
 \qquad(v>0),
\]
and its strict robustness is at least $\varphi^5$.
No monotonicity, continuity, or homogeneity of either cost is assumed.
\end{theorem}
\begin{proof}
First fix $v>0$ and suppose $c=q_-(v,v)/v<d$; by
\eqref{r29:eq:nonshrinking}, $c\ge1$. All following coordinates are
multiplied by this fixed $v$, without rescaling the cost function.
Put $t=\alpha-\epsilon$ and $x=1-t$, with sufficiently small
$\epsilon>0$. Offer negative $V=[0,1)$, three positive tiles
partitioning $[1-\epsilon-x,1-\epsilon)$, then positive
$T_0=[1-\epsilon,1-\epsilon+t)$. All four positives fail the special
length test and are shorter than the ordinary cost, which is at least
one. Their rejections give ledger $x+t=1$ exactly. Negative
$W=[\epsilon-c,\epsilon)$ now replaces $V$ at ordinary equality.
Offer four positive tiles partitioning $[2\epsilon-c,0)$, and then
$T_1=[0,t-\epsilon)$. Their lengths are below both $\alpha c$ and $c$,
so all reject, leaving ledger
$B=c+1+t-3\epsilon$.

Finally offer a positive $L$ ending at $2\epsilon-c$ with length
$\ell=\alpha B-\epsilon$. It fails the special ledger test. If it
is accepted ordinarily, its ratio is $1+B/\ell>d$. If it is rejected,
the ratio is $(\varphi B-\epsilon)/c$, and
\[
 \varphi B-\epsilon-dc=d-c-(4\varphi+1)\epsilon>0.
\]
The positives are consecutive and tile the full offered hull; both
negatives lie inside it. Hence their labels are perfect in either
branch. This contradiction proves $q_-(v,v)\ge dv$ and uses no
positive cost.

Next fix $u>0$ and suppose $q_+(u,du)<Mu$. Choose $0<c<M$ with
$cu\ge q_+(u,du)$, and multiply all following coordinates by $u$.
Offer negative $K=[1-\epsilon,d-\epsilon)$ and three positive tiles
partitioning $[1,d)$. They reject, giving length and ledger both
$\varphi$. The first part makes its ordinary cost at least $M$.
Positive $Q=[0,1)$ is special at equality in both tests and has ledger
$d$. Negative $W=[\epsilon-c,\epsilon)$ replaces it ordinarily.
Four positive tiles partitioning $[2\epsilon-c,0)$ reject, leaving
ledger $B=d+c-2\epsilon$. A final positive $L$ ending at
$2\epsilon-c$, of length $\ell=\alpha B-\epsilon$, again gives
ratio $1+B/\ell>d$ if accepted, whereas rejection gives
\[
 \varphi B-\epsilon-dc=M-c-(2\varphi+1)\epsilon>0.
\]
The positives again partition the complete hull containing both
negatives. Thus $q_+(u,du)\ge Mu$.

For the robustness witness, use unit scale and negative
$K=[1-\epsilon,d-\epsilon)$ followed by three positive tiles
partitioning $[1,d)$. These rejections give ledger exactly $\varphi$,
so $q_-(\varphi,\varphi)\ge M$. Offer
\[
\begin{array}{ccl}
 I&0&[\epsilon,d-2\epsilon),\\
 R&0&[d-2\epsilon,d-2\epsilon+M-\epsilon),\\
 J&1&[0,1),\\
 L&0&[2\epsilon-M,\epsilon).
\end{array}
\]
The first two requests reject. Request $J$ is special at equality
in both tests, with ledger $d$, and its positive cost is at least $M$.
Thus $L$ also rejects and the final gain is one. The compatible
$L,I,R$ tile the entire hull and give
$\mathrm{OPT}=\varphi^5-5\epsilon$. The four positive predictions
are mutually compatible, so the filter changes no label in this
witness or in either perfect-label construction above. Letting
$\epsilon\downarrow0$ proves the bound.
\end{proof}

Unlike Proposition~\ref{sg:prop:floor}, which fixes positive cost $Mv$ and assumes $q(v,0)\ge dv$ but not $C=d$, this floor permits arbitrary fixed positive costs and assumes the global condition~\eqref{r29:eq:nonshrinking} together with $C=d$. In particular, changing the fixed positive coefficient to $\kappa<M$, including $\kappa=d$, cannot preserve $d$-consistency under these hypotheses. The unrestricted claim for arbitrary negative costs remains open.

\paragraph{Discounting positive-ledger states can fail.}
The lower-cost rule
\[
 q_{\rm low}(v,b)=\min\{2\varphi v,
 \max\{\min\{\varphi(v+b),dv\},\varphi v+\alpha b\}\}
\]
is not $d$-consistent, even with the original positive cost $Mv$.
Here is a 14-request family with $0<\epsilon\le10^{-3}$.
Put $v=\varphi^5$, $u=\varphi^4$, $w=\varphi^7$ and offer
\[
\begin{array}{ccl}
 V_0&0&[0,1),\\
 S_0,S_1&1&[0,1/2),\ [1/2,1),\\
 V_1&0&[1-\epsilon,\varphi+2-\epsilon),\\
 V_2&0&[\varphi+2-2\epsilon,4\varphi+3-2\epsilon),\\
 V_3&0&[4\varphi+3-3\epsilon,9\varphi+6-3\epsilon),\\
 Q&1&[9\varphi+6-4\epsilon,12\varphi+8-4\epsilon),\\
 W&0&[9\varphi+6-5\epsilon,22\varphi+14-5\epsilon).
\end{array}
\]
Next offer four positive tiles partitioning
$[12\varphi+8-4\epsilon,22\varphi+14-6\epsilon)$, followed by
positive $L=[d-3\epsilon,9\varphi+6-4\epsilon)$ and positive $R$
starting at $22\varphi+14-6\epsilon$ with length
$\alpha(1+w+\varphi^6-3\epsilon)-\epsilon$.

The first two positives reject. At ledger one, the successive negative
costs are $d,3\varphi+1,5\varphi+3$, so $V_1,V_2,V_3$ are ordinary
replacements at equality. Request $Q$ is special at length equality
and $W$ is ordinary at equality $w=Mu$. All remaining requests
reject: the four tiles and $L$ fail the length test, and $R$ fails
the ledger test. They also fail ordinary acceptance. Thus
\[
 \mathrm{ALG}=w,\qquad
 w(P)=dw+\varphi-(3\varphi+1)\epsilon>dw.
\]
To check optimality, the first two positives pay for $V_0$ at equality.
Request $L$ is longer than every compatible subset of
$\{V_1,V_2,V_3\}$, whose largest weight is
$w(V_1)+w(V_3)=6\varphi+4<\varphi^6-\epsilon$.
The positive $Q$, the four tiles, and $R$ all conflict with $W$ and
have total weight $w-2\epsilon+w(R)>w$. These three assigned positive
groups are disjoint and compatible, and each belongs to the positive
conflict set of its assigned negative group. Replacing any feasible
negative subset by its conflicting positives therefore never loses
weight. Hence $P$ is an optimum; it need not be unique because
$V_0$ ties $S_0,S_1$. The ratio tends to
$d+\varphi^{-6}\approx2.674>d$.
Unlike the zero-ledger relaxation, this policy transports nonzero
source mass through several discounted negative replacements.

\section{A barrier for singleton-negative replacement}
\label{app:history-barrier}

The five-request floor in Proposition~\ref{sg:prop:floor} fixes the
positive cost and special predicates. The following result permits
arbitrary ledgers, positive costs, and history-dependent rules. It requires
only one restriction on negative arrivals, evaluated before their decisions.

\begin{theorem}[Singleton-negative barrier]
\label{hist:thm:singleton-negative}\label{hist:prop:barrier}
Consider a deterministic policy $A$ in the model of Section~\ref{sec:model}
that satisfies the following condition at every reachable prearrival state:
\begin{quote}
\textbf{(N)} If the sole held request is a negative $K$ of length $v$, an
incoming negative that conflicts with $K$ and has length less than $dv$
is rejected.
\end{quote}
If $C(A)\le d$ and $R(A)<\infty$, then $R(A)\ge\varphi^5$.
The result allows voluntary revocation and arbitrary dependence on the
entire rejected history. It also holds after the permanent feasibility filter.
\end{theorem}

\begin{proof}
Fix a finite robustness bound $r>\max\{1,R(A)\}$ and
$0<\epsilon<1/1000$. Put
\[
 v=\varphi(1-4\epsilon),\qquad p=1+v-2\epsilon,
 \qquad z=dv-\epsilon.
\]
Offer the negative requests
\[
 K=[1-\epsilon,1-\epsilon+v),\qquad
 I=[\epsilon,p),\qquad R=[p,p+z)
\]
in this order. Finite robustness forces $K$ to be retained. Both subsequent
requests conflict with it and lie strictly below the threshold $dv$, with
margins $\varphi v-1+3\epsilon>0$ and $\epsilon$, respectively.
Condition (N) rejects both; discarding $K$ would leave zero reward.

Offer a finite positive partition of $[1,p)$, each tile shorter than $v/r$.
At the first departure from $K$, only the incoming tile could remain,
whereas historical $K$ gives $\mathrm{OPT}\ge v$. Robustness therefore
forces every tile to be missed and $K$ retained. Now offer positive
$J=[0,1)$. Let $\mathcal A$ consist of $J$ and these tiles; it tiles
$[0,p)$ and has weight $p$. This positive family is feasible but is not an
optimum of the current history, since $R$ is compatible with all of it.

\paragraph{A perfect completion forces retention of $J$.}
Suppose $J$ is rejected. Then $K$ remains, since deleting it gives zero
reward. Offer positive $T=[p,p+v)$ and put $E=p+v$. After any response,
the local held weight is at most $v$, whereas
\[
 w(\mathcal A\cup\{T\})=E,\qquad E-dv=2\epsilon.
\]
Set $\eta=\epsilon/2$, choose finite
$h>\max\{dv,(r-1)v,1\}$, and put
\[
 u=\alpha h,\quad b=E+u-\eta,\quad
 H=[b,b+h),\quad U=[E,b+\eta).
\]
Offer negative $H$. It is disjoint from all earlier requests; in particular,
\[
 b-(p+z)=\alpha h-\varphi v+\epsilon-\eta>0.
\]
The historical packing $\{K,H\}$ has weight $v+h>rv$, so $H$ must
be retained. Offer a finite positive partition $\mathcal F$ of
$[b+\eta,b+h-\eta)$, with each tile shorter than $(v+h)/r-v>0$.
At the first departure from $H$, the held weight is at most $v$ plus the
incoming tile. The same historical packing forces every tile to be missed
and $H$ retained, even if the local holder is voluntarily revoked.

Offer positive $U$. If retained, stop, with held weight at most $v+\alpha h$.
If $U$ is not retained, robustness forces $H$ to remain. In that case offer
positive $V=[b+h-\eta,b+2h-\eta)$ and stop; every response has held
weight at most $v+h$. All positives are compatible. Their totals in the
two branches are $E+\varphi h-2\eta$ and $E+dh-2\eta$, so
\begin{align*}
 E+\varphi h-2\eta-d(v+\alpha h)&=\epsilon,\\
 E+dh-2\eta-d(v+h)&=\epsilon.
\end{align*}

To apply consistency, we check the entire history. At the first terminal
the negatives are $K,I,R,H$. Every feasible negative family is a subset
of $\{K,H\}$ or $\{I,R,H\}$. The following positive groups conflict
with the indicated negative families and have the stated strict surpluses:
\begin{center}
\begin{tabular}{lll}\toprule
Negative family & Positive group & Surplus\\\midrule
$\{K\}$ & $\mathcal A\cup\{T\}$ & $E-v$\\
$\{I\}$ & $\mathcal A$ & $\epsilon$\\
$\{R\}$ & $\{T,U\}$ & $\alpha h-\varphi v+\epsilon$\\
$\{H\}$ & $\{U\}\cup\mathcal F$ & $\alpha h-2\eta$\\
$\{R,H\}$ & $\{T,U\}\cup\mathcal F$
 & $\alpha h-\varphi v+\epsilon-2\eta$\\\bottomrule
\end{tabular}
\end{center}
All margins are positive because $\alpha h>\varphi v$ and
$2\eta=\epsilon$. The groups for $K$ and $H$ are disjoint. The group
$\mathcal A$ for $I$ is disjoint from each group for $R$, $H$, or
$\{R,H\}$. Adding the corresponding margins therefore covers all nine
nonempty feasible negative subsets, counting the shared $U$ only once.
Adding $V$ preserves these certificates. The strict exchange argument
in~\eqref{r11:eq:strict-exchange} proves that the complete positives are
the unique optimum at both terminals. Either terminal contradicts
$C(A)\le d$. Thus $J$ must be accepted and retained on the common prefix;
accepting and then discarding it would leave zero reward.

\paragraph{A second perfect completion forces rejection of $L$.}
Define
\[
 \ell=\varphi p-2\epsilon,\qquad s=\epsilon-\ell,
 \qquad L=[s,\epsilon),
\]
and offer negative $L$. It overlaps the sole holder $J$ by $\epsilon$.
Suppose it is accepted; finite robustness forces its retention. Put
$\chi=2\alpha-1=\alpha^3>0$ and $\eta=\chi\epsilon/4$.
Offer a positive partition of $[s+\eta,0)$ with each tile shorter than
$\ell/r$. Historical $L$ forces every tile to be missed by the same
first-departure argument. Next offer positive
$B=[s+\eta-\ell,s+\eta)$, of length $\ell$.
It overlaps $L$ by $\eta$ and is compatible with every earlier positive.
After either response the local held weight is at most $\ell$.
Let $\mathcal Z$ consist of $B$ and this left partition. The full local
positive family $\mathcal P=\mathcal Z\cup\mathcal A$ tiles
$[s+\eta-\ell,p)$, with weight
\begin{equation}
 P_0=p+2\ell-\epsilon-\eta,\qquad
 P_0-d\ell=\chi\epsilon-\eta=3\chi\epsilon/4.
 \label{hist:eq:local-excess}
\end{equation}

The negatives are $\mathcal N=\{K,I,L\}$ and the unavailable
$R=[p,p+z)$. Among $K,I,L$, only $K$ and $I$ conflict. All five
nonempty feasible negative subsets have strict exchange certificates:
\begin{center}
\begin{tabular}{lll}\toprule
Negative family & Positive group & Surplus\\\midrule
$\{K\}$ & $\mathcal A$ & $p-v=1-2\epsilon$\\
$\{I\}$ & $\mathcal A$ & $\epsilon$\\
$\{L\}$ & $\mathcal Z\cup\{J\}$ & $\ell+1-\epsilon-\eta$\\
$\{K,L\}$ & $\mathcal Z\cup\mathcal A$
 & $p+\ell-v-\epsilon-\eta$\\
$\{I,L\}$ & $\mathcal Z\cup\mathcal A$ & $\ell-\eta$\\\bottomrule
\end{tabular}
\end{center}
Every margin is positive for the stated $\epsilon$ range. Apply
Lemma~\ref{r11:lem:quantitative-completion} with $c=d$, $\gamma=0$,
$a=\ell$, and $\zeta=\chi\epsilon/8$. Historical $L$ has weight $a$.
Choose finite $h>\max\{(r-1)\ell,1\}$ with
$\alpha h>z+2\zeta$. Then $h>2\zeta$ and
$p+\alpha h-\zeta>p+z$, the largest old right endpoint, while
\[
 P_0-d\ell-2\zeta=\chi\epsilon/2>0.
\]
Thus every hypothesis of the lemma holds. Its satellite construction
repairs $R$ as well as the new negative $H$; for their joint subset the
new conflict group has surplus $\alpha h-z-2\zeta>0$ and counts its
shared boundary positive once. This group is disjoint from every local
group above, so mixed negative subsets also satisfy strict exchange.
The resulting robustness or actual-perfect consistency violation rules
out acceptance of $L$.

\paragraph{The finite robustness witness.}
Stop on the common prefix immediately after $L$, which must be rejected.
$J$ remains the sole holder, so $\mathrm{ALG}=1$. The compatible
negatives $L,I,R$ tile the entire offered hull, containing every request
on this short input. Length reward therefore gives the exact optimum
\begin{align*}
 \mathrm{OPT}
 &=\ell+(p-\epsilon)+(dv-\epsilon)\\
 &=d+2dv-(2\varphi+6)\epsilon\\
 &=\varphi^5-(18\varphi+14)\epsilon.
\end{align*}
Every partition above is finite: a positive mass can be split into
finitely many tiles below any positive threshold. The two longer
completions are causal witnesses that force decisions on their shared
prefixes; neither changes an earlier label or assumes prefix optimality.
Taking the supremum over positive $\epsilon$ proves $R(A)\ge\varphi^5$.
All constructed positive families are compatible, so permanent filtering
changes no bit and the same witnesses apply.
\end{proof}

In particular, a source-cost policy whose sole negative holder of length
$v$ costs at least $dv$, and whose negatives have no acceptance route
outside that cost gate, satisfies (N). Its positive costs, source fees,
inheritance, and history-dependent special rules do not enter the proof.
Any deterministic policy
with $C\le d$ and finite $R<\varphi^5$ must instead accept some conflicting
negative shorter than $dv$ at a reachable state whose sole holder is a
negative of length $v$.

The zero-ledger policy $q_0$ in Appendix~\ref{r29:sec:necessary-cost}
gives a concrete escape from this early negative prefix. With the same
$v=\varphi(1-4\epsilon)$ and $p=1+v-2\epsilon$, offer $K$ followed
immediately by $I=[\epsilon,p)$, before any positive source arrives.
The ledger is zero and
\[
 w(I)-q_0(v,0)=1+v-3\epsilon-\varphi v=\epsilon>0.
\]
Thus $q_0$ replaces $K$ by $I$, breaking the intended singleton-anchor
branch. This policy is nevertheless $d$-consistent and finitely robust.
It shows why an early-negative branch cannot simply be presumed to reject.
However, $q_0(v,b)\ge v$ at every negative state, so
Theorem~\ref{r29:thm:nonshrinking-floor} still gives
$R(q_0)\ge\varphi^5$.

\section{Additional ledger policies and guarantees}\label{app:additional-policies}
Credit varies the negative ledger coefficient; SUM uses total
conflicting length. MAX uses maximum conflicting length and is retained
as a descriptive experimental comparator. Credit-BANK $T_c$ has negative cost
$\min\{Mv,dv+cb\}$, $c\ge0$. The constant-cost gates are
\begin{equation}
 G_\lambda(D)=
 \begin{cases}
 \lambda\sum_{J\in D}w(J),&\text{SUM},\\
 \lambda\max_{J\in D}w(J),&\text{MAX}.
 \end{cases}
 \label{eq:domain}
\end{equation}
\begin{definition}[Ledger policies]\label{sum:def:policy}
Shaped-BANK $T_*$ uses~\eqref{sg:eq:main-cost} and positive cost $Mv$;
Credit-BANK $T_c$ uses~\eqref{credit:eq:cost}.
SUM-BANK $S_\lambda$, for $\lambda>d$, and MAX-BANK $A_\lambda$,
for $\lambda\ge d$, use Algorithm~\ref{alg:bank} with the respective
gate in~\eqref{eq:domain}. All other rules are shared: ordinary
acceptance has priority, all thresholds accept equality, and special
acceptance requires a positive prediction, only negative conflicts,
and both length and ledger tests.
\end{definition}
The guarantees below concern Shaped, Credit, SUM and RP.

Table~\ref{tab:extended-guarantees} collects the endpoint comparisons.
The full scalar envelope, its exact consistency threshold
$c_*=\alpha^3$, and its global minimum $4\varphi+5$ are in
Theorem~\ref{credit:thm:family}. The fixed comparison $T=T_{\varphi/2}$
has exact robustness $6\varphi+3$.
For SUM, Corollary~\ref{sharp:cor:exact-endpoint} proves the policy-specific
threshold $\lambda\ge\varphi^3$, and Proposition~\ref{sum:prop:robustness}
gives exact robustness $\varphi(2\lambda+1)$ on that range, with the
upper bound $K(\varphi,\lambda)$ from~\eqref{sum:eq:transfer} for $\lambda>d$.
For $c\in[c_*,1]$ and $\lambda\ge\varphi^3$, respectively, both families
satisfy $w(P)\le d\,\mathrm{ALG}$ for feasible complete positives.

For RP, Corollary~\ref{sum:cor:rp} proves exact robustness $2\lambda+2$
when $\lambda\ge2$, and $\lambda+2\lambda/(\lambda-1)$ when
$1<\lambda<2$. Its published consistency bound
$3\lambda/(\lambda-1)$ need not be exact.

Theorem~\ref{hist:thm:singleton-negative} treats the shared singleton-negative obstruction. Tables~\ref{tab:ledger-ablation} and
\ref{tab:absolute-endpoints} compare the policies empirically;
Appendices~\ref{app:ledger-ablation}--\ref{sec:experiment-protocol} specify
the matched repeats and differing robustness budgets.

\begin{table}[t]
\caption{Deterministic guarantees. $d=\varphi^2$, $c_*=\alpha^3$, $B_*=\varphi^5$, $B_{C*}=4\varphi+5$, $B_S=7\varphi+4$. In paired rows the error guarantee applies only to the filtered policy (superscript $F$); $\eta_K$ uses original labels. LR is due to \citet{garay1997}; its interval formulation, Naive and the RP consistency bound are given by \citet{karavasilis2025}. }
\label{tab:extended-guarantees}
\centering\scriptsize\setlength{\tabcolsep}{2.5pt}
\begin{tabular}{@{}lccc@{}}\toprule
Algorithm & Consistency & Robustness & Error guarantee\\\midrule
LR & $M$ & $M$ (exact) & ---\\
RP$_\lambda$, $\lambda\ge2$ & $\le\frac{3\lambda}{\lambda-1}$ & $2\lambda+2$ (exact) & ---\\
Naive & $1$ & $\infty$ & $\mathrm{OPT}\le\mathrm{ALG}+\eta_K$\\
Shaped $T_*,T_*^F$ & $d$ & $B_*$ (exact) & $\mathrm{OPT}\le d\mathrm{ALG}+\eta_K$\\
Credit $T_{c_*},T_{c_*}^F$ & $d$ & $B_{C*}$ (exact) & $\mathrm{OPT}\le d\mathrm{ALG}+\eta_K$\\
SUM $S_{\varphi^3},S_{\varphi^3}^F$ & $d$ & $B_S$ (exact) & $\mathrm{OPT}\le d\mathrm{ALG}+\eta_K$\\
\bottomrule
\end{tabular}\end{table}

\section{Credit-dependent replacement costs}\label{app:credit}

Put $\alpha=1/\varphi$, $d=\varphi^2$, and $M=\varphi^3$.
For any fixed real $c\ge0$, Credit-BANK $T_c$ uses the shared
ledger updates and special tests of Algorithm~\ref{alg:bank}, with
ordinary gate $\sum_{J\in D}\tau_c(J)$, where
\begin{equation}
 \tau_c(J)=\begin{cases}
 Mw(J),&p(J)=1,\\
 \min\{Mw(J),dw(J)+c b(J)\},&p(J)=0.
 \end{cases}\label{credit:eq:cost}
\end{equation}
Write $q_c(v,b)=\min\{Mv,dv+cb\}$ for the negative cost.
The coefficient changes only this ordinary cost and is fixed for
an entire execution. Ordinary acceptance has priority, thresholds
accept equality, and rejection fees use simultaneous old-state
snapshots with length-failure priority. Special acceptance still
requires a positive arrival, only negative conflicts,
$W,B\le\varphi w(I)$, and full ledger inheritance.
The permanent-filter composition is denoted $T_c^F$.

\begin{theorem}[Exact scalar envelope and consistency threshold]
\label{credit:thm:family}
For $0\le c\le1$, both raw and filtered policies have exact robustness
\begin{align}
 R(T_c)=R(T_c^F)&=K(c):=\max\{K_{\rm out}(c),K_{\rm br}(c)\},\notag\\
 K_{\rm out}(c)&=5\varphi+3+\varphi c,\notag\\
 K_{\rm br}(c)&=5\varphi+2+2\varphi c.
 \label{credit:eq:exact-R}
\end{align}
Within this range, their consistency equals $d$ if and only if
$c\ge c_*:=\alpha^3=2\varphi-3$. For $c\in[c_*,1]$, every
execution with feasible complete positive set $P$ satisfies
\begin{equation}
 b(J)\le d w(J)\ \text{throughout every held lifetime},
 \qquad w(P)\le d\,\mathrm{ALG}.\label{credit:eq:strong}
\end{equation}
Among all $c\ge0$ whose raw policy has consistency $d$, the minimum
robustness is $4\varphi+5$, attained only at $c=c_*$. The same
statement holds for the filtered policies.
\end{theorem}

The branches in~\eqref{credit:eq:exact-R} agree at $c=\alpha$;
the outward branch dominates below it and the bridge branch above it.
In particular, the comparison policy $T=T_{\varphi/2}$ retains
$C(T)=d$ and $R(T)=6\varphi+3$.

\subsection{Consistency from the common source-measure interface}

Theorem~\ref{sg:thm:interface} applies precisely when the scalar
cost satisfies its two negative-cost inequalities. Globally $q_c(v,b)\ge dv$.
For $0\le b\le dv$, the cap branch obeys
\[
 \varphi v+\alpha b\le(\varphi+\alpha d)v
 =2\varphi v<Mv.
\]
The affine branch dominates the same expression exactly when
\begin{equation}
 v+(c-\alpha)b\ge0\qquad(0\le b\le dv).
 \label{credit:eq:interface}
\end{equation}
For $c\ge\alpha$ this is immediate. For $0\le c<\alpha$, its
minimum over this interval is $[1+d(c-\alpha)]v$, so the necessary
and sufficient condition for the interface is
\[
 c\ge\alpha-1/d=\alpha-\alpha^2=\alpha^3=c_*.
\]
At $c=c_*$, equality occurs at $b=dv$, where the cost is
$2\varphi v<Mv$. Thus equality is included without invoking the cap.
The source-measure induction proves
\eqref{credit:eq:strong}, in fact for every $c\ge c_*$.
On $[c_*,1]$ the finite robustness proved below and
Theorem~\ref{thm:lower} give exact consistency $d$.
The following finite construction proves necessity.

\subsection{Eleven requests proving the consistency threshold}
\label{credit:sec:necessity}

Fix $0\le c<c_*$ and put
\[
 k=3\varphi+1,\qquad
 \Delta=2\varphi-d(1+c)>0,\qquad
 D=(\varphi-c)k+1>0.
\]
Choose
\begin{equation}
 0<\epsilon<\min\{1/100,\Delta/(2D)\}.
 \label{credit:eq:necessity-epsilon}
\end{equation}
First offer the negative unit interval $X=[0,1)$ and then the
positive seed intervals $[\epsilon,1/2)$ and $[1/2,1-\epsilon)$.
Their lengths $1/2-\epsilon<\alpha$ fail the special length test
and are below the ordinary cost, which is at least $d$.
Both reject and deposit their complete lengths into the sole
held conflict $X$, giving ledger $1-2\epsilon$.
The positive interval
\[
 A=[2\epsilon-\alpha,\epsilon)
\]
has length $\alpha-\epsilon$ and rejects for the same reasons.
The ledger is now $b_1=\varphi-3\epsilon$.

Next offer positive $B$ starting at $1-\epsilon$ with length
\[
 \ell_B=\alpha b_1-\epsilon
 =1-(3\alpha+1)\epsilon.
\]
Since $\epsilon<1/100$ and $\alpha<31/50$, we have
$\ell_B>0.96>\alpha$. Thus $B$ passes the special length test
but fails the ledger test: $\varphi\ell_B=b_1-\varphi\epsilon<b_1$.
It also fails the ordinary test because $\ell_B<1<d$.
Its complete length enters $X$'s ledger, which becomes
\[
 b_X=d-k\epsilon.
\]
Let $g=d+cb_X$ and offer negative
$J=[\epsilon-g,\epsilon)$. The cap on $X$ is inactive, since
\[
 g\le d(1+c)<d(1+c_*)=2\varphi<M.
\]
Therefore $J$ has exactly the old ordinary cost of $X$ as its
length. It overlaps $X$ in $[0,\epsilon)$, accepts at equality,
and inherits all of $b_X$. We have $d\le g<2\varphi$.

Partition $[2\epsilon-g,2\epsilon-\alpha)$ into four equal
positive tiles. Each has positive length $(g-\alpha)/4<\alpha g$
because $4\alpha>1$, and its length is below $dg$, a lower bound
on $J$'s ordinary cost. Each tile overlaps $J$, fails both the
ordinary test and the special length test, and deposits its full
length. After the four tiles, the ledger of $J$ is
\[
 B_{\rm pre}=b_X+g-\alpha=g+2-k\epsilon.
\]
Finally offer positive $C$ ending at $2\epsilon-g$ with length
\[
 z=\alpha B_{\rm pre}-\epsilon.
\]
It passes the special length test because
$\varphi z-g=2-(k+\varphi)\epsilon>0$, using
$k+\varphi=4\varphi+1<8$. It fails the ledger test strictly,
as $\varphi z=B_{\rm pre}-\varphi\epsilon<B_{\rm pre}$.
It also fails ordinary acceptance, since
\[
 z<\alpha(g+2)<dg;
\]
the last inequality follows from $d-\alpha=2$ and $g>\alpha$.
Its overlap with $J$ is positive: $z>\alpha g>\epsilon$.
Thus every fee in the construction is the full arriving positive
length, and the final held set is $\{J\}$ with reward $g$.

In spatial order the positive requests are $C$, the four new tiles,
$A$, the two seed tiles, and $B$. They are consecutive and tile one
interval. Its left end lies left of $J$ because $z>\epsilon$, and
its right end lies right of $X$ because $\ell_B>\epsilon$;
hence it contains both historical negatives. The positive tiling
attains the length of the entire offered hull and is optimal.
Moreover it is the unique optimum. Covering the two extreme
segments requires $C$ and $B$, respectively. They exclude $J$
and $X$, leaving only the positive tiles, all of which are needed
to cover their positive-length segments. The labels are therefore
perfect and the filter changes none of them.

The exact optimum is
$\mathrm{OPT}=B_{\rm pre}+z=\varphi(g+2-k\epsilon)-\epsilon$,
so
\begin{align}
 \mathrm{OPT}-d\,\mathrm{ALG}
 &=2\varphi-g-(\varphi k+1)\epsilon\notag\\
 &=\Delta-[(\varphi-c)k+1]\epsilon
 =\Delta-D\epsilon>\Delta/2>0.
 \label{credit:eq:necessity-gap}
\end{align}
This is a strict finite violation for every $c<c_*$, for both raw
and filtered policies. At $c=c_*$, all gates and the uniqueness
argument still hold for $0<\epsilon<1/100$; only the strict
violation disappears. Its limit is
\[
 \lim_{\epsilon\downarrow0}\frac{\mathrm{OPT}}{\mathrm{ALG}}
 =\varphi+\frac{2\alpha}{1+c_*}=d.
\]
For positive $\epsilon$ at this endpoint,
$\mathrm{OPT}-dg=-D\epsilon<0$, consistent with the upper bound.

\subsection{Terminal costs and the scalar robustness envelope}

We use the replacement-forest and connected-hull framework of Lemma~\ref{sg:lem:charge-transfer} and the proof of Theorem~\ref{curve:thm:robustness}. The scalar-specific
estimates below verify its hypotheses without any feasible-label
assumption. Throughout this subsection $0\le c\le1$.
For each accepted identity $X$, let $q_X$ be its ordinary cost
immediately before revocation, or at termination if retained.
Ledgers grow by nonnegative fees during a held lifetime, so costs
are nondecreasing and $dw(X)\le q_X\le Mw(X)$.
An ordinary parent of length $h$ satisfies $\sum_iq_i\le h$.
A special parent is positive and has negative, ordinarily accepted
children. Its special tests imply
\begin{equation}
 W,B\le\varphi h,\qquad
 \sum_iq_i\le dW+cB\le Ah,\qquad
 A=M+\varphi c,\quad\rho=A/M=1+c/d.
 \label{credit:eq:terminal}
\end{equation}
In this range $1\le\rho\le1+1/d<d$.

The terminal-cost induction from the shaped proof applies with
these $A,\rho$: an ordinary root's omitted-root bound is
$\rho\sum_iq_i\le\rho h<q_J$; a special root's ordinary
children sum to at most $Ah$, as does its singleton $Mh$.
Thus the full antichain bound is $q_J$ or $Aw(J)$ respectively.
The history bound~\eqref{sg:eq:history} also applies, since it
uses only the same positive cost $Mh$, negative lower bound $dh$,
and special length gate. An ordinary birth gives history at most
$h+\sum_iq_i\le2h$; a special birth gives $h+2W\le Mh$.
The connected hulls obey these same bounds.

Use Lemma~\ref{sg:lem:charge-transfer} with $\omega_X=q_X$:
\begin{equation}
 F_J(x,y)=|H_J\cap[x,y]|+L_J(x)+U_J(y).
 \label{credit:eq:joint-cut}
\end{equation}
At an ordinary root, an eligible root singleton attains the side
maximum $q_J$; an ineligible side has weight at most $\rho h$.
With two eligible endpoints the gap lies inside the root and has
length at most $h$; otherwise use its full history bound $2h$.
Thus the coefficients for two, one or no eligible endpoints are
\begin{equation}
 1+2M,\qquad 2+M+\rho,\qquad 2+2\rho.
 \label{credit:eq:ordinary-cases}
\end{equation}

At a special root $J=[a,b)$, each side chooses the eligible root
singleton of cost $Mh$ or the sum of the children's side maxima.
A zero maximum may choose children. Two singleton choices give
$(1+2M)h$. With exactly one singleton, reflect so it is the left.
If $b>y$, the gap is inside $J$, giving $(1+M+A)h$.
If $b\le y$, every child starts below $b$, so at most one
compatible child $k$ ends above $y$. If such a child exists, the
gap has length at most $h+h_k$ and
\begin{align}
 F_J&\le(M+1)h+(1-\rho)h_k+q_k+\rho W\notag\\
 &\le(M+1)h+(1+d-\rho)h_k+cb_k+\rho W\notag\\
 &\le(M+1)h+(1+d)W+cB\notag\\
 &\le[ M+1+\varphi(1+d+c)]h=K_{\rm out}(c)h.
 \label{credit:eq:outward}
\end{align}
Here $q_k\le dh_k+cb_k$, $1+d-\rho>0$, $h_k\le W$,
and $b_k\le B$. Without an outward child, all right-side maxima
are ineligible, and the full-hull bound gives
$(2M+\rho\varphi)h$.

If both endpoint choices use children, at most one child bridges
both cuts, meaning $s(C_k)<x\le y<f(C_k)$. This child contains
the gap and contributes at most $h_k+2q_k$. Every other child
has at most one eligible endpoint and contributes at most
$q_i+\rho h_i$. This also bounds a child with two ineligible
endpoints because $q_i\ge dh_i>\rho h_i$. Consequently
\begin{align}
 F_J&\le h_k+2q_k+\sum_{i\ne k}(q_i+\rho h_i)\notag\\
 &\le W+2\sum_iq_i\le(1+2d)W+2cB
 \le K_{\rm br}(c)h.
 \label{credit:eq:bridge}
\end{align}
The second line uses $(\rho-1)h_i\le q_i$.
If no child bridges, each contributes at most $q_i+\rho h_i$;
together with the full-hull bound, this gives
$(M+A+\rho\varphi)h$.

All remaining coefficients lie below the two-term envelope.
For the special-root cases,
\begin{align*}
 K_{\rm out}-(1+2M)&=\varphi+\varphi c,\\
 K_{\rm out}-(1+M+A)&=\varphi,\\
 K_{\rm out}-(2M+\rho\varphi)&=1+c,\\
 K_{\rm br}-(M+A+\rho\varphi)&=c.
\end{align*}
For the remaining ordinary cases,
\begin{align*}
 K_{\rm out}-(2+M+\rho)&=3\varphi-1+2\alpha c>0,\\
 K_{\rm out}-(2+2\rho)&=5\varphi-1+(3\varphi-4)c>0.
\end{align*}
At $c=0$ one non-bridge comparison is equality, which is sufficient.
Strict endpoint eligibility assigns contact to the ineligible case;
bridge uniqueness holds also when $x=y$. We have proved
$F_J(x,y)\le K(c)w(J)$ in every case.

For the charge hypothesis, assign each rejected $I$ pieces
$w(I)\tau_c(X)/\sum_{Y\in D}\tau_c(Y)$ to its current conflicts.
Strict ordinary rejection and lifetime monotonicity bound each
piece by $q_X\ge w(X)$. These charges are used only in the analysis. The full antichain bound is at most $Ah$, and
$K_{\rm out}-A=3\varphi+2>0$. Lemma~\ref{sg:lem:charge-transfer}
therefore yields $\mathrm{OPT}\le K(c)\mathrm{ALG}$.
The same upper bound holds on effective labels after filtering.

\subsection{Two finite families matching the envelope}

\paragraph{Eight requests for the outward coefficient.}
Fix $0<\epsilon<1/100$. Offer negative
$K=[1-\epsilon,d-\epsilon)$, of length $\varphi$, then three
positive tiles partitioning $[1+\epsilon,d-\epsilon)$ equally.
Each tile has length $(\varphi-2\epsilon)/3<1=\alpha w(K)$,
so it fails the special length test. Its length is also below
$dw(K)=M$, so it rejects ordinarily. Each pays its full length
to $K$, giving ledger $b_K=\varphi-2\epsilon$.
For $0\le c\le1$, put $g=M+cb_K$. The cost is exactly $g$,
since $g\le M+\varphi<M\varphi$.
Next offer, in this order,
\[
 \begin{array}{ccl}
 I&0&[\epsilon,d-2\epsilon),\\
 R&0&[d-2\epsilon,d-2\epsilon+g-\epsilon),\\
 J&1&[0,1),\\
 L&0&[2\epsilon-M,\epsilon).
 \end{array}
\]
Both $I$ and $R$ overlap the held $K$ and have lengths
$d-3\epsilon<g$ and $g-\epsilon<g$, respectively; being
negative, they reject without changing the ledger.
The positive $J$ fails ordinary acceptance and passes both
special tests: $\varphi w(J)=w(K)$ and $b_K<\varphi$.
It replaces $K$. Finally $L$ overlaps $J$ in positive length
and has length $M-\epsilon<M$, so it rejects against the
positive ordinary cost. The algorithm ends with reward one.

The four positive intervals are compatible, so filtering leaves
all labels unchanged. The intervals $L,I,R$ are consecutive,
compatible, and tile the entire offered hull; all other requests
are contained in that hull. Therefore their total is exactly
\begin{align}
 \mathrm{OPT}&=(M-\epsilon)+(d-3\epsilon)+(g-\epsilon)\notag\\
 &=K_{\rm out}(c)-(5+2c)\epsilon.
 \label{credit:eq:outward-witness}
\end{align}
Taking positive $\epsilon\downarrow0$ gives the outward lower
supremum for both raw and filtered policies.

\paragraph{Seven requests for the bridge coefficient.}
Fix $0<\epsilon<1/100$, put $p=(1-2\epsilon)/3$,
$b_0=1-2\epsilon$, $g=d+cb_0$, and $a=g-\epsilon$.
Offer negative $J=[0,1)$, then three positive tiles partitioning
$[\epsilon,1-\epsilon)$ into lengths $p$, then negative
\[
 L=[\epsilon-a,\epsilon),\qquad
 R=[1-\epsilon,1-\epsilon+a),
\]
and finally positive $Q=[1-\epsilon,1-\epsilon+\alpha)$.
Every tile fails ordinary acceptance and the special length test
because $p<\alpha$, pays exactly $p$ to the sole conflict $J$,
and leaves it held. The ledger becomes $b_0$. Its ordinary cost
is exactly $g<M$, since $d+cb_0<d+1<M$.
Both exterior negatives overlap $J$ in length $\epsilon$ and
have length $a<g$, so they reject. The final $Q$ fails the
ordinary gate but satisfies both special tests:
$\varphi w(Q)=1$ and $b_0<1$. Thus $\mathrm{ALG}=\alpha$.

The two exterior negatives and the three tiles form a feasible
tiling of the entire input hull. All offered requests lie in it,
including $Q\subset R$ because $a>\alpha$. Thus
\begin{equation}
 \frac{\mathrm{OPT}}{\mathrm{ALG}}
 =\frac{2a+b_0}{\alpha}
 =K_{\rm br}(c)-4\varphi(c+1)\epsilon.
 \label{credit:eq:bridge-witness}
\end{equation}
All positives are compatible, so the same execution occurs after
filtering. Their feasible positive mass has ratio
$d-2\varphi\epsilon<d$; the robustness witness has imperfect
labels. Taking positive $\epsilon\downarrow0$ gives the bridge
lower supremum. Since
$K_{\rm out}-K_{\rm br}=1-\varphi c$, the two constructions
match the upper envelope, including the switch at $c=\alpha$
and both endpoints $c=0,1$.

\subsection{The global scalar minimum}
\label{credit:sec:global-minimum}

It remains to exclude larger coefficients. For any fixed $c\ge1$, use the eight-request
outward construction with the actual capped cost
\[
 g=\min\{M\varphi,M+c(\varphi-2\epsilon)\}.
\]
The tile rejections still follow from the lower cost $M$ on $K$.
The negative requests $I,R$ have lengths $d-3\epsilon<g$ and
$g-\epsilon<g$, so reject. The same unit positive $J$ fails the
ordinary gate and passes both special tests; the held negative's ledger
is still $\varphi-2\epsilon$. The final $L$ rejects against $M$.
The compatible $L,I,R$ still tile the entire input hull; hence
$\mathrm{ALG}=1$ and $\mathrm{OPT}=M+d+g-5\epsilon$.
Continuity of the minimum as $\epsilon\downarrow0$ gives
\begin{align}
 R(T_c),R(T_c^F)
 &\ge M+d+\min\{M\varphi,M+c\varphi\}\notag\\
 &=\varphi^5+\varphi\min\{c,\varphi\}\notag\\
 &\ge\varphi^5+\varphi>4\varphi+5.
 \label{credit:eq:large-c-floor}
\end{align}
This remains valid at $c=\varphi$, even if the cap is only reached
in the limit. The filter changes no positive label.

On $[0,1]$ both affine robustness branches have strictly positive
slope, so $K(c)$ is strictly increasing. Since $c_*<\alpha$,
\[
 K(c_*)=\varphi^5+\varphi\alpha^3
 =\varphi^5+\alpha^2=4\varphi+5.
\]
Every $c<c_*$ fails consistency $d$ by the eleven-request
construction. Every $c\in(c_*,1]$ has strictly larger exact
robustness, and~\eqref{credit:eq:large-c-floor} excludes every
$c\ge1$ regardless of its consistency. This proves the global
minimum and uniqueness in Theorem~\ref{credit:thm:family}.
Shaped-BANK's exact robustness $\varphi^5$ improves this scalar
minimum by $\alpha^2$.

\section{The exact consistency threshold of SUM-BANK}
\label{sharp:sec:threshold}
SUM-BANK uses the constant cost $\lambda v$ for either prediction.
Its optimal-consistency range follows from the source-measure
theorem, while a separate finite construction proves necessity.

\begin{theorem}[SUM feasible-positive capacity]
\label{sharp:thm:attainment}
For every $\lambda\ge M=\varphi^3$ and feasible complete positive
set $P$, each held interval satisfies
\begin{equation}
 b(J)\le dw(J),\qquad w(P)\le d\,\mathrm{ALG}.
 \label{sharp:eq:ledger-and-feasible}
\end{equation}
Consequently $C(S_\lambda)=d$.
\end{theorem}
\begin{proof}
Apply Theorem~\ref{sg:thm:interface} with
$q_+(v,b)=q(v,b)=\lambda v$. Globally $q\ge dv$ and $q_+\ge Mv$.
For $0\le b\le dv$,
\[
 \varphi v+\alpha b\le2\varphi v<Mv\le\lambda v.
\]
Thus every hypothesis holds with the exact shared fees and
inheritance. Conservation gives the feasible-positive bound.
The finite robustness in Proposition~\ref{sum:prop:robustness}
and Theorem~\ref{thm:lower} give equality for consistency.
\end{proof}

\begin{corollary}[Exact threshold]\label{sharp:cor:exact-endpoint}
For $\lambda>d$, $C(S_\lambda)=d$ if and only if
$\lambda\ge\varphi^3$.
\end{corollary}
\begin{proof}
The preceding theorem proves sufficiency. For $d<\lambda<\varphi^3$,
Proposition~\ref{prop:double-source} gives
\[
 C(S_\lambda)\ge
 \varphi\left(1+\frac{\varphi}{\lambda}
                   +\frac{\varphi^3}{\lambda^2}\right)>d.
\]
The expression decreases strictly in $\lambda$ and equals $d$ at
$\varphi^3$.
\end{proof}

\subsection{A two-source lower bound with eleven requests}
\label{app:double-source}

The next construction gives the converse for the sharp golden-line
threshold. It concerns SUM-BANK on that parameter line.

\begin{proposition}\label{prop:double-source}
For every real $\lambda>d=\varphi^2$,
\begin{equation}
 C(S_\lambda)
 \ge \varphi\left(1+\frac{\varphi}{\lambda}
                         +\frac{\varphi^3}{\lambda^2}\right).
 \label{eq:double-source-bound}
\end{equation}
Eleven-request inputs with a unique optimum and perfect predictions
approach this bound. In particular it exceeds $d$ for
$d<\lambda<\varphi^3$ and equals $d$ at the endpoint in the limit.
\end{proposition}

\begin{proof}
Write $\beta=\varphi$, so $d=\beta+1$ and $\beta^3=2\beta+1$.
Fix $\lambda>d$ and choose
\[
 0<q,\eta,\zeta<\tfrac14,\qquad
 0<\varepsilon<\min\{\tfrac1{100},\tfrac{\beta}{4\lambda}\}.
\]
Set
\[
 x=q-\beta+\varepsilon,\quad
 y=x+\varepsilon-\beta/\lambda,\quad
 m=(q+1+\lambda-\eta)/2,
\]
\[
 B=\lambda+\beta+\beta^3/\lambda-q-\eta-3\varepsilon,
 \qquad u=(1-\zeta)B/\beta.
\]
Reveal the following distinct requests in the displayed order.
\begin{center}
\begin{tabular}{lll}
\toprule
Request & Interval & Prediction \\
\midrule
$Z$ & $[y-d/\lambda+\varepsilon,y+\varepsilon)$ & $0$ \\
$E_1$ & $[y-d/\lambda,y-d/(2\lambda))$ & $1$ \\
$E_2$ & $[y-d/(2\lambda),y)$ & $1$ \\
$Y$ & $[y,y+\beta/\lambda)$ & $1$ \\
$X$ & $[x,x+\beta)$ & $0$ \\
$Q$ & $[q,q+1)$ & $1$ \\
$J$ & $[0,\lambda)$ & $0$ \\
$T$ & $[x+2\varepsilon,q)$ & $1$ \\
$H_1$ & $[q+1,m)$ & $1$ \\
$H_2$ & $[m,\lambda-\eta)$ & $1$ \\
$U$ & $[\lambda-\eta,\lambda-\eta+u)$ & $1$ \\
\bottomrule
\end{tabular}
\end{center}
All lengths are positive. In particular,
$w(T)=\beta-3\varepsilon>0$ and
$w(H_i)=(\lambda-1-q-\eta)/2>0$.
Also $B>\lambda-53/100>147/100$ and
$u>(3/4)(147/100)/2=441/800>\eta$, so $U$ crosses the right
endpoint of $J$. The positive requests in spatial order are
$E_1,E_2,Y,T,Q,H_1,H_2,U$. They are compatible: $Y$ ends at
$x+\varepsilon$, before $T$ starts at $x+2\varepsilon$, and all
other consecutive contacts are at endpoints. Among negative requests,
$Z$ is compatible with $X$ and $J$, whereas $X$ overlaps $J$.
Indeed, $Z$ ends at $x+2\varepsilon-\beta/\lambda<x<0$, while
$X$ ends at $q+\varepsilon>0$.

We now follow the actual acceptance rules, with ordinary acceptance
taking priority and equality accepted. Initially $Z$ is accepted.
Each $E_i$ has length $d/(2\lambda)<\beta/\lambda=w(Z)/\beta$.
Both tests fail, and each assigns its full length to its sole conflict
$Z$. Hence $b(Z)=d/\lambda$. Request $Y$ overlaps $Z$ by
$\varepsilon$ and fails the ordinary test. Both special thresholds
hold at equality:
\[
 w(Y)=\beta/\lambda=w(Z)/\beta=b(Z)/\beta.
\]
Thus $Y$ replaces $Z$ and has ledger
$(d+\beta)/\lambda=\beta^3/\lambda$.
Request $X$ overlaps $Y$ by $\varepsilon$ and replaces it ordinarily
at equality, $w(X)=\lambda w(Y)=\beta$, inheriting that ledger.
Next $Q$ fails the ordinary test, meets the special length test at
equality, and satisfies its ledger test since
$b(X)/\beta=d/\lambda<1=w(Q)$. Therefore
$b(Q)=1+\beta^3/\lambda$.
Request $J$ replaces $Q$ ordinarily at equality, with birth ledger
$b_0(J)=1+\beta^3/\lambda$.

Each of $T,H_1,H_2$ has only $J$ as a current conflict and length
less than $\lambda/\beta$: for $T$, use $\lambda>d=\beta^2$;
for $H_i$, use $w(H_i)<\lambda/2<\lambda/\beta$.
Each fails both ordinary acceptance and the special length test.
Their full lengths are assigned to $J$, leaving ledger $B$ just
before $U$. The special ledger test for $U$ strictly fails because
$u=(1-\zeta)B/\beta<B/\beta$. The ordinary test also fails:
$B<\lambda+2\beta$ implies $u<\lambda/\beta+2<\lambda^2$.
For the last inequality, $\lambda^2-\lambda/\beta-2$ is increasing
for $\lambda>d$ and equals $2\beta>0$ at $\lambda=d$.
Whether or not the special length test fails, the sole-conflict
allocation assigns all of $u$ to $J$. The terminal solution is
$\{J\}$ and $\mathrm{ALG}=\lambda$, while ledger conservation
gives $w(P)=B+u$.

To check perfect predictions, fix a feasible subset of negative
requests and extend it by all compatible positives. This is the best
extension since all weights are positive. Besides the empty subset,
the only feasible choices and their deficits relative to $P$ are
\[
\begin{array}{c|l}
 \{Z\} & \beta/\lambda \\
 \{X\} & 1+\beta/\lambda-3\varepsilon \\
 \{J\} & \beta-q-\eta-3\varepsilon+u \\
 \{Z,X\} & 1+\beta/\lambda-3\varepsilon \\
 \{Z,J\} & \beta/\lambda+\beta-q-\eta-3\varepsilon+u.
\end{array}
\]
All are strictly positive under the parameter restrictions. A proper
subset of $P$ also loses positive weight. Thus $P$ is the unique optimum.
Consequently the exact finite ratio is
\[
 \frac{\mathrm{OPT}}{\mathrm{ALG}}
 =\left(1+\frac{1-\zeta}{\beta}\right)
  \left(1+\frac{\beta}{\lambda}+\frac{\beta^3}{\lambda^2}
             -\frac{q+\eta+3\varepsilon}{\lambda}\right).
\]
Let all four positive perturbations tend to zero through their
permitted ranges, keeping eleven requests. Since $1+1/\beta=\beta$,
this proves the supremum lower bound~\eqref{eq:double-source-bound}.
Subtracting $\beta(1+d/\lambda)$ gives
$\beta(\beta^3-\lambda)/\lambda^2$.
The lower bound exceeds $d$ exactly when
$\lambda^2-d\lambda-d^2<0$, whose positive root is $\beta^3$.
\end{proof}

For an explicit finite witness, take $\lambda=3$, $q=1/100$, and
$\varepsilon=\eta=\zeta=1/1000$. The ratio is
\[
 \frac{2502479+7474021\varphi}{4500000}
 \approx 3.243488669>d.
\]

\section{Joint geometric cuts for aggregate replacement}
\label{sum:app:robustness}

Proportional charging of rejected requests is used in
\citet[Theorems~4.10--4.11]{karavasilis2025}. Here we bound the
charge arriving at a final replacement tree. Its two exterior
charges and the space left between them must be bounded together.
The proof distinguishes an ancestry antichain from a feasible set
of intervals: only the former is required for the intermediate maxima.

\begin{theorem}[Aggregate-gate transfer]\label{sum:thm:transfer}
Let $\beta\ge1$ and $\lambda>\beta$. A causal policy accepts a
request by revoking precisely its current conflicts, retains all
nonconflicts, leaves its held set unchanged on rejection, and makes
no other revocations. Label each acceptance $\mathsf M$ or
$\mathsf S$, with empty-conflict acceptances labelled $\mathsf M$.
Assume that:
\begin{enumerate}
\item An $\mathsf M$ node's children have total length at most
its length divided by $\lambda$.
\item An $\mathsf S$ node has only $\mathsf M$ children, with
total length at most $\beta$ times its length.
\item Every rejected request $I$ has a nonempty conflict set $D$
and satisfies $w(I)<\lambda w(D)$.
\end{enumerate}
Put
\begin{equation}
 q_M=\frac{\lambda+1}{\lambda-\beta},\qquad
 q_S=\frac{\lambda(\beta+1)}{\lambda-\beta}.
 \label{sum:eq:hull-coefficients}
\end{equation}
\begin{equation}
 K(\beta,\lambda)=\max\{\beta(2\lambda+1),\,
 (\beta+1)\lambda+q_S,\,
 \beta\lambda+\beta^2+q_S\}.\label{sum:eq:transfer}
\end{equation}
Every feasible request set $O$ then satisfies
$w(O)\le K(\beta,\lambda)\mathrm{ALG}$.
No correctness or feasibility of the predictions is assumed.
\end{theorem}

\begin{proof}
We first bound each replacement tree and then distribute the optimal requests over the final trees.
\paragraph{Replacement trees and restricted antichains.}
Each accepted identity is a node, made a child of its unique
replacement if revoked. The roots are the final held requests.
Child roots are compatible, and each overlaps its parent in positive
length. An ancestry antichain contains no node and its strict
ancestor. Subtree induction bounds its mass by $h$ at an
$\mathsf M$ root of length $h$, and by $\beta h$ at an
$\mathsf S$ root. If the root is included, the antichain is its
singleton. Otherwise an $\mathsf M$ root gives mass at most
$\beta h/\lambda<h$, while an $\mathsf S$ root sums the
$\mathsf M$ child bounds to at most $\beta h$.

Let $H_J$ be the closure of the union of the accepted intervals in
$J$'s subtree. This union is connected along replacement edges.
The identities $q_M=1+q_S/\lambda$ and $q_S=1+\beta q_M$ bound
the sum of all node lengths, and hence $|H_J|$, by $q_Mw(J)$ or
$q_Sw(J)$ according to the root's type. Also $q_S\ge q_M$.
Different final hulls need not be disjoint.

For arbitrary real $x\le y$, let $L_J(x)$ be the maximum mass of
an ancestry antichain restricted to nodes $V$ with $s(V)<x$.
Define $U_J(y)$ using $f(V)>y$. Empty antichains are allowed. Set
\begin{equation}
 F_J(x,y)=|H_J\cap[x,y]|+\lambda\{L_J(x)+U_J(y)\}.
 \label{sum:eq:joint-cut}
\end{equation}
We prove $F_J(x,y)\le K(\beta,\lambda)w(J)$ for every subtree.
These maxima may come from different times.

\paragraph{An ordinary root.}
Write $J=[a,b)$ and $h=b-a$. Root-omitting antichains have mass
at most $\beta h/\lambda<h$. Consequently
\[
 L_J(x)=h\quad(a<x),\qquad
 L_J(x)\le\beta h/\lambda\quad(a\ge x),
\]
and the analogous right statements hold with $b>y$ and $b\le y$.
If both root endpoint tests hold, the gap lies inside $J$, and
$F_J\le(2\lambda+1)h$. If exactly one holds, the coefficient is
at most $q_M+\lambda+\beta$. It does not exceed
$(\beta+1)\lambda+q_S$, since the difference is
$\beta(\lambda-1)+(q_S-q_M)\ge0$. If neither holds, the
coefficient $q_M+2\beta$ is smaller still because $\lambda>\beta$.

\paragraph{A special root.}
Its ordinary children $C_i$ have lengths $h_i$ and total
$H=\sum_i h_i\le\beta h$. The exact recursions are
\[
 L_J(x)=\max\{h\mathbf1_{a<x},\sum_iL_{C_i}(x)\},\qquad
 U_J(y)=\max\{h\mathbf1_{b>y},\sum_iU_{C_i}(y)\}.
\]
Select maximizing branches, using children when the maximum is
zero. If both selections are the singleton root, the gap lies in
$J$, giving $(2\lambda+1)h$. If one selects the root and the
other selects children, the coefficient is at most
$q_S+(\beta+1)\lambda$.

Suppose both select children. Call $C_k$ a bridge when
$s(C_k)<x\le y<f(C_k)$. At most one bridge exists, including
when $x=y$, because two would overlap in positive length.
If it exists, the entire gap lies in $C_k$ and both its maxima
equal $h_k$. Every other child has at most one eligible root
endpoint, so its two maxima sum to at most
$(1+\beta/\lambda)h_i$. This also covers two ineligible endpoints,
since $2\beta/\lambda<1+\beta/\lambda$. Therefore
\[
 F_J\le(2\lambda+1)h_k+(\lambda+\beta)(H-h_k)
 \le(2\lambda+1)H\le\beta(2\lambda+1)h.
\]
If there is no bridge, the hull bound instead gives
$F_J\le q_Sh+(\lambda+\beta)H
\le(q_S+\beta\lambda+\beta^2)h$.
These cases prove the asserted joint-cut bound. Equality at an
endpoint is consistently assigned to the ineligible case.

\paragraph{From cuts to actual optimal requests.}
Apply Lemma~\ref{sg:lem:charge-transfer} with
$\omega_V=\lambda w(V)\ge w(V)$. Its endpoint maxima are
$\lambda$ times the length maxima used here, so its $F_J$ is
exactly~\eqref{sum:eq:joint-cut}. A rejected $I$ sends
$w(I)w(V)/w(D)<\lambda w(V)$ to each current conflict $V$.
The full weighted antichain bound is at most $\lambda\beta h$,
below $K(\beta,\lambda)h$ since $K\ge\beta(2\lambda+1)$.
The established joint-cut bound therefore proves the theorem.
The empty input and empty $O$ are immediate.
\end{proof}

\begin{lemma}[A single-conflict witness for SUM]
\label{r30:lem:single-conflict-witness}
For every $\lambda>d$, the following family gives SUM-BANK ratio
approaching $(2\lambda+1)\varphi$. The family has one positive
prediction, so permanent filtering leaves it unchanged.
\end{lemma}
\begin{proof}
Put $\alpha=1/\varphi$ and choose
$0<\tau<\min\{(1-\alpha)/2,(\lambda-1)/2\}$. Reveal
\[
 [0,1)_0,\quad[-\lambda+2\tau,\tau)_0,\quad
 [\tau,1-\tau)_0,\quad
 [1-\tau,1+\lambda-2\tau)_0,\quad
 [\tau,\tau+\alpha)_1,
\]
where subscripts are prediction labels. The first request is accepted.
Each later request has exactly one maintained conflict. The three subsequent
negative requests have length below $\lambda$, are rejected, and add
no ledger. The last request passes the special length test at equality
with zero conflicting ledger, leaving $\mathrm{ALG}=\alpha$.
The middle three requests are pairwise compatible and form the unique
optimum: the central one strictly contains and outweighs the last
positive request, while the initial unit interval conflicts with every
other request and has smaller weight than their packing. Their total
is $2\lambda+1-4\tau$. Letting $\tau\downarrow0$ proves the ratio
through finite inputs with positive overlaps.
\end{proof}

\begin{proposition}[SUM robustness]\label{sum:prop:robustness}
For every $\lambda>d$,
\begin{equation}
 \varphi(2\lambda+1)\le R(S_\lambda)
 \le\overline R_S(\lambda):=K(\varphi,\lambda).
 \label{sum:eq:robustness}
\end{equation}
For $\lambda\ge\varphi^3$, these bounds match:
\begin{equation}
 R(S_\lambda)=\varphi(2\lambda+1),\qquad
 R(S_{\varphi^3})=7\varphi+4.
 \label{sum:eq:endpoint-robustness}
\end{equation}
For $3\le\lambda\le\varphi^3$, the upper coefficient is
$d\lambda+\lambda d/(\lambda-\varphi)$.
\end{proposition}
\begin{proof}
Ordinary acceptance gives the aggregate child bound. Special
acceptance has only negative children, each of type $\mathsf M$,
and its length gate gives total child length at most $\varphi w(I)$.
Every rejection fails the ordinary gate. Thus all transfer hypotheses
hold with $\beta=\varphi$, independently of the ledgers and labels.

The first term of $K$ dominates when
$q_S\le(\varphi-1)\lambda+\varphi$ and
$q_S\le\varphi\lambda-1$. At $\lambda=\varphi^3$ the first is
an equality and the second has positive slack; $q_S$ decreases
while both right sides increase. Below that threshold, down to
$3$, the mixed term dominates the first; it exceeds the third
by $\lambda-d>0$. This proves the stated upper coefficients.

The lower bound is supplied by Lemma~\ref{r30:lem:single-conflict-witness}.
\end{proof}

\begin{corollary}[Exact robustness of Revoke-Proportional]
\label{sum:cor:rp}
For every $\lambda>1$, the unchanged Revoke-Proportional algorithm
of \citet[Algorithm~4]{karavasilis2025} satisfies
\begin{equation}
 R(\mathrm{RP}_\lambda)
 =\max\left\{2\lambda+2,\lambda+\frac{2\lambda}{\lambda-1}\right\}
 =\begin{cases}
 \lambda+2\lambda/(\lambda-1),&1<\lambda\le2,\\
 2\lambda+2,&\lambda\ge2.
 \end{cases}
 \label{sum:eq:rp-exact}
\end{equation}
The same formula holds with the permanent prediction filter of
Section~\ref{new:sec:wrapper}. Thus $2$ is the exact cutoff for the
formula $2\lambda+2$. The published consistency bound
$3\lambda/(\lambda-1)$ is unchanged.
\end{corollary}
\begin{proof}
RP has the ordinary aggregate gate. Its special branch requires a
positive label, only negative conflicts, and new length at least their
total. Thus the replacement-tree and charge-conversion definitions
of Theorem~\ref{sum:thm:transfer} apply with $\beta=1$. Every subtree
ancestry antichain has mass at most its root length. In particular,
an eligible singleton root attains the corresponding restricted
maximum $L_J$ or $U_J$.

\paragraph{Coupled one-sided tails.}
For $J=[a,b)$ of length $h$ and $y\ge a$, set
\[
 T_J(y)=|H_J\cap[a,y]|+\lambda U_J(y),\qquad
 t_M=\max\left\{\lambda+1,\frac{\lambda+1}{\lambda-1}\right\},
 \quad t_S=t_M+1.
\]
We claim $T_J(y)\le t_Mh$ at an ordinary root and
$T_J(y)\le t_Sh$ at a special root. The reflected statement
also holds. If $a\le y<b$, the root attains $U_J(y)=h$, and
the bound is $(\lambda+1)h$. If $y\ge b$, the root is
ineligible and $U_J(y)=\sum_iU_{C_i}(y)$. The closed hull
$H_J$ is the union of $[a,b]$ and its child hulls. Each child
starts at $a_i<b\le y$, so
\[
 |H_J\cap[b,y]|
 \le\sum_i|H_{C_i}\cap[b,y]|
 \le\sum_i|H_{C_i}\cap[a_i,y]|.
\]
Consequently $T_J(y)\le h+\sum_iT_{C_i}(y)$. At an ordinary
root the child lengths sum to at most $h/\lambda$, giving
$(1+t_S/\lambda)h\le t_Mh$ because
$(\lambda-1)t_M\ge\lambda+1$. At a special root all children
are ordinary and their lengths sum to at most $h$, giving
$(1+t_M)h=t_Sh$. This proves the tail bounds by induction,
including endpoint equality.

\paragraph{Joint cuts.}
Set $c_M=\lambda+t_M$ and $c_S=c_M+1$. We show that the
functional $F_J(x,y)$ from~\eqref{sum:eq:joint-cut} is at most
$c_Mh$ or $c_Sh$ according to the root type. If both root
endpoints are eligible, the gap is within $J$ and
$F_J\le(2\lambda+1)h\le c_Mh$. If exactly one is eligible,
say $a<x$ and $b\le y$, then $L_J(x)=h$ and
\[
 F_J(x,y)\le\lambda h+|H_J\cap[a,y]|+\lambda U_J(y)
 =\lambda h+T_J(y).
\]
The tail bound gives the desired coefficient for either type.
Reflection handles the other orientation.

If neither endpoint is eligible, both restricted maxima decompose
over children, and the hull union yields
$F_J(x,y)\le h+\sum_iF_{C_i}(x,y)$. At a special root this is
at most $(1+c_M)h=c_Sh$. At an ordinary root it is at most
$(1+c_S/\lambda)h\le c_Mh$, since
$(\lambda-1)c_M\ge\lambda+1$. Equalities in endpoint tests
are ineligible cases; an eligible root may always be selected in
an antichain tie. These cases cover all real $x\le y$.

The charge conversion in Theorem~\ref{sum:thm:transfer} now
gives $R\le c_S$; its whole-hull covering case has charge at
most $\lambda h\le c_Sh$. Algebraically,
$c_S=\max\{2\lambda+2,\lambda+2\lambda/(\lambda-1)\}$.

\paragraph{Lower bound for $\lambda\ge2$.}
For $0<\varepsilon<1/8$, reveal
\[
 \begin{array}{c|c|c}
 K&[1-\varepsilon,2-\varepsilon)&0\\
 G&[\varepsilon,2-2\varepsilon)&0\\
 U&[2-2\varepsilon,2-2\varepsilon+\lambda-\varepsilon)&0\\
 J&[0,1)&1\\
 L&[-\lambda+2\varepsilon,\varepsilon)&0.
 \end{array}
\]
The trace is ordinary acceptance, two rejections, special acceptance,
and rejection. In particular $w(G)=2-3\varepsilon<\lambda$,
also at $\lambda=2$. The final value is one, and the compatible
comparison set $\{L,G,U\}$ has value
$2\lambda+2-5\varepsilon$. Let $\varepsilon\downarrow0$.

\paragraph{Alternating-chain lower bound.}
For any $\lambda>1$, integer $n\ge0$, and
$0<\varepsilon<1/2$, define candidate lengths
\[
 \ell_{2k}=\ell_{2k+1}=\lambda^{k-n}\quad(0\le k\le n),
 \qquad a_0=0,\quad a_{i+1}=a_i+(1-\varepsilon)\ell_i.
\]
Reveal $A_i=[a_i,a_i+\ell_i)$ in increasing $i$, labelling even
indices negative and odd indices positive. Immediately after $A_i$,
reveal the negative filler
\[
 F_i=[a_i+\varepsilon\ell_i,\ a_i+(1-\varepsilon)\ell_i).
\]
Each even candidate after the first ordinarily replaces its sole
predecessor at growth equality. Each odd candidate specially replaces
its negative predecessor at length equality. Consecutive candidates
overlap in positive length, every filler is rejected, and the final
held set is the singleton $A_{2n+1}$ of length one.
Writing $b=f(A_{2n+1})$, finally reveal the negative request
$U=[b-\varepsilon,b-\varepsilon+\lambda-\varepsilon)$.
It overlaps the held interval and is rejected.

The fillers are compatible: $F_i$ ends at $a_{i+1}$ and the
next filler starts strictly to its right. The final filler ends
where $U$ starts. They therefore give the feasible comparison value
\[
 (1-2\varepsilon)2\sum_{k=0}^n\lambda^{k-n}
       +\lambda-\varepsilon.
\]
Letting first $\varepsilon\downarrow0$ and then $n\to\infty$
gives $R\ge\lambda+2\lambda/(\lambda-1)$. All executions in
the family are finite and use positive-length overlaps. For
$1<\lambda<2$ this limit exceeds $2\lambda+2$ by
$(2-\lambda)(\lambda+1)/(\lambda-1)>0$; for $\lambda\ge2$
the preceding five-request family gives the larger branch.

Finally, consecutive positive carriers in the chain, of lengths
$a$ and $\lambda a$, have gap
$[\lambda-\varepsilon(1+\lambda)]a>0$. Thus all positive
predictions are feasible and the permanent filter changes no label.
The five-request family has just one positive. The upper bound allows
arbitrary effective labels, proving the filtered formula as well.
\end{proof}

\section{Prediction filtering and zero-error inputs}\label{new:sec:wrapper}\label{new:app:wrapper}
Fix an optimum $O$, the original positive set $P$, and the set $N_O(F)$ of optimum requests conflicting with $F$. The error of \citet{karavasilis2025} is
\begin{equation}
 \eta_K(P,O)=w(O\setminus P)+\sum_{F\in P\setminus O}[w(N_O(F))-w(F)].
 \label{new:eq:eta}
\end{equation}
Optimality makes each summand nonnegative; a false positive of equal weight to its optimal conflicts contributes zero.

\begin{proposition}[Zero-error limitation]\label{new:prop:zero-error}
No finite $k\ge0$ makes $\mathrm{OPT}\le a\,\mathrm{ALG}+k\eta_K(P,O)$ valid for all inputs of raw Shaped-BANK when $a<4\varphi+1$, or raw scalar Credit with fixed $c>0$ when $a<4\varphi+3$. For SUM-BANK at $\lambda>d$, the obstruction holds for $a<2\lambda+1$. Finite witnesses with pairwise distinct interval geometries have $\mathrm{ALG}=1$, $\mathrm{OPT}>a$ and zero error for every optimum.
\end{proposition}
Distinct overlapping equal-length alternatives inflate the ledger while contributing zero error; their complete positive set is infeasible.

The filter maintains a permanent feasible history $P'$. An originally positive request receives effective label $1$ precisely when it is compatible with all of $P'$; insert it immediately, before running the selected policy. Every other effective label is $0$. Every request reaches the algorithm, and the effective label is fixed for all decisions and ledger updates. Requests remain in $P'$ after rejection or revocation. This is the Naive selection rule of \citet{karavasilis2025}, used as a filter. Denote the compositions by $T_*^F$, $T_c^F$, and $S_\lambda^F$.

\begin{corollary}[Simultaneous guarantees]\label{new:cor:wrapper}
For original labels and fixed optimum $O$, Shaped-BANK satisfies
\begin{equation}
 \mathrm{ALG}(T_*^F)\ge
 \max\left\{\frac{\mathrm{OPT}}{\varphi^5},
             \frac{\mathrm{OPT}-\eta_K(P,O)}d\right\},
 \label{sg:eq:error}
\end{equation}
with $C=d$ and exact $R=\varphi^5$.
The same simultaneous formula holds for scalar Credit
$T_c^F$, $c\in[\alpha^3,1]$, with robustness $K(c)$ from
Theorem~\ref{credit:thm:family}; its minimum is $4\varphi+5$.
For SUM at $\lambda\ge\varphi^3$,
replace the robustness denominator by $(2\lambda+1)\varphi$.
These policies also have $C=d$ and exact stated robustness.
\end{corollary}

\subsection{The permanent-history filter and its quality}
\label{new:app:filter}

We analyze the permanent history $P'$ defined above, produced by the
\textsc{Naive} rule~\citep[Algorithm~1]{karavasilis2025}. Original labels
determine $\eta_K(P,O)$; effective labels determine every policy decision
and ledger increment.

\begin{lemma}[Quality of the known \textsc{Naive} filter]
\label{new:lem:naive-quality}
For every finite request stream and every fixed optimum $O$, the
permanent history produced by the filter in Section~\ref{new:sec:wrapper} is feasible,
is a subset of $P$, and satisfies
\begin{equation}
\label{new:eq:naive-quality}
w(P')\ge\mathrm{OPT}-\eta_K(P,O).
\end{equation}
\end{lemma}

The inequality is the \textsc{Naive} guarantee of
\citet[Theorem~3.5]{karavasilis2025}.

\begin{proof}
Each insertion preserves compatibility with every earlier insertion,
so the entire set $P'$ is feasible. Only originally positive requests
are inserted, giving $P'\subseteq P$.
For $F\in P\setminus O$, set
\[
e_F=w(N_O(F))-w(F).
\]
The set $(O\setminus N_O(F))\cup\{F\}$ is feasible. Optimality of $O$
therefore gives
\[
w(O)\ge w(O)-w(N_O(F))+w(F),
\]
and hence $e_F\ge0$ for every such $F$.

Partition the positive predictions and their filtered subset as
\[
A=P\cap O,\quad A'=P'\cap O,\qquad
B=P\setminus O,\quad B'=P'\setminus O.
\]
If $I\in A\setminus A'$, it was blocked at its arrival by some
earlier request $F\in P'$. Since $I\in O$ and $O$ is feasible,
that blocker cannot belong to $O$; thus $F\in B'$. Consequently,
\begin{equation}
\label{new:eq:naive-cover}
A\setminus A'\subseteq\bigcup_{F\in B'}N_O(F),
\qquad
w(A\setminus A')\le\sum_{F\in B'}w(N_O(F)).
\end{equation}
Overlapping neighborhoods may count a request more than once, which
only enlarges the right-hand side.
Using $B'\subseteq B$ and the nonnegativity of every $e_F$, we obtain
\begin{align*}
w(P')
 &=w(A)-w(A\setminus A')+w(B')\\
 &\ge w(A)-\sum_{F\in B'}w(N_O(F))+\sum_{F\in B'}w(F)\\
 &=w(A)-\sum_{F\in B'}e_F\\
 &\ge w(A)-\sum_{F\in B}e_F\\
 &=w(O)-w(O\setminus P)-\sum_{F\in P\setminus O}e_F\\
 &=\mathrm{OPT}-\eta_K(P,O).
\end{align*}
This also accounts for false positives that the filter never inserts.
\end{proof}

\subsection{Composition with the feasible-set guarantee}\label{new:app:composition}
\begin{proof}[Proof of Corollary~\ref{new:cor:wrapper}]
Let $A$ be a policy with a feasible-positive guarantee
$w(P')\le d\,\mathrm{ALG}_A$ and an arbitrary-label robustness bound
$\mathrm{OPT}\le R_A\mathrm{ALG}_A$. The filter passes every request
to $A$ with an effective label fixed on arrival. Its complete effective
positive set is the feasible set $P'$, while the offline optimum is
unchanged. Lemma~\ref{new:lem:naive-quality} therefore gives
\[
 \mathrm{ALG}_{A^F}\ge
 \max\left\{\frac{\mathrm{OPT}}{R_A},
             \frac{\mathrm{OPT}-\eta_K(P,O)}d\right\}.
\]
This remains valid when the second term is negative. The required
policy-specific bounds are as follows.
\begin{center}\small
\begin{tabular}{@{}llll@{}}\toprule
Policy and range & $R_A$ & Feasible-positive bound & Robustness bound\\\midrule
Shaped-BANK & $\varphi^5$ & Theorem~\ref{sg:thm:interface} & Theorem~\ref{sg:thm:exact}\\
Credit, $c\in[\alpha^3,1]$ & $K(c)$ & Theorem~\ref{credit:thm:family} & Theorem~\ref{credit:thm:family}\\
SUM, $\lambda\ge\varphi^3$ & $(2\lambda+1)\varphi$ & Theorem~\ref{sharp:thm:attainment} & Proposition~\ref{sum:prop:robustness}\\
\bottomrule
\end{tabular}
\end{center}
In particular, the Credit specialization is
\begin{equation}
 \mathrm{ALG}(T_c^F)\ge
 \max\left\{\frac{\mathrm{OPT}}{K(c)},
             \frac{\mathrm{OPT}-\eta_K(P,O)}d\right\}.
 \label{credit:eq:error}
\end{equation}
For perfect original predictions, every positive belongs to a feasible
optimum, so filtering changes no label. Thus the filtered and raw
consistency suprema agree and equal $d$. To retain exact robustness,
one must also check the matching witnesses. Shaped-BANK's family in
Proposition~\ref{sg:prop:floor} has one positive request. Both Credit
families in Theorem~\ref{credit:thm:family} have compatible positive
sets. The SUM family in
Lemma~\ref{r30:lem:single-conflict-witness} has one positive request.
All therefore pass unchanged, proving equality with the listed $R_A$.
For Credit this includes $4\varphi+5$ at $c=\alpha^3$ and
$6\varphi+3$ at $c=\varphi/2$.
\end{proof}
Standalone Naive returns $P'$ and has the stronger coefficient-one error bound but unbounded robustness: a single negative-labelled request gives it zero reward. At the optimal-consistency configurations, the compositions add finite robustness while attaining the smallest deterministic consistency compatible with it.

\subsection{A zero-error obstruction with distinct intervals}
\label{r2:app:zero-error}

\begin{proposition}[Uniform zero-error obstruction]
\label{r2:prop:zero-error}
For SUM-BANK at every $\lambda>d$, with $\alpha=1/\varphi$ and $\beta=\varphi$, and
every $c<2\lambda+1$, there is a finite input with pairwise
distinct interval geometries, $\mathrm{ALG}=1$,
$\mathrm{OPT}>c$, and $\eta_K(P,O)=0$ for every optimum $O$.
Therefore no finite $k\ge0$ makes
$\mathrm{OPT}\le c\,\mathrm{ALG}+k\eta_K(P,O)$ valid for all
inputs of that unfiltered configuration.
\end{proposition}

\begin{proof}
Every arrival below has at most one maintained conflict. Choose $\delta>0$ with
$\delta<1/4$, $\lambda-\delta>\max\{\alpha,\delta\}$, and
$2\lambda+1-5\delta>\max\{c,1\}$. Choose an integer
$n>1/\alpha$, set $h=(1-3\delta)/n$, and choose an integer $m$
large enough that
\[
 B=mnh=m(1-3\delta)>\beta(\lambda-\delta).
\]
For $j=1,\ldots,n$ let
$a_j=\delta+(j-1)(h+\delta/n)$ and define $m$ positive
requests in group $j$ by
\[
 I_{j,k}=\left[a_j+\frac{k\delta}{mn},\,
                  a_j+\frac{k\delta}{mn}+h\right),
 \qquad k=0,\ldots,m-1.
\]
Within a group, all starts are different and their total span
is less than $\delta/n<h$, so all intervals overlap in positive
length. Consecutive groups are separated by a gap
$\delta/(mn)>0$. The first group starts at $\delta$ and the
last endpoint is $1-\delta-\delta/(mn)<1-\delta$.
Thus all central intervals have distinct geometries and lie in
$[\delta,1-\delta)$.

Offer the negative anchor $J=[0,1)$, then all the $I_{j,k}$
in lexicographic $(j,k)$ order, and finally the positive requests
\[
 L=[2\delta-\lambda,\delta),\qquad
 U=[1-\delta,1+\lambda-2\delta).
\]
Each central request has length $h<\alpha<\lambda$, so it
fails both ordinary acceptance and the special length test.
Its sole conflict is $J$, which receives its full length.
After all central requests the ledger is $B$.
Each boundary request has length $\lambda-\delta<\lambda$
and fails the special ledger test because the current ledger
is at least $B>\beta(\lambda-\delta)$. Both reject and the
final reward is one. The strict overlaps with $J$ have length
$\delta$, and no requests in this construction are geometrically
identical.

Every other request conflicts with $J$. A feasible set avoiding
$J$ contains at most one interval of each central group, and
can contain $L,U$ together with any such selection. All weights
are positive, and
\[
 \mathrm{OPT}=2(\lambda-\delta)+nh
             =2\lambda+1-5\delta>\max\{c,1\}.
\]
Consequently the optima are exactly $L,U$ and one representative
per group. Every optimum is contained in $P$. An unselected
positive conflicts with exactly one of its members, the chosen
representative of its own group, which has the same length $h$.
Every term of \eqref{new:eq:eta} is therefore zero, for every
optimum. Substitution proves the assertion for all finite $k$.
\end{proof}

The construction has multiple optima but no coincident intervals.
The complete positive set contains overlapping alternatives, so the predictions are not perfect. At $\lambda=\varphi^3$
it excludes every error-bound coefficient below $4\varphi+3$.

\begin{corollary}[Credit-BANK has the same zero-error limitation]
\label{credit:cor:zero-error}
For every fixed $c>0$, every $a<4\varphi+3$, and every finite
$k\ge0$, raw $T_c$ fails
$\mathrm{OPT}\le a\,\mathrm{ALG}+k\eta_K$ on some finite input
with distinct interval geometries.
\end{corollary}
\begin{proof}
Use the construction above with $\lambda=M=\varphi^3$ and
choose $\delta$ so that $2M+1-5\delta>a$. Increase $m$ until
\[
 B=m(1-3\delta)>\max\{\varphi(M-\delta),\varphi/c\}.
\]
Initially the unit negative anchor costs $d$, and every central
length satisfies $h<\alpha<d$, so all central requests reject.
The inequality $cB>\varphi=M-d$ brings the anchor cost to its
cap $M$ before $L,U$ arrive. Each boundary request then fails
the ordinary and ledger tests, and subsequent deposits preserve
the cap. The reward is one and the same optimum calculation
gives $\mathrm{OPT}=2M+1-5\delta>a$ and zero error for every
optimum.
\end{proof}

\begin{corollary}[Shaped-BANK's raw zero-error limitation]
\label{sg:cor:zero-error}
For every $a<4\varphi+1$ and every finite $k\ge0$, raw Shaped-BANK
fails $\mathrm{OPT}\le a\,\mathrm{ALG}+k\eta_K(P,O)$ on a finite
input with pairwise distinct interval geometries,
$\mathrm{ALG}=1$, $\mathrm{OPT}>a$, and
$\eta_K(P,O)=0$ for every optimum $O$.
\end{corollary}
\begin{proof}
Use Proposition~\ref{r2:prop:zero-error}'s overlapping-group
construction with $\lambda=2\varphi$. Choose $\delta>0$ small
enough that $4\varphi+1-5\delta>a$. The initial negative unit
anchor costs $d$, and all central requests have length
$h<\alpha<d$. They therefore reject ordinarily and by the
special length test, each depositing its whole length.
The finite multiplicity $m$ gives
\[
 B=m(1-3\delta)>\varphi(2\varphi-\delta)>d.
\]
The shaped cost has reached its cap $2\varphi$ before the two
boundary arrivals. Each boundary length is $2\varphi-\delta$,
below both its ordinary threshold and the ledger threshold.
Both reject, their deposits preserve the cap, and the reward
is one. The optimal sets are the two boundary requests and one
equal-length representative of each central group. They have
weight $4\varphi+1-5\delta>a$, and every unchosen positive has
exactly one equal-weight conflict in each optimum. Thus
$\eta_K(P,O)=0$ for every optimum, as required.
\end{proof}
The two policies' different caps give different zero-error
coefficients. For raw Shaped-BANK, $4\varphi+1>d$ already
excludes the optimal-consistency coefficient in such a bound.

\subsection{Exact implementation of the special tests}\label{new:app:rational-witness}
For rational $w>0$ and $X\ge0$,
\[
 w\ge X/\varphi\quad\Longleftrightarrow\quad w^2+wX-X^2\ge0.
\]
For $X>0$, divide by $X^2$ and use that $t^2+t-1$ is strictly increasing for $t\ge0$ with root $1/\varphi$; for $X=0$, both conditions hold. Taking $X=W$ or $X=B$ implements the two special tests without rounding the golden ratio.

\subsection{Precision, recall, and confidence scores}\label{r12:sec:precision}
Classification summaries and conflict-weighted error answer different
questions. Fix an optimum $O$ of positive weight and write $T=P\cap O$,
$\mathcal F=P\setminus O$. The length-weighted recall is $r_w=w(T)/w(O)$;
when $w(P)>0$, length-weighted precision is $p_w=w(T)/w(P)$.
For each false-positive request $F\in\mathcal F$ define
$\kappa_F=w(N_O(F))/w(F)\ge1$. The error definition gives exactly
\[
 \frac{\eta_K(P,O)}{\mathrm{OPT}}
 =1-r_w+\sum_{F\in P\setminus O}
        \frac{w(F)}{\mathrm{OPT}}(\kappa_F-1).
\]
If $p_w>0$ and every false positive has $\kappa_F\le\kappa$,
this implies the conditional bound
\[
 \frac{\eta_K}{\mathrm{OPT}}
 \le 1-r_w+(\kappa-1)r_w\frac{1-p_w}{p_w}.
\]
There is no universal bound on $\kappa$ in our model: a short
false-positive interval can conflict with a much longer optimal interval.

Indeed, fix $0<\epsilon<1/10$ and an integer $m\ge2$. Let $O$
contain $[0,1)$ and $m^2-1$ mutually compatible intervals to its
right with total length $\epsilon$. Add $m$ distinct intervals of
length $\epsilon/m^2$ inside $[0,1)$, and label every request positive.
These added intervals have combined length at most $\epsilon/m<1$,
so $O$ is the unique optimum. Recall is one, count precision is
$m^2/(m^2+m)\to1$, and weighted precision is
$(1+\epsilon)/(1+\epsilon+\epsilon/m)\to1$. Nevertheless
\[
 \frac{\eta_K}{\mathrm{OPT}}
 =\frac{m-\epsilon/m}{1+\epsilon}\longrightarrow\infty.
\]
Thus even simultaneously high count and weighted precision and perfect
recall do not alone make this error guarantee informative. The observed
conflict costs, not only the confusion matrix, must be reported.

A predictor may output a score $s(I)\in[0,1]$ and fix its bit at arrival
using $p(I)=\mathbf1\{s(I)\ge\tau\}$. For every threshold fixed without
future test outcomes, the deterministic robustness theorem still applies
to the resulting bits, and the filtered error theorem applies to their
actual $\eta_K$. Calibration or threshold selection alone does not bound the
conflict multipliers above.

\section{Randomized guarantees and pathwise protection}
\label{w02r4:sec:randomized-scope}

For the expected-gain comparison, allow randomized algorithms with the same
online information and feasible, irrevocable rejection and revocation
rules as in Section~\ref{sec:model}. Fix each complete nonempty finite
labelled input $\sigma$ independently of the algorithm's random coins.
The input includes the reveal order, all intervals, all labels, and its
termination point; its length and future requests are not disclosed to
the algorithm. Define
\[
 R_{\mathrm{exp}}(A)
 =\sup_{\sigma}\frac{\mathrm{OPT}(\sigma)}
                       {\mathbb E[\mathrm{ALG}_A(\sigma)]},
 \qquad
 C_{\mathrm{exp}}(A)
 =\sup_{\sigma:P\text{ is optimal}}
       \frac{\mathrm{OPT}(\sigma)}
            {\mathbb E[\mathrm{ALG}_A(\sigma)]}.
\]
A zero denominator gives ratio $+\infty$, and no additive term is
allowed. For deterministic algorithms, $R_{\mathrm{exp}}$ and $C_{\mathrm{exp}}$ coincide with the ratios in Definition~\ref{def:performance}.

\subsection{The exact initial-mixture curve}
\label{r12:sec:mixture-exact}
\begin{proposition}[An initial mixture of two existing algorithms]
\label{w02r4:prop:mixture}
\label{r12:prop:mixture-exact}
Let $r_{\mathrm{LR}}=2\varphi+1$ and $0\le\pi<1$. At initialization,
$M_\pi$ selects \textnormal{\textsc{Naive}} with probability $\pi$ and LR with
probability $1-\pi$, and then follows the selected algorithm for the
whole input. Put
\[
 \overline C_\pi
 =\frac{1}{\pi+(1-\pi)/r_{\mathrm{LR}}},
 \qquad
 \overline R_\pi=\frac{r_{\mathrm{LR}}}{1-\pi}.
\]
Then
\[
 C_{\mathrm{exp}}(M_\pi)=\overline C_\pi,
 \qquad
 R_{\mathrm{exp}}(M_\pi)=\overline R_\pi<\infty.
\]
Eliminating $\pi$ gives $R=2\varphi C/(C-1)$ for
$1<C\le\varphi^3$. This is the exact curve of this mixture family.
For arbitrary original labels and any one fixed optimum $O$,
\begin{align}
 \mathbb E[\mathrm{ALG}_{M_\pi}]
 &\ge \pi\max\{0,\mathrm{OPT}-\eta_K(P,O)\}
      +\frac{1-\pi}{r_{\mathrm{LR}}}\mathrm{OPT} \notag\\
 &=\max\left\{
       \frac{\mathrm{OPT}}{\overline R_\pi},
       \frac{\mathrm{OPT}}{\overline C_\pi}
          -\pi\eta_K(P,O)\right\}.
 \label{w02r4:eq:mixture-error}
\end{align}
Here $P$ is the original positive set and $\eta_K(P,O)$ is the error
in~\eqref{new:eq:eta}.
\end{proposition}

\begin{proof}
The known guarantees used here are
$\mathrm{ALG}_{\mathrm{Naive}}
 \ge\max\{0,\mathrm{OPT}-\eta_K(P,O)\}$ and
$\mathrm{ALG}_{\mathrm{LR}}\ge\mathrm{OPT}/r_{\mathrm{LR}}$
\citep[Theorem~3.5 and Section~4]{karavasilis2025}.
The first also follows from Lemma~\ref{new:lem:naive-quality}, since
standalone \textsc{Naive} returns its permanent feasible history.
For each fixed input, linearity of expectation gives
\[
 \mathbb E[\mathrm{ALG}_{M_\pi}]
 =\pi\,\mathrm{ALG}_{\mathrm{Naive}}
   +(1-\pi)\,\mathrm{ALG}_{\mathrm{LR}}.
\]
Substituting the two guarantees proves the first line of
\eqref{w02r4:eq:mixture-error}. Its second line is the identity
$a+\pi\max\{0,x\}=\max\{a,a+\pi x\}$ for $\pi\ge0$,
with $a=(1-\pi)\mathrm{OPT}/r_{\mathrm{LR}}$ and
$x=\mathrm{OPT}-\eta_K(P,O)$. In particular, it remains valid for
$\pi=0$ and for $\eta_K(P,O)>\mathrm{OPT}$.
Under perfect predictions, $P=O$ for some optimum, and
\textsc{Naive} accepts every member of $P$ and no other request.
It therefore gains $\mathrm{OPT}$. This yields the stated uniform
consistency bound; the LR term yields the uniform robustness bound.
For matching lower bounds, take finite unlabelled inputs whose LR ratios
tend to its exact worst-case ratio $r_{\mathrm{LR}}$. Labelling one
optimum positive gives Naive reward OPT without changing LR, and hence
mixture ratios tending to $\overline C_\pi$. Labelling every request
negative instead gives Naive reward zero and ratios tending to
$\overline R_\pi$.
\end{proof}

\subsection{A fixed-input lower bound for randomized algorithms}
\label{r12:sec:randomized-lower}

A different, elementary bound
shows a uniform separation from consistency one at every fixed robustness
budget. Define
\[
 L_{\mathrm{exp}}(B)
 =\frac{1}{2\sqrt{1-B^{-1}}-(1-B^{-1})},\qquad B\ge4/3.
\]

\begin{proposition}[A two-input randomized lower bound]
\label{r12:prop:oblivious-lower}
Every randomized algorithm in the fixed-input expected-gain model has
$R_{\mathrm{exp}}\ge4/3$. If $R_{\mathrm{exp}}\le B<\infty$, then
\[
 C_{\mathrm{exp}}\ge L_{\mathrm{exp}}(B)>1.
\]
Both inequalities use inputs with only three requests. In particular this
lower bound is valid for an oblivious adversary.
\end{proposition}

\begin{proof}
Fix $v\ge1$. The following two inputs have the same labelled prefix:
\[
 A=[0,1),\ p(A)=1;\qquad
 N=[1/2,1/2+v),\ p(N)=0.
\]
The first input ends with $D=[1,1+v)$ labelled positive. The second ends
with $L=[-1/2,1/2)$ labelled negative. All conflicts have positive
overlap. Both optima have weight $v+1$: the first is uniquely $\{A,D\}$
and has perfect predictions; the second is uniquely $\{N,L\}$ and
does not. These inputs and their stopping points are fixed before the
algorithm draws its coins.

Let $p$ be the probability that $A$ is still held after the common
prefix. Coupling the two executions on the same random tape gives the
same value of $p$. If $A$ is held, the respective final rewards are at
most $v+1$ and $1$. If it is not held, no recall permits it to return,
and the rewards are at most $v$ and $v+1$. These bounds also cover an
empty held set and voluntary revocation. Consequently
\[
 \mathbb E[\mathrm{ALG}_{+}]\le v+p,
 \qquad
 \mathbb E[\mathrm{ALG}_{-}]\le v+1-vp.
\]
For $v=1$, their sum is at most $3$, whereas both optima are $2$;
therefore robustness is at least $4/3$.

Write $c=C_{\mathrm{exp}}$ and assume $R_{\mathrm{exp}}\le B$.
The consistency guarantee on the first input and robustness guarantee
on the second imply
\[
 \frac{v+1}{c}-v\le p
 \le\frac{(v+1)(1-B^{-1})}{v}.
\]
Thus, for every $v\ge1$,
\begin{equation}
 \frac1c\le \frac{v}{v+1}+\frac{1-B^{-1}}{v}.
 \label{r12:eq:randomized-two-input}
\end{equation}
Set $t=\sqrt{1-B^{-1}}\in[1/2,1)$ and
$v=t/(1-t)\ge1$. The right side becomes $2t-t^2$, proving the
claimed bound. Since $t<1$, this number is strictly less than one.
\end{proof}

The bound is tight for this two-input family, including its stopping
prefixes. At the displayed $v=t/(1-t)$, accept $A$, retain it at $N$
with probability $t$, and otherwise replace it by $N$. At $D$, take
$A+D$ when $A$ survived and obtain $v$ otherwise. At $L$, retain $A$
when it survived and take $N+L$ otherwise. The two final expected
rewards are $v+t$ and $v+1-vt$, giving exactly
$(C_{\mathrm{exp}}^{\mathcal F},R_{\mathrm{exp}}^{\mathcal F})=(L_{\mathrm{exp}}(B),B)$ for the restricted family $\mathcal F$ of these two inputs and their two distinct nonempty proper prefixes. The prefix after $N$ has expected
reward $2t$ and ratio $v/(2t)=1/(2(1-t))\le B$; the first prefix
has ratio one.

\begin{corollary}[An unattained randomized consistency infimum]
\label{w02r4:prop:randomized-infimum}
Among randomized algorithms with finite $R_{\mathrm{exp}}$, the infimum of
$C_{\mathrm{exp}}$ is one and is not attained.
\end{corollary}
\begin{proof}
Every feasible reward is at most OPT. The initial mixtures have finite
robustness and consistency tending to one as $\pi\uparrow1$, while
Proposition~\ref{r12:prop:oblivious-lower} excludes equality at every finite budget.
\end{proof}

\subsection{What changes under adaptive or runwise guarantees}
\label{r12:sec:adaptive-scope}

An adaptive-offline adversary chooses each next labelled request and
the stopping point from the observed past requests and retained sets,
without seeing future random bits, and compares the result with the
offline optimum of the realized history. Its expected robustness
criterion is
\[
 \mathbb E[\mathrm{OPT}(\sigma)]
 \le B\,\mathbb E[\mathrm{ALG}(\sigma)]
\]
for every such finite adversary. Here $\sigma$ is random.

\begin{proposition}[The initial mixture is not adaptively robust]
\label{r12:prop:adaptive-mixture}
For every $0<\pi<1$, the initial mixture $M_\pi$ has infinite
expected robustness against adaptive-offline adversaries.
\end{proposition}

\begin{proof}
Offer one negative interval $[0,1)$. Naive rejects it and LR accepts
it, so the observed retained set reveals the chosen algorithm. If it
was accepted, stop. Otherwise offer one more negative interval
$[2,2+H)$ and stop. Naive rejects that request as well. Thus
\[
 \mathbb E[\mathrm{ALG}]=1-\pi,
 \qquad
 \mathbb E[\mathrm{OPT}]=1+\pi H.
\]
The ratio is unbounded as finite $H$ increases. Each adversary uses
at most two requests and no future coin information.
\end{proof}

Even on a fixed one-request negative input, the mixture has zero reward with probability $\pi$. Hence its expected-gain bound cannot
be interpreted as a finite per-run ratio with failure probability less than $\pi$.

\subsection{Finite-tree transfer of the lower bounds}\label{r39:sec:finite-tree-transfer}
The deterministic lower bounds do extend to randomized algorithms
required to satisfy both stated ratios almost surely on every fixed
input. To see this, take a finite deterministic adversary tree proving
that a pair $(c,B)$ is impossible. It has finitely many terminal input
sequences. If both ratios held almost surely on each fixed sequence,
the intersection of those finitely many probability-one events would
still have probability one. A tape in that intersection defines a
deterministic causal rule that survives the entire tree, a
contradiction. This argument does not extend to ratios of expected
gains: distinct leaves can offset one another in expectation.

Ben-David et al.'s adaptive-offline derandomization theorem addresses
a single comparison criterion. With two criteria, perfect and imperfect terminal histories carry different target factors. Adaptive robustness covers every adversary, whereas adaptive consistency covers adversaries whose every terminal history has perfect predictions. These separate requirements do not supply the mixed-terminal comparison needed to transfer the deterministic frontier to expected adaptive guarantees.

\subsection{Almost-sure robustness with expected consistency}
\label{r37:sec:hybrid}\label{r38:sec:hybrid}
Fix the complete input, including termination, independently of the coins.
The hybrid requirements are
\[
 \Pr\{\mathrm{OPT}(\sigma)\le B\mathrm{ALG}(\sigma)\}=1
 \quad\text{for every fixed finite labelled }\sigma,
\]
and $\mathrm{OPT}\le C_{\rm exp}\mathbb E[\mathrm{ALG}]$ on perfectly
labelled inputs, without an additive term. The finite-tree transfer in Appendix~\ref{r39:sec:finite-tree-transfer} requires both guarantees almost surely.

\begin{proposition}[A necessary hybrid budget bound]
\label{r38:prop:hybrid-simple}
Every almost-surely $B$-robust algorithm, for finite $B>1$, satisfies
\[
 C_{\mathrm{exp}}\ge 2+\frac1{B-1}.
\]
\end{proposition}

\begin{proof}
Put $s=1-2\delta$. Present a negative unit interval $J=[0,1)$,
then positive tiles partitioning $[\delta,1-\delta)$, each of length
less than $1/B$, then a positive interval $L=[\delta-t,\delta)$.
Choose $t>\delta$, $s+t>1$, and $t<s/(B-1)$. The input either ends
there or continues with the positive unit interval
$U=[1-\delta,2-\delta)$. Both complete positive sets are optimal,
of weights $s+t$ and $s+t+1$ respectively.

At the first tile on which the algorithm abandons $J$, its available
reward is at most that tile's length, whereas the prefix optimum is
at least one. Almost-sure robustness on each of these finitely many
prefixes therefore forces it to keep $J$ and reject all tiles.
On the fixed input ending at $L$, accepting $L$ would give
$(s+t)/t>B$, so this action has probability zero. By causality the
same holds on the fixed longer input. There its final reward is at
most one, since all tiles and $L$ are irrecoverable and $J$ conflicts
with $U$. Hence $C_{\mathrm{exp}}\ge s+t+1$. Let $\delta\downarrow0$
and $t\uparrow s/(B-1)$, retaining the strict inequalities.
\end{proof}

The bound is exact for the original single-anchor/two-terminal family when $B\ge3$. For $0<t\le1$, accepting $L$
with probability $p$ gives maximal expected rewards $1-p+pt$ and
$1+pt$. When both actions are legal, $p=1/(1+s)$ equalizes the two
ratios at $1+s$. Below $t=s/(B-1)$ only rejection is legal, giving
$1+s+t$; for $t\ge1$, accepting always gives a ratio at most $1+s$.
These responses protect every prefix of this family for $B\ge3$.
Taking suprema gives $2+1/(B-1)$. At the original $t=\alpha$, the balanced ratios tend to two.

\paragraph{Limits of mixing the stated deterministic policies.}
Let $\alpha=1/\varphi$ and $d=\varphi^2$.
There is a common perfect-input obstruction to initially mixing
SG, LR, and the stated SUM/Credit policies. Take $0<t<\alpha$ and
$0<\delta<\min\{t,1/4\}$. Start with a negative unit anchor $[0,1)$,
then positive tiles partitioning $[\delta,1-\delta)$,
each shorter than $\alpha$. Next issue
$L=[\delta-t,\delta)$ and $U=[1-\delta,2-\delta)$.
All these policies retain the anchor until $U$;
after $U$ their gain is at most one. The positives tile the whole
request hull, of length $2-2\delta+t\to d$ as
$t\uparrow\alpha$ and $\delta\downarrow0$.
The same input therefore precludes consistency below $d$ for every
initial mixture of these policies. The strict inequality $t<\alpha$
removes the complementarity present at $t=\alpha$.

\begin{proposition}[Sharp guarantee for a safe action family]
\label{r38:prop:random-action-envelope}
Consider a policy that accepts empty-conflict arrivals, accepts whenever
$w(I)\ge4w(D)$, and otherwise may accept only when $I$ is positive,
all conflicts are negative, and $w(I)\ge w(D)/2$.
An acceptance replaces exactly $D$; a rejection changes nothing.
Its choice within this last action set may be randomized and depend
on the history. Every execution satisfies
$\mathrm{OPT}\le18\,\mathrm{ALG}$. The coefficient $18$ is sharp
uniformly over this action family, also after permanent feasibility
filtering.
\end{proposition}
\begin{proof}
Label ordinary and empty-conflict acceptances $\mathsf M$, and other
acceptances $\mathsf S$. Every $\mathsf M$ node has child mass at most
one fourth of its length. Every $\mathsf S$ node has only negative,
hence $\mathsf M$, children, with mass at most twice its length.
Every rejection has nonempty $D$ and $w(I)<4w(D)$.
Theorem~\ref{sum:thm:transfer} applies with $(\beta,\lambda)=(2,4)$:
$q_S=6$ and all three terms defining $K$ equal $18$.

For sharpness, fix $0<\epsilon<1/8$ and offer
\[
 \begin{array}{ccl}
 J&0&[0,1),\\
 X&0&[2\epsilon-4,\epsilon),\\
 Y&0&[1-\epsilon,5-2\epsilon),\\
 Z&0&[\epsilon,1/2-\epsilon),\\
 P&1&[1/2-\epsilon,1-\epsilon).
 \end{array}
\]
The rules force retention of $J$ through $Z$: each subsequent
negative request overlaps $J$ and has length below four.
The final request $P$ has length $1/2$ and is eligible for special
acceptance at equality. On the accepting execution, $\mathrm{ALG}=1/2$.
The compatible requests $X,Z,P,Y$ tile the entire request hull,
so $\mathrm{OPT}=9-4\epsilon$ and the ratio is $18-8\epsilon$.
The policy that accepts every special-eligible arrival realizes
this execution for every $\epsilon$, proving sharpness as
$\epsilon\downarrow0$. There is only one positive request, so
the filter leaves the construction unchanged.
\end{proof}

Let $H_p$ use an independent Bernoulli$(p)$ coin at each eligible
nonordinary positive arrival in this family.
For every fixed $p>0$, the preceding construction accepts $P$ with
probability $p$, so $H_p$ and its filtered form have exact almost-sure
robustness $18$. This uniform sharpness does not impose robustness
$18$ on every policy in the action family or determine its best
hybrid budget. On this input, $\mathrm{OPT}/\mathbb E[\mathrm{ALG}]$
tends to $9/(1-p/2)$, which is below $18$ when $p<1$.

\begin{proposition}[Fixed acceptance probabilities do not beat $d$]
\label{r38:prop:constant-coin-obstruction}
For every constant $p\in[0,1]$, the expected consistency of $H_p$
is at least $d$, despite its pathwise robustness bound of $18$.
\end{proposition}
\begin{proof}
Write $q=1-p$. For $n\ge2$, $1\le r<4$, and small $\varepsilon>0$,
first issue disjoint negative unit anchors
\[
 J_i=[a_i,a_i+1),\qquad a_i=i(1+r-2\varepsilon),
 \quad 0\le i<n.
\]
Then issue positive tiles of length below $1/2$ partitioning each
$[a_i+\varepsilon,a_i+1-\varepsilon)$, and positive intervals of
length $t<1/2$ crossing the two exterior anchor endpoints by
$\varepsilon$, where $t>\varepsilon$.
All are rejected. Finally issue, from left to right, the positive bridges
\[
 I_i=[a_i+1-\varepsilon,a_{i+1}+\varepsilon),\qquad 0\le i<n-1.
\]
Each has length $r$. The positives tile the entire request hull,
so their optimal value is
$n+(n-1)r-2n\varepsilon+2t$.
At $I_i$, the untouched right anchor is still held, and the left
anchor may also be held. Thus $w(D)\in\{1,2\}$, and every bridge
passes the special gate and strictly fails the ordinary gate.
Its acceptance is exactly its independent coin outcome.
Consequently
\[
 \mathbb E[\mathrm{ALG}]=(n-1)pr+2q+(n-2)q^2.
\]
Taking $\varepsilon\downarrow0$, $t\uparrow1/2$, and then
$n\to\infty$ gives the consistency lower bound
\[
 H_r(p)=\frac{1+r}{(1-p)^2+pr}.
\]
Set $p_*=1/d=(3-\sqrt5)/2$. If $p<p_*$, let $r\uparrow4$;
then $H_r(p)\to5/(1+p)^2>d$, since $(1+p_*)^2=5/d$.
If $p_*\le p\le1/2$, use $r=1$:
$1-p+p^2\le1-p_*+p_*^2=2/d$, hence $H_1(p)\ge d$.
At $p=p_*$ this is a supremum statement, sufficient to exclude
every uniform consistency constant below $d$.

For $p\ge1/2$, use a different perfect input with
$0<\varepsilon<1/4$. Start with $[0,1)$
negative, and finish with $K=[1/4,3/4)$ positive. Before $K$, issue
positive tiles partitioning $[\varepsilon,1/4)$ and
$[3/4,1-\varepsilon)$, and the positive intervals
$[\varepsilon-t,\varepsilon)$ and
$[1-\varepsilon,1-\varepsilon+t)$, where
$\varepsilon<t<1/2$ and all tiles are shorter than $1/2$.
These early positives are rejected. The positives together tile
a hull of length $1-2\varepsilon+2t\to2$.
Only $K$ invokes a coin, yielding expected gain $1-p/2$.
The limiting ratio is $2/(1-p/2)\ge8/3>d$.
\end{proof}

State-dependent probabilities and correlated decisions, which remain pathwise protected within Proposition~\ref{r38:prop:random-action-envelope}'s action set, are outside this obstruction.

\section{Same-threshold ledger ablation}\label{app:ledger-ablation}
The ablation removes only SUM-BANK's special ledger acceptance test,
retaining ordinary priority, length thresholds, all-negative conflicts,
ledger updates and irrecoverability. It measures the gate's total effect
on the resulting trajectory.

\subsection{A five-request explanation of the ledger test}
Fix $\lambda\ge\varphi^3$ and $0<\varepsilon<1/8$. Reveal
the following intervals in order; the subscript is the prediction:
\[
 [0,1)_0,\quad
 [\varepsilon,1/2)_1,\quad
 [1/2,1-\varepsilon)_1,\quad
 [-1/\varphi+2\varepsilon,\varepsilon)_1,\quad
 [1-\varepsilon,1-\varepsilon+1/\varphi)_1.
\]
The four positives tile the full input hull, giving the unique
optimum $1+2/\varphi-3\varepsilon$. The initial negative
conflicts with each of them. Both policies reject the next three
requests because their lengths are below $1/\varphi$; none passes
the ordinary gate. Immediately before the last arrival, the held
length is one and the ledger is $B=\varphi-3\varepsilon>1$.
The last positive passes the special length test at equality.
Without the ledger test it replaces the negative, leaving reward
$1/\varphi$ and ratio
\[
 \varphi+2-3\varphi\varepsilon\longrightarrow\varphi+2>d.
\]
With the test, it is rejected because
$B>1=\varphi w(I)$, and SUM retains reward one. Permanent
filtering changes no label because the positives are compatible.
Thus the ablation loses the proved consistency guarantee even
though its mean effect on the workloads below is small. This
construction also applies to Shaped-BANK and every Credit-BANK $T_c$
in the stated family when only its special ledger test is disabled and ledger
updates are retained: throughout the prefix the sole held negative
has ordinary cost at least $d$, whereas every positive length is
at most $1/\varphi$. Hence all ordinary tests still fail, and the
same length and ledger decisions follow. The conclusion concerns only these policies with the test removed.

\subsection{Paired workload intervention}

Each of the eight conditions has 100 repetitions; within a repetition, enabled and
disabled SUM gates share inputs, optima, permutations and labels.
Differences are percentage points of $\mathrm{ALG}/\mathrm{OPT}$, using
the separate eight-comparison bootstrap family in
Appendix~\ref{sec:experiment-protocol}. This intervention does not cover
Shaped-BANK or the low-FP rates in Table~\ref{tab:main-mixtures}.

\begin{table}[htbp]\centering\small
\caption{Ledger ablation, 100 repetitions per condition. Win/tie/loss counts compare exact paired
rewards; positive favors the enabled gate.}
\label{tab:ledger-ablation}
\begin{tabular}{llrrr}\toprule
Log & Condition & Mean difference & Adjusted 99.375\% CI & Win/tie/loss\\\midrule
NASA & Perfect & 0.00006 & $[0.00000,0.00019]$ & 2/98/0 \\
NASA & FN .1 & 0.00000 & $[0.00000,0.00000]$ & 0/100/0 \\
NASA & FP .1 & 0.01227 & $[0.00238,0.02129]$ & 72/0/28 \\
NASA & FP 1 & 0.00000 & $[0.00000,0.00000]$ & 0/100/0 \\
CTC & Perfect & 0.00776 & $[0.00297,0.01375]$ & 17/81/2 \\
CTC & FN .1 & 0.00406 & $[0.00064,0.00847]$ & 10/89/1 \\
CTC & FP .1 & 0.54094 & $[0.48731,0.59718]$ & 99/0/1 \\
CTC & FP 1 & 0.00000 & $[0.00000,0.00000]$ & 0/100/0 \\
\bottomrule\end{tabular}\end{table}
All-positive inputs make the special branch inactive, producing identical complete trajectories. On NASA FN .1, no disabled-gate veto occurs in these streams, also giving equality. The remaining comparisons include losses in individual repetitions despite nonnegative means. 

\section{Reward and error-bound activation}\label{r5:app:current-comparison}
\begin{table*}[t]\centering\small\setlength{\tabcolsep}{5pt}
\caption{Random-order reward, percent of optimum (20 inputs per condition).
Each cell is $\text{mean}_{-\ell}^{+u}\,/\,q_{.05}$: the pointwise 95\%
bootstrap interval is $[\text{mean}-\ell,\text{mean}+u]$; $q_{.05}$ is the
empirical individual-run fifth percentile: the sample minimum for each
deterministic method, and the weighted quantile of 40 branch outcomes for
each mixture. Mix$_C$, Mix$_B$ use $\pi=\alpha^2/2$, $1/\varphi$ in Section~\ref{sec:randomized-main} (expected consistency $d$, resp. robustness $\varphi^5$); means integrate the coin. Act. counts inputs with $\eta_K<(4-2\varphi)\mathrm{OPT}$ (Appendix~\ref{sec:experiment-protocol}). SG denotes Shaped-BANK; $\mathrm{SG}^F$ is filtered SG. These are exploratory fixed-log summaries.}
\label{tab:main-mixtures}
\begin{tabular}{@{}llcrrrr@{}}\toprule
Log & Labels & Act. & LR & $\mathrm{Mix}_C$ & $\mathrm{Mix}_B$ & $\mathrm{SG}^F$\\\midrule
NASA & Perfect & 20/20 & $96.83_{-.13}^{+.14}\,/\,96.25$ & $97.43_{-.10}^{+.12}\,/\,96.47$ & $98.79_{-.05}^{+.05}\,/\,96.51$ & $99.10_{-.09}^{+.09}\,/\,98.76$\\
NASA & FP 0.01923 & 20/20 & $96.83_{-.13}^{+.14}\,/\,96.25$ & $94.95_{-.15}^{+.14}\,/\,85.85$ & $90.75_{-.37}^{+.36}\,/\,85.41$ & $97.07_{-.10}^{+.11}\,/\,96.61$\\
NASA & FP 0.1 & 0/20 & $96.83_{-.13}^{+.14}\,/\,96.25$ & $91.66_{-.18}^{+.20}\,/\,69.28$ & $80.11_{-.44}^{+.44}\,/\,67.55$ & $94.36_{-.20}^{+.20}\,/\,93.60$\\
NASA & FP 1 & 0/20 & $96.83_{-.13}^{+.14}\,/\,96.25$ & $90.21_{-.11}^{+.12}\,/\,61.36$ & $75.42_{-.24}^{+.24}\,/\,61.00$ & $89.12_{-.28}^{+.27}\,/\,87.84$\\
CTC & Perfect & 20/20 & $93.27_{-.11}^{+.10}\,/\,92.77$ & $94.55_{-.08}^{+.08}\,/\,92.86$ & $97.43_{-.04}^{+.04}\,/\,92.99$ & $95.58_{-.14}^{+.14}\,/\,94.83$\\
CTC & FP 0.0034 & 20/20 & $93.27_{-.11}^{+.10}\,/\,92.77$ & $91.74_{-.15}^{+.15}\,/\,84.26$ & $88.33_{-.39}^{+.42}\,/\,83.22$ & $93.80_{-.17}^{+.16}\,/\,92.95$\\
CTC & FP 0.1 & 0/20 & $93.27_{-.11}^{+.10}\,/\,92.77$ & $84.81_{-.12}^{+.12}\,/\,48.49$ & $65.89_{-.28}^{+.28}\,/\,47.58$ & $89.25_{-.14}^{+.14}\,/\,88.72$\\
CTC & FP 1 & 0/20 & $93.27_{-.11}^{+.10}\,/\,92.77$ & $87.72_{-.10}^{+.11}\,/\,63.87$ & $75.32_{-.22}^{+.22}\,/\,63.43$ & $86.39_{-.16}^{+.16}\,/\,85.68$\\
\bottomrule\end{tabular}\end{table*}

Table~\ref{tab:main-mixtures} gives the reward, uncertainty and tail summaries
discussed in Section~\ref{sec:empirical}.
Figure~\ref{r5:fig:current-comparison} compares eight policies on twenty
shared joint repetitions. Appendix~\ref{sec:experiment-protocol} defines
the policies, inputs and conditional uncertainty estimates.

\begin{table*}[t]\centering\small
\caption{Mean retained reward (\% of OPT; twenty shared repetitions per cell).
Shaped-BANK (SG) and RP at $\lambda=(5\varphi+1)/2$ share exact
robustness $5\varphi+3$. Superscript $F$ denotes permanent filtering.
The best scalar Credit uses $c_*=2\varphi-3$ and budget $4\varphi+5$;
SUM has budget $7\varphi+4$; MAX($\varphi^3$) is a descriptive comparator.
Credit at $c=\varphi/2$ is included in the full comparison tables.}
\label{tab:absolute-endpoints}\setlength{\tabcolsep}{4pt}
\begin{tabular}{@{}llrrrrrrrr@{}}\toprule
Log & Labels & SG$^F$ & SG & Scalar$^F$ & SUM$^F$ & MAX($\varphi^3$) & LR & Half(4) & RP$^F$\\\midrule
NASA & Perfect & 99.10 & 99.10 & 99.10 & 99.17 & 98.93 & 96.83 & 99.26 & 98.67 \\
NASA & FP .1 & 94.36 & 94.92 & 94.35 & 92.26 & 96.20 & 96.83 & 93.15 & 91.76 \\
NASA & FP 1 & 89.12 & 88.63 & 89.08 & 88.56 & 93.07 & 96.83 & 89.01 & 88.17 \\
CTC & Perfect & 95.58 & 95.58 & 95.58 & 95.91 & 95.51 & 93.27 & 97.22 & 92.80 \\
CTC & FP .1 & 89.25 & 89.68 & 89.24 & 87.34 & 90.63 & 93.27 & 86.84 & 87.24 \\
CTC & FP 1 & 86.39 & 85.72 & 86.37 & 85.73 & 89.15 & 93.27 & 86.12 & 85.37 \\
\bottomrule\end{tabular}\end{table*}

\paragraph{Reward comparisons.}
Filtered SG loses to filtered SUM on both perfect-label cells,
and to raw SG on all six nonzero low-FP cells. MAX($\varphi^3$)
beats SG at FP $.1$, $.5$ and $1$ on both logs; LR is strongest among
the displayed methods at those rates and at FN 1.
Table~\ref{tab:chronology-diagnostic} contains a low-FP exception to this ranking.
 The filtered-RP means favor SG in
all 20 cells, but at CTC FP .1 the raw-RP paired interval includes zero.

Figure~\ref{fig:sg-low-fp} shows activation counts 20,7,0 on NASA and
20,5,0 on CTC at the three nonzero low-FP rates; both perfect cells
activate 20/20. In contrast, the main FP grid has nonpositive error branches.

Scalar Credit at $c=\varphi/2$ has a different low-FP threshold
and middle-rate counts 9/20 and 11/20. SG's larger guaranteed robust
floor can reduce the error-branch activation count while strengthening the worst-case guarantee.

\paragraph{Best scalar control.}
Filtered SG minus filtered best scalar Credit ($c_*$;
Appendix~\ref{sec:scalar-protocol}) has fifteen positive means,
four exact zero means and one negative mean; the range is
$[-0.00085,0.03927]$ percentage points. At NASA FP$.05769$,
SG minus scalar is $-0.00085$ with pointwise interval
$[-0.00946,0.00592]$ and 2/17/1 wins/ties/losses. Raw differences
have fourteen positive and six zero means, with maximum 0.16959
points at CTC FP$.5$ (interval $[0.13983,0.20048]$).
The workload-reward differences are small despite the strict improvement
in the worst-case guarantee.

\begin{figure}[htbp]\centering
\includegraphics[width=\textwidth]{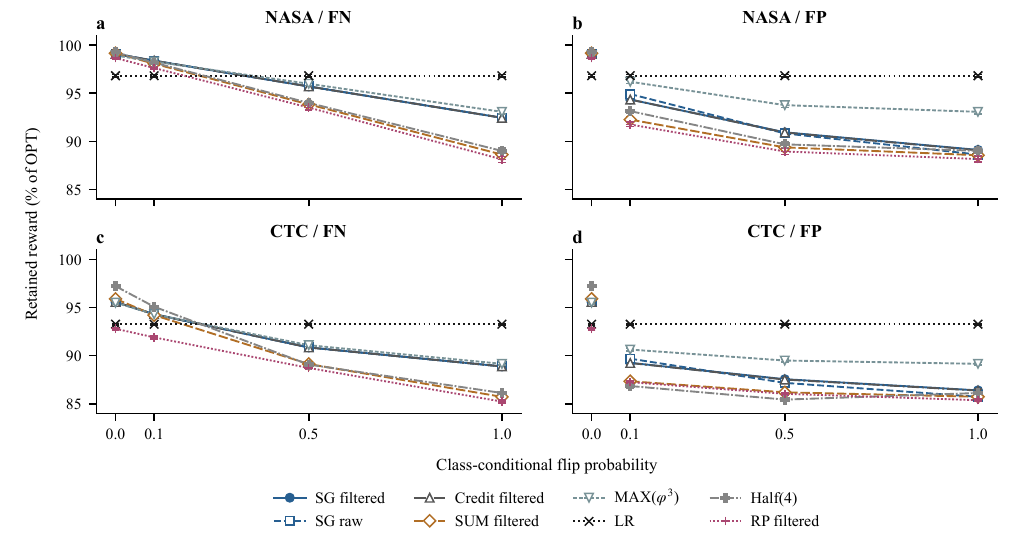}
\caption{Main-regime rewards: NASA/CTC rows and FN/FP columns. The eight methods are raw/filtered SG, filtered best scalar Credit, filtered SUM, MAX($\varphi^3$), LR, Half(4), and filtered RP at the SG robustness budget.
Points are twenty-repetition means with pointwise 95\% bootstrap
intervals (10,000 shared resamples). Perfect labels appear in both
columns; FN and FP flip different classes. The FP perfect-to-$0.1$
segment is omitted because the low-FP regime lies between them.
Methods have different guarantee budgets; MAX is descriptive.
Raw/filtered SG coincide in the FN panels, while filtered SG and
scalar Credit nearly overlap at this scale.}
\label{r5:fig:current-comparison}\label{fig:filtered-rp}
\end{figure}

\begin{figure}[htbp]\centering
\includegraphics[width=\textwidth]{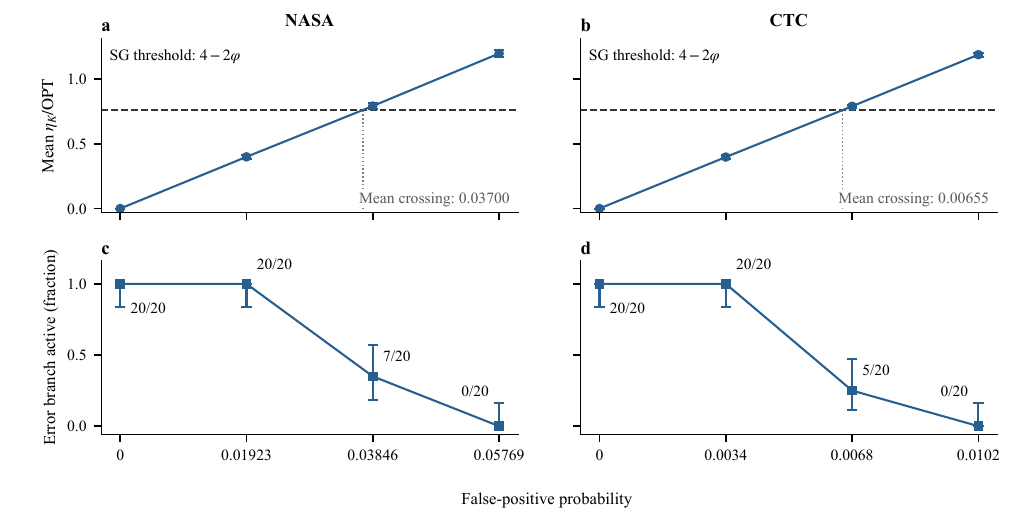}
\caption{Bounded low-FP study. Top: mean original-label
$\eta_K/\mathrm{OPT}$ with pointwise bootstrap intervals; dashed
line: SG's exact activation threshold $4-2\varphi$. Bottom:
actual activation counts out of twenty with pointwise Wilson
intervals. Dotted vertical references mark the analytic expected-error
crossings, $.03700$ and $.00655$; the measured grid was fixed around
the Credit crossings. Rates share nested labels and reveal orders. Mean-error crossing
and realization-level activation are distinct quantities.}
\label{fig:sg-low-fp}
\end{figure}

\subsection{Uncertainty and individual-run rewards}
\label{r39:sec:saved-quantiles}
Table~\ref{tab:main-mixtures} uses the twenty matched inputs and
pointwise bootstrap protocol of Appendix~\ref{sec:experiment-protocol}.
Each input fixes the requests, order, labels and termination; mixture
coins introduce no additional input replications.

The mixture mean integrates the initial policy coin. Its empirical
fifth percentile instead uses the forty branch outcomes with masses
$\pi/20$ and $(1-\pi)/20$. For a fixed input, both branch probabilities exceed $.05$, so
the conditional fifth percentile is $\min\{N_r,L_r\}$. Averaging these
conditional quantities is not the pooled quantile in Table~\ref{tab:main-mixtures}.
The latter can differ between Mix$_C$ and Mix$_B$ even though their
conditional minima coincide. 

Explicitly, with rewards $N_r,L_r$ as percentages of optimum,
\[
 \widehat F_\pi(t)=\frac1{20}\sum_{r=1}^{20}
 [\pi\mathbf1_{N_r\le t}+(1-\pi)\mathbf1_{L_r\le t}],
 \qquad q_{.05}=\inf\{t:\widehat F_\pi(t)\ge .05\}.
\]

\section{Shared experimental protocol and provenance}
\label{sec:experiment-protocol}

\subsection{Instances, policies and comparison budgets}
Raw/filtered SG and raw/filtered RP at $\lambda=(5\varphi+1)/2$
share exact robustness $5\varphi+3$. Each uses two fixed logs, seven
main conditions (perfect labels; FN or FP at $.1,.5,1$), and three
additional low-FP rates per log, with twenty joint repetitions, giving 400 cases. The low-FP rates are
$.01923,.03846,.05769$ on NASA and $.0034,.0068,.0102$ on CTC.

Thirteen main-regime controls share those repetitions: raw/filtered
Credit ($c=\varphi/2$), raw/filtered SUM ($\lambda=\varphi^3$),
MAX at $\varphi^3$ and $\Lambda$, LR, Half(4), and five RP settings.
These RP settings are $\varphi^3$, $(7\varphi+2)/2$,
$\lambda_0=(B_0+\sqrt{B_0(B_0-8)})/4$ with $B_0=8\varphi+5$,
and raw/filtered $(6\varphi+1)/2$. The last pair matches Credit's
budget $6\varphi+3$; $(7\varphi+2)/2$ matches SUM's $7\varphi+4$.
RP at $\lambda_0$ has exact robustness $2\lambda_0+2$, not $B_0$.
Here $\Lambda=(5\varphi+2+\sqrt{33\varphi+21})/2\simeq9.357715$; MAX is a descriptive
comparator with no claimed guarantee.

The separate SUM ledger ablation (Appendix~\ref{app:ledger-ablation}) uses 100
repetitions; its paired differences are computed within that study.

The logs are NASA-iPSC-1993-3.1-cln~\cite{pwaNASA} and
CTC-SP2-1996-3.1-cln~\cite{pwaCTC} from the Parallel Workloads
Archive~\cite{pwa2014}. Mirrored inputs match later official downloads;
earlier download logs and original cleaning source are unavailable.
Raw SWF files are not redistributed because unrestricted redistribution is not
established. Mapping follows
SWF~\citep{pwaSWF}.

Each retained record gives $[s,s+t)$ with weight $t>0$ seconds.
NASA uses SWF column 2 as a start-time surrogate, without adding its
unknown wait $-1$; CTC uses $s=\texttt{submit}+\texttt{wait}$ with known
nonnegative wait. Cleaning requires eighteen columns, valid required
integer fields, nonnegative start surrogate/submission and positive
runtime; neither status nor processor count is a filter. All exclusions
are zero-runtime records. Stable IDs are physical SWF line numbers;
coincident intervals remain distinct.

\begin{table}[htbp]\centering\small
\caption{Fixed workload-derived instances. Exact objectives are seconds.}
\label{tab:e00-data-details}
\begin{tabular}{@{}lrr@{}}\toprule
Quantity & NASA & CTC\\\midrule
Data records & 18,239 & 77,222\\
Excluded zero-runtime records & 173 & 17\\
Retained intervals & 18,066 & 77,205\\
Exact optimum $V$ & 5,816,181 & 28,353,920\\
Fixed optimal-set size (FN eligible) & 3,635 & 1,464\\
Complement size (FP eligible) & 14,431 & 75,741\\
\bottomrule\end{tabular}\end{table}

Integer weighted-interval dynamic programming sorts by
$(\mathrm{finish},\mathrm{start},\mathrm{id})$ and computes
$V_j=\max\{V_{j-1},w_j+V_{k(j)}\}$, where $k(j)$ counts earlier
requests finishing no later than request $j$ starts. A tie excludes
the current request, fixing an optimum $O$ independently of reveal
order and labels. Endpoint contact is compatible. Only offline
optimization and label construction use the full instance. Online
methods see immutable labels one request at a time, without termination
knowledge, recall or geometric expiration. This single-resource,
random-reveal reduction does not reproduce queues or processor allocation.

\subsection{Randomization and uncertainty}
For each log and repetition $r=0,\ldots,19$, PCG64 streams supply
one reveal permutation and one uniform per stable ID. FN flips members
of $O$ to zero; FP flips its complement to one. An eligible ID flips
when its uniform is below the binary64 conversion of the stored rational
rate $q$. Rates therefore share nested labels; methods and error types
share streams, so cells are not independent samples. Perfect FN/FP labels
coincide. FN positives are feasible; FP positives need not be.

Rewards are $100\,\mathrm{ALG}/V$. The bootstrap draws $10{,}000$ resamples of the twenty whole repetitions per log, shared
across methods and conditions. Pointwise 95\% percentile intervals use
linear interpolation at .025 and .975; paired intervals resample
within-repetition differences in percentage points. The 60 primary
contrasts compare filtered SG with raw SG and raw/filtered matched RP
in twenty cells. The 34 secondary contrasts compare filtered SG with
filtered $c=\varphi/2$ Credit in twenty cells and filtered SUM in fourteen
main cells. Forty further contrasts compare raw/filtered SG with their
best scalar counterparts. The ancillary tables (see \texttt{anc/README.md}) give full summaries for all twenty cells and all 134 contrasts,
including adverse and zero results. These bootstrap intervals and the Wilson activation intervals are pointwise and
conditional on the two fixed logs.

The SUM ablation uses $10{,}000\times100$ whole-repetition resamples and
quantiles .003125/.996875 for its eight-comparison family, separately
from the pointwise SG comparisons. Without public preregistration, the ablation conditions, low-FP rates and $c_*$ were
fixed before their outcomes; the eight-policy display before scalar outcomes; and the
models, thresholds and independent-cohort comparisons in
Appendix~\ref{app:rollout-interface} before test rewards.

\subsection{Original-label error and activation}\label{sec:error-scale}
For original positives $P$ and fixed optimum $O$, compute exactly
\[
 \eta_K=w(O\setminus P)+\sum_{F\in P\setminus O}
 \bigl[w(N_O(F))-w(F)\bigr].
\]
Original labels define this statistic; immutable effective labels govern
filtered decisions and fees. SG's error branch exceeds its robust floor
precisely when
\[
 \eta_K/V<1-d/(5\varphi+3)=4-2\varphi\simeq .763932.
\]
For marginal FP probability $q$, $\mathbb E[\eta_K/V]=qS$, with
$S=120075280/5816181\simeq20.6450$ on NASA and
$S=3309283172/28353920\simeq116.7134$ on CTC.
Credit's threshold $(4-\varphi)/3\simeq .793989$ gives mean crossings
$.038459/.006803$; SG gives $.0370031776/.0065453654$.
The added grid surrounds both crossings below the original minimum
nonzero rate $.1$. 

\subsection{Policy definitions and implementation}
With $d=\varphi+1$, $\alpha=\varphi-1$, $M=2\varphi+1$, SG uses
negative cost $\min\{2\varphi v,\max\{dv,\varphi v+\alpha b\}\}$
and positive cost $Mv$, capped per holder before summation. Breakpoints
$b=\varphi v,dv$ and rational/golden-field comparisons are exact.
Ordinary acceptance has priority; non-strict boundaries accept equality.
Rejection fees use old values, and replacement inherits whole ledgers.
Permanent positive history records a compatible effective positive even
if subsequently rejected or revoked.

Credit at $c=\varphi/2$ uses $v+(v+\min\{v,b/2\})\varphi$ for a
negative holder. SUM's ordinary test uses total conflicting length;
MAX changes only this ordinary quantity to maximum conflicting length,
retaining SUM's special gates and ledger updates. LR follows Algorithm~3
of \citet{karavasilis2025}: empty acceptance or strict $x>\varphi y$,
where $y$ is maximum conflicting length. RP follows Algorithm~4: after
ordinary acceptance fails, a positive arrival with all-negative conflicts
must satisfy $w(I)\ge W$. Half(4) changes this last test to $2w(I)\ge W$
and uses ordinary parameter 4. It is an author-chosen empirical control;
no RP certificate or optimal-parameter claim transfers to it.
These controls use the definitions above, rather than reproducing published workload
curves.

An independent full-scan implementation agrees on all SG decisions and rewards; labels,
errors and bootstrap summaries were separately reconstructed.

\paragraph{Arithmetic resources.}\label{r39:sec:arithmetic-resources}
For fixed parameters of the displayed SUM, Shaped and $T_a$ costs,
a full-scan implementation uses $O(n^2)$ exact arithmetic operations
and comparisons, and $O(n)$ stored scalars and interval records over
$n$ arrivals. At time $t$, scan at most $t-1$ holders to form conflicts,
sum lengths, costs and ledgers, then update fees or replace records.
Each displayed cost takes constant arithmetic work; fee denominators
are positive on their branches. The permanent positive filter adds one
scan and at most $t$ records. Source measures and replacement forests
are proof objects, not stored policy state. Parameters are exact
arithmetic constants; rational inputs and rational $a$ use rational
operations, and the golden endpoint uses a fixed quadratic field.
This bounds scalar counts, not their bit sizes or floating-point error;
OPT computation and retained diagnostic traces are excluded.

\subsection{Strongest-scalar control}\label{sec:scalar-protocol}
The scalar $c_*=\alpha^3=2\varphi-3$ is the best consistent scalar endpoint, with exact
robustness $4\varphi+5$ and cost $\min\{Mv,dv+c_*b\}$. Raw/filtered scalar Credit uses
the same 400 cases and bootstrap resamples as SG.
Appendix~\ref{r5:app:current-comparison} summarizes its forty contrasts; $c=\varphi/2$
remains in the full tables.

\section{Experimental diagnostics: randomization, learned labels, and application limits}
\label{app:randomized-experiment}

\paragraph{A randomized control with its coin integrated exactly.}
Naive accepts a positively predicted request only when it is compatible
with every previously accepted request, and never revokes one. The
initial mixture selects Naive with probability $\pi$ and LR otherwise,
once before the input is processed. For each fixed input and prediction
vector we report
\[
 100\{\pi\,\mathrm{Naive}+(1-\pi)\,\mathrm{LR}\}/\mathrm{OPT}.
\]
We use $\pi_C=\alpha^2/2$ and
$\pi_B=1/\varphi$, with expected guarantees $(\bar C,\bar R)\le
(d,2d)$ and $(\varphi^5/(3\varphi+3),\varphi^5)$, respectively.
Inputs are fixed independently of the coin; these are expected-reward
guarantees, not finite per-execution robustness bounds.

In the perfect/FP $.1$/FP $1$ rows of Table~\ref{tab:main-mixtures}, the underlying Naive means are
$100.00/69.78/62.18\%$ on NASA and $100.00/48.97/64.23\%$ on CTC.
The budget-matched mixture beats filtered SG on perfect CTC labels,
but loses at FP $.1$ on both logs; the consistency-matched mixture
beats filtered SG at FP $1$ on both logs. 

\paragraph{Finite binary64 decision check.}
With integer geometry/rewards and binary64 thresholds/ledgers, literal
min/max evaluation agrees with exact breakpoint evaluation on branches,
labels, final IDs and rewards for 800 SG streams (400 cases;
38,108,400 arrivals). This finite check gives no guarantee for
adversarial near-boundary inputs or accumulated ledger error.

\paragraph{A CTC submission-order diagnostic.}
We keep CTC's interval geometry $[\mathrm{submit}+\mathrm{wait},
\mathrm{submit}+\mathrm{wait}+\mathrm{runtime})$, fixed optimum and uniform arrays, but reveal requests in ascending submission
time, breaking ties by physical-line ID. Perfect labels give one
deterministic case; FP probabilities $.0068$ and $.1$ each use the twenty label streams. Both SG variants use exact arithmetic. This changes reveal order alone; the application interpretation is
examined in Appendix~\ref{sec:r030-decision-alignment}.

\begin{table}[htbp]\centering\small
\caption{CTC submission-order sensitivity. Mean percentages of the fixed
optimum; $n=1$ for perfect labels and $n=20$ for each nonzero FP rate.
The initial coin is integrated for both mixture columns.}
\label{tab:chronology-diagnostic}
\begin{tabular}{@{}lrrrrrr@{}}\toprule
Condition & Naive & LR & $\mathrm{Mix}_C$ & $\mathrm{Mix}_B$ & SG & $\mathrm{SG}^{F}$\\\midrule
Perfect & 100.00 & 94.98 & 95.94 & 98.08 & 95.61 & 95.61\\
FP .0068 & 95.05 & 94.98 & 94.99 & 95.02 & 95.30 & 95.28\\
FP .1 & 78.11 & 94.98 & 91.76 & 84.55 & 94.10 & 93.57\\
\bottomrule\end{tabular}\end{table}

At FP $.1$, submission order raises filtered SG's mean from
$89.25\%$ to $93.57\%$, and the budget-matched mixture's from
$65.89\%$ to $84.55\%$. LR remains better than filtered SG.
At FP $.0068$, filtering costs only $0.020$ percentage points,
compared with $1.276$ under the random-reveal protocol.
The perfect-label mixture still exceeds SG. The magnitude of these comparisons thus depends strongly on reveal order.

\subsection{Learned labels, structured errors, and a planted obstruction}
\label{sec:r029-extensions}
To test errors beyond independent flips, we construct separate
instances from the logs. Requests are ordered by
$(t_i,\mathrm{ID}_i)$, where $t_i$ is SWF column 2 and $\mathrm{ID}_i$
is the physical-line identifier. The 60/20/20 split starts at
$\lfloor .6n\rfloor,\lfloor .8n\rfloor$ and advances each cut past equal
timestamps. Training/validation/test counts are 10,839/3,613/3,614 for
NASA and 46,323/15,441/15,441 for CTC. Consecutive blocks of 512 requests form separate
instances; a final block shorter than 32 is appended to its predecessor.
Each block has its own canonical offline optimum, so conflicts across
blocks are omitted. Its optimum uses the tie-breaking DP of Appendix~\ref{sec:experiment-protocol}. Block boundaries, unlike
split boundaries, may separate equal timestamps. There are 7 NASA and 31 CTC test blocks. These blocks
are not assumed to be independent samples from a workload population.

\paragraph{A concrete prediction pipeline.}
For each log, an $\ell_2$-regularized logistic model learns membership in
the training blocks' canonical optima. Its seven causal features are
log-duration, log-lead-time, log-gap since the preceding reveal, log-count
and log-total-length of previously revealed conflicts, and two time-of-day
coordinates. Logarithms use $\log(1+x)$; lead time is $s_i-t_i$.
History starts empty in every block, including at split boundaries;
the first gap is zero. Earlier equal-time requests enter in ID order.
Conflict features include all earlier overlapping reveals in the block, regardless of labels
or acceptance, and their full lengths are summed. The periodic features
are sine and cosine of $2\pi(t_i\bmod86400)/86400$, with no calendar offset.
Training-only unweighted means and population standard deviations
scale all splits, with zero deviations replaced by one. The constructed
interval inputs reveal runtime at arrival; NASA's column 2 is a start-time surrogate. See
Appendix~\ref{sec:r030-decision-alignment} for the application boundary.

Writing $h_i=\beta_0+\beta^\top z_i$ for standardized features, training
minimizes the \emph{sum} of logistic losses
$\sum_i[\log(1+\exp h_i)-y_i h_i]+\|\beta\|_2^2/2$;
the intercept is unpenalized. A deterministic Newton solver fits the model.
Labels are $\mathbf1\{p_i\geq\tau\}$. Validation pooled F1 chooses
$\tau\in\{.05,.10,\ldots,.95\}$, breaking ties toward the largest
threshold and giving $.35$ and $.20$; neither test data nor online reward
is used for tuning. Fixed $.5$, explicit all-negative/all-positive and
perfect-label controls are also reported.

\begin{table}[t]\centering\small
\caption{Held-out learned-label checks. Precision/recall pool request
counts; weighted recall and $\eta_K/\mathrm{OPT}$ pool lengths. Reward
rows are means of block percentages. The selected threshold maximizes
validation F1, not the online objective.}
\label{tab:r029-learning}
\begin{tabular}{@{}lrr@{}}\toprule
Quantity & NASA & CTC\\\midrule
Precision / recall & $.851/.660$ & $.428/.426$\\
Weighted recall & $.610$ & $.364$\\
$\eta_K/\mathrm{OPT}$ & $.577$ & $1.574$\\
Active error branch & $6/7$ & $0/31$\\\midrule
Shaped / filtered & $94.58/94.58$ & $89.81/89.66$\\
SUM / MAX at $\varphi^3$ & $95.17/95.17$ & $89.43/89.95$\\
LR / Naive & $96.86/61.43$ & $91.25/62.68$\\
Consistency-matched mixture & $90.09$ & $85.79$\\
Budget-matched mixture & $74.96$ & $73.59$\\
Filtered Shaped, all-negative & $96.96$ & $90.30$\\
Naive, all-positive & $90.45$ & $83.51$\\\bottomrule
\end{tabular}
\end{table}

Both constant controls outperform their corresponding learned-label
policies in Table~\ref{tab:r029-learning}. Fixed threshold $.5$ gives
$94.88/89.90\%$ for filtered Shaped, and perfect labels give
$97.61/96.58\%$. Validation F1 need not select useful advice for the
online objective or determine conflict-weighted error.

\paragraph{Matched counts with different error structure.}
On each test block we create FP errors at target $.01$ or FN errors at
target $.1$, rounding the eligible count to the nearest integer.
Uniform placement, length-proportional sampling without replacement,
and a cyclic consecutive burst in eligible reveal order have exactly
the same error count within a block. They do not match error weight or
per-request marginals. Each block and placement uses one fixed seed; the comparisons are descriptive. With 85 total NASA false negatives,
uniform versus length-weighted placement changes mean $\eta_K/\mathrm{OPT}$
from $.081$ to $.526$, Naive's reward from $91.85$ to $47.41\%$, and
Shaped's from $97.92$ to $96.71\%$. For CTC the corresponding numbers are
$.108$ versus $.272$, $89.24$ versus $72.80\%$, and $94.96$ versus
$92.86\%$. Bursts and FP placements give different rankings; the ancillary tables (see \texttt{anc/README.md}) list all 228 cases,
including unfavorable comparisons.

\paragraph{Embedding the zero-error construction.}
We instantiate the distinct-geometry family with $\delta=1/100$, three
central groups, eight alternatives per group, and rational boundary
parameter $16/5<2\varphi$. Its 27 requests have
$\mathrm{OPT}=147/20$, raw Shaped reward $1$, and $\eta_K=0$.
A spatially disjoint scaled copy is embedded into the first test block
of each log, whose background labels are perfect. The witness supplies
a fraction $a\in\{0,1/4,1/2,9/10,1\}$ of total optimum; its requests
retain their adversarial order. At $a=1/2$, raw/filtered Shaped give
$56.46/93.05\%$ on NASA and $55.12/91.72\%$ on CTC. At $a=1$, raw
Shaped, SUM and MAX give $20/147$, filtered Shaped and LR give
$638/735$, and Naive is optimal. The error remains zero, although the
positive set is infeasible. This explains the filter's role on this proved construction.

Decisions and ledgers use exact rational/golden arithmetic; mixture coins are integrated analytically.

\subsection{Decision-aligned labels and the application boundary}
\label{sec:r030-decision-alignment}
We keep the split, seven features and linear-logistic capacity of
Appendix~\ref{sec:r029-extensions}, but align training and threshold
selection with decision costs. This secondary design was chosen after test-block inspection. It uses weighted
logistic training, motivated by prediction versus decision loss
\citep{elmachtoub2020spo}, rather than SPO+ or its guarantee.

\paragraph{A cost target tied to the error.}
For a training block with optimum $O$ of value $V$, assign cost
\[
 c_i=\begin{cases}
 w_i/V,&i\in O,\\
 \bigl(w(N_O(i))-w_i\bigr)/V,&i\notin O.
 \end{cases}
\]
Then $\sum_i c_i\mathbf1\{b_i\ne\mathbf1_{i\in O}\}
=\eta_K(b,O)/V$. In addition to the ordinary-logistic score,
we fit a score using $c_i$-weighted binary log loss, globally rescaling
training costs to mean one and retaining the same L2 coefficient~1
on the summed loss. Thus the loss weights are $mc_i/\sum_jc_j$ for
$m$ training requests. Zero-cost negatives have zero loss weight but
remain in the unweighted training standardization. Only the cost model is fitted anew, using the same solver and an
additional final gradient infinity-norm check below $10^{-5}$.
Such a cost score is not an ordinary calibrated membership probability. Neither does minimizing $\eta_K$ maximize retained reward.

Each score uses the 21 validation thresholds $0,.05,\ldots,1$; the endpoints give
all-positive and all-negative labels on the evaluated inputs. Explicit constant controls
are retained.
Selection maximizes pooled F1, minimizes mean block $\eta_K/V$, or
maximizes mean block filtered-SG reward; ties choose the largest
threshold. Together with fixed $.5$ thresholds and negative/positive/
perfect controls, this gives 11 conditions for each of the eight policies.

\begin{table}[t]\centering\small
\caption{Secondary objective-alignment diagnostic. Entries are mean block
percentages for filtered SG on the same 7 NASA and 31 CTC test blocks.
Reward selection uses only validation blocks. Threshold 1 gives
all-negative labels on these scores.}
\label{tab:r030-alignment}
\begin{tabular}{@{}lrr@{}}\toprule
Condition & NASA & CTC\\\midrule
Ordinary score, F1-selected & 94.58 & 89.66\\
Ordinary score, error-selected & 93.71 & 90.30\\
Ordinary score, reward-selected & 96.96 & 90.30\\
Cost score, F1-selected & 93.13 & 90.21\\
Cost score, error-selected & 94.91 & 90.30\\
Cost score, reward-selected & 96.96 & 90.48\\\midrule
All-negative & 96.96 & 90.30\\
Perfect labels & 97.61 & 96.58\\
LR (ignores labels) & 96.86 & 91.25\\\bottomrule
\end{tabular}
\end{table}

Reward selection chooses all-negative advice for both ordinary scores and
the NASA cost score. Their improvement over F1 therefore cannot be credited
to learning. The CTC cost score selects $.30$ and gains $.1802$ percentage
points over the constant: 3 blocks improve, 27 tie and 1 worsens. The
leave-one-block-out mean range is $[.1018,.1991]$ points, a sensitivity range, not a confidence interval. Its mean remains $.7690$
points below LR. The ancillary tables (see \texttt{anc/README.md}) list all twelve paired contrasts, including adverse results.

The error-oriented NASA score reduces pooled $\eta_K/V$ from $.5767$ to
$.4838$ and makes the error branch active in all seven blocks rather than
six, yet gives lower reward than the all-negative control, whose error
branch is inactive. Error informativeness and empirical utility are thus
distinct diagnostics. The CTC reward-selected cost score reduces pooled
error from $1.5735$ to $.9501$, activates the error branch in only 4 of
31 blocks and produces four special admissions. These results support
objective sensitivity, not a claim that objective mismatch is the sole
reason for the ordinary predictor's failure.

\paragraph{Feasibility and boundary context.}
Permanently projecting each positive through the existing causal Naive
filter gives feasible immutable labels. Raw SG on those labels exactly
equals filtered SG on the original labels in every evaluated case, verifying the wrapper identity. Feasibility
does not resolve future membership ambiguity: the first offer $[0,1)$
has identical observed features whether the input stops or later reveals
$[0,2)$, although its optimal membership differs. An immutable feasible
positive history cannot retain both overlapping offers.

We also concatenate each test tail, recomputing its optimum and causal
features without internal resets while keeping all learned coefficients,
standardization and validation thresholds fixed. History still starts
empty, with first gap zero, at the test boundary. NASA's sum of block optima is
$1{,}054{,}070$, compared with continuous optimum $1{,}032{,}493$;
CTC's corresponding values are $16{,}376{,}035$ and $6{,}390{,}734$.
There are 357/152,107 cross-block overlap pairs and 64/1,041 canonical
label changes. Hence the blocks are substantively different constructed
problems, particularly for CTC. Continuous-tail reward-selected cost SG
gives $97.16/94.03\%$, versus $97.26/94.95\%$ for LR; the CTC margin
over constant-negative SG is only $.0182$ points. This jointly changes
feature context, labels, geometry and normalization, so it does not isolate a boundary effect.

\paragraph{What these inputs mean for an application.}
Prefix-invariant features are causal only \emph{given known interval
endpoints}. SWF distinguishes realized runtime from requested time
\cite{pwaSWF}. NASA substitutes start times for unavailable submissions
\cite{pwaNASA}; CTC's actual waits and runtimes are retrospective
observations~\cite{pwaCTC}. Under the chronological interpretation,
ordinary-F1 filtered SG revokes 58/58 previously started NASA intervals
and 20/263 previously started CTC intervals. These are simulated revocations in historical traces. A real
reservation interface can specify start and duration while disallowing
deletion during running jobs~\cite{slurmReservations}. Thus a sufficient
model interpretation is a provisional planning batch completed before
any service starts; interleaved execution needs additional constraints
and analysis. The current data do not validate that application.

A valid application test needs complete offers, original reveal times,
then-available features and permission for the algorithm's actions.
Completed-service records and these adjacent-block comparisons establish
neither that match nor improvement over LR.

\section{Forecasting the effect of an advice bit}\label{app:rollout-interface}
An alternative to optimal-membership prediction is to estimate how the
current advice bit affects final retained length. This interface specializes
fixed-reference rollout \citep{tesauro1996rollout} to BANK. It preserves
arbitrary-bit robustness, with accuracy measured by conditional reward contrasts.

\paragraph{The conditional target.}
Consider a finite horizon with possibly earlier random termination and
integrable terminal rewards. The exogenous offer process is independent
of the algorithm's choices. A state $s$ contains the observed history,
time, held set, ledgers, permanent filter history, and any reference
memory. Fix a deterministic pure causal bit rule $A$, defined at every
reachable counterfactual state. Let $Q^A(s,b)$ be expected final length
after using bit $b$ now and following $A$ thereafter, under the true
conditional future-offer and stopping law. Define
\[
 \Delta_A(s)=Q^A(s,1)-Q^A(s,0).
\]
A paired estimate replays the same hypothetical suffix under both
current bits, with each branch continuing $A$ from its own state.
The deployed controller $\nu$ recomputes this estimate at each arrival
and chooses its sign; ties may use $A(s)$. Hypothetical suffixes must not
be the test episode's realized future. Historical samples alone also do
not certify that their conditional law matches deployment.

\paragraph{Why immediate retained value is insufficient.}
An advice bit can change a ledger even when both current actions reject
the request. One reachable example begins with a negative $[0,1)$,
followed by positive intervals
$[-.6,.01)$, $[.01,.5)$ and $[.5,.99)$. All three are rejected, leaving
the unit holder with ledger $1.59$. At $I=[.99,1.6)$ both bits again
reject; bit one raises its ledger to $2.20$. Now offer the negative
$[0,2.8)$. The smaller ledger permits replacement, while the larger one
prevents it. A final negative $[1,4)$ gives branch rewards $2.8$ and $4$;
replacing that final offer by $[3,6)$ gives $5.8$ and $4$.
The bit contrast is therefore $6/5$ or $-9/5$ despite identical
immediate reward and identical suffix lengths. The positive intervals
are compatible, so filtering does not erase the example. It diagnoses
information that the interface can use; it is not a performance
comparison of policies started from empty state.

\begin{proposition}[Expected transfer of a reference certificate]
\label{r32:prop:rollout-transfer}
Let $Y_A,Y_\nu$ be final rewards of the reference and deployed
controllers, both using filtered Shaped-BANK from empty state.
For an estimated contrast $\widehat\Delta_A$, suppose the quantities
below are integrable and put
\[
 \varepsilon_{\rm roll}
   =\mathbb E_\nu\sum_{t}
        |\widehat\Delta_A(s_t)-\Delta_A(s_t)|.
\]
Let $\eta_A$ be the error in~\eqref{new:eq:eta} of the raw labels
produced by $A$ on its own trajectory. Then
\begin{equation}
 \mathbb E Y_\nu\ge
 \max\left\{\frac{\mathbb E\mathrm{OPT}}{\varphi^5},
       \frac{\mathbb E\mathrm{OPT}-\mathbb E\eta_A}{d}
             -\varepsilon_{\rm roll}\right\}.
 \label{r32:eq:rollout-transfer}
\end{equation}
\end{proposition}
\begin{proof}
The finite-horizon performance-difference identity
\citep[Lemma~4.3]{ross2014interactive} gives
\[
 \mathbb E Y_\nu-\mathbb E Y_A
 =\mathbb E_\nu\sum_t
    (\nu(s_t)-A(s_t))\Delta_A(s_t).
\]
The sign rule satisfies
$(\nu-A)\widehat\Delta_A\ge0$ and $|\nu-A|\le1$,
so each summand is at least
$-|\widehat\Delta_A-\Delta_A|$.
Apply Corollary~\ref{new:cor:wrapper} to $A$'s own trajectory.
Independently, arbitrary-bit robustness applies pathwise to $\nu$,
giving the other branch of~\eqref{r32:eq:rollout-transfer}.
\end{proof}

The error is measured on $\nu$'s visited states, not on a fixed archive
of reference states; it includes scenario-model and sampling error.
The member-label error $\eta_A$ belongs to $A$, not to $\nu$.
Exact conditional contrasts remove $\varepsilon_{\rm roll}$, but do
not turn rollout bits into perfect membership predictions or supply
a per-execution consistency-$d$ guarantee. With an all-negative
reference, $\eta_A=\mathrm{OPT}$ and the second branch is vacuous.
Thus a good causal reference is a substantive part of this interface.

\subsection{Complete-policy diagnostics}
We use two fixed synthetic offer laws, with all policies starting empty.
The 16 interval types combine starts $\{0,3,6,9\}$ and lengths
$\{1,3,6,12\}$. Initial types are uniform; episodes stop with
probability $.1$ after each offer and are capped at 24 offers. IID
transitions are uniform. In the Context law, an observed binary context
is fixed for each episode. Given current type $i$, next-type weights are
$1+4\mathbf1\{j\text{ compatible with }i\}+2\mathbf1\{w_j\le3\}$
in context zero and
$1+4\mathbf1\{j\text{ contains }i\}+2\mathbf1\{w_j\ge6\}$
in context one. Both laws have full support. Laplace-smoothed counts fit
conditional or pooled transitions; fitted stopping accounts for censoring
at the cap. 

\paragraph{Learning versus lookahead.}
The original attribution study used six independent Context training
trials and two IID controls, each with 2,048 training and 64 fresh test
episodes. Every method continues with the all-negative reference after
the current advice bit. Conditional, pooled and uniform models differ
only in their future laws; uniform learns no transitions and knows the
specified stopping rate. A known-law control knows the generator but
never the realized future. The original Monte Carlo comparison used 16/128 common-random-number
suffixes; the finite-expectation analysis below reuses its fitted models and test episodes.

\begin{table}[t]\centering\small
\caption{Fixed all-negative-reference attribution cohort: six Context
trials and 384 test episodes. No-MC columns use certified finite
expectations on the same models and episodes as the 128-suffix column.
Retained and revoked lengths are in the original interval units;
the primary paired outcome is retained length.}
\label{tab:r33-expectation}
\begin{tabular}{@{}lrrrr@{}}\toprule
& 128 suffixes & No MC & No MC & No MC\\
Method & \% OPT & \% OPT & Retained & Revoked\\\midrule
Conditional & 94.334 & 94.817 & 13.76302 & 4.86719\\
Pooled & 94.474 & 94.458 & 13.68750 & 4.91406\\
Uniform & 94.143 & 94.423 & 13.68229 & 4.66406\\
Known law & 94.660 & 94.665 & 13.74740 & 4.89063\\
Myopic & --- & 90.652 & 12.91667 & 2.45313\\
All negative & --- & 87.297 & 12.45573 & .83333\\
\bottomrule\end{tabular}\end{table}

Finite-state recursion evaluates the 512 test episodes (384 Context,
128 IID) through the remaining horizon. Outward bounds certify nonzero contrast signs under the sampler's 53-bit law, and exact
rational recursion certifies the remaining 26 zeros.

On the 384 Context episodes, conditional minus uniform gains $.08073$
in mean retained length and $.3943$ percentage points in mean
$\mathrm{ALG}/\mathrm{OPT}$ (13 wins, 365 ties, 6 losses). Four of six trial length differences are positive and two negative.
Uniform minus myopic gains $.76563$ length and $3.7711$ points.
Conditional minus pooled gains $.07552$ length and $.3596$ points,
compared with $-.1409$ points using 128 suffixes. Sampling precision
therefore affects this cohort's comparison. In all 128 IID episodes,
the four no-Monte-Carlo laws have identical terminal values, averaging
$95.4432\%$ of OPT. These reused-test results are descriptive; Table~\ref{tab:r34-g-confirmation} reports an
independent cohort.

\paragraph{Independent cohort.}
The independent cohort used 24 Context and eight IID training trials,
each with 2,048 training and 64 test episodes under the same generator,
fitting rules, all-negative reference and finite-expectation evaluation.
All seven policies started empty; a swapped control exchanged the fitted
context rows. The primary contrast was conditional minus uniform length,
averaged within training trial (Table~\ref{tab:r34-g-confirmation}).
Secondary intervals are descriptive and unadjusted for multiplicity.

\begin{table}[t]\centering\small
\caption{Independent Context cohort: 24 training trials and 1,536
 test episodes. Differences are in retained-length units; Student-$t$
 intervals use 23 degrees of freedom. Episodes sharing a fitted model
 are not independent training replications.}
\label{tab:r34-g-confirmation}
\begin{tabular}{@{}lrr@{}}\toprule
Predeclared comparison & Mean difference & Trial-level 95\% interval\\\midrule
Conditional $-$ uniform (primary) & $+.003255$ & $[-.041704,+.048214]$\\
Conditional $-$ pooled & $-.010417$ & $[-.053072,+.032239]$\\
Conditional $-$ swapped & $+.024740$ & $[-.033268,+.082747]$\\
Uniform $-$ myopic & $+.462891$ & $[+.375407,+.550374]$\\
Known law $-$ uniform & $+.019531$ & $[-.019780,+.058843]$\\
\bottomrule\end{tabular}\end{table}

The primary trial means split 11 positive, three zero and ten negative
(22 episode wins, 1,496 ties, 18 losses). Its interval confirms neither a
positive mean gain nor a zero effect; the earlier six-trial positive cohort
remains separate. Uniform minus myopic gains $2.0836$ percentage points
in mean $\mathrm{ALG}/\mathrm{OPT}$ (descriptive 95\% interval $[1.5439,2.6234]$).
Conditional/pooled/uniform mean test negative log-likelihoods per opportunity, averaging
trials equally, are $2.630468/2.772968/2.818286$: better future-law fit does
not establish added terminal reward. In the eight IID trials, conditional
minus uniform is $-.007813$ length, with 95\% interval $[-.049707,.034082]$.

\paragraph{Earlier cohorts and interpretation.}
Earlier exploratory cohorts used four fitted training trials per law
and separate test episodes. They do not provide uniform learning gains:
with an all-negative reference, fitted Context rollout averaged
$92.83\%$ versus uniform's $94.36\%$, and increasing suffixes from
16 to 128 lowered another fitted Context mean from $95.44\%$ to $94.98\%$.
A learned-reference diagnostic instead fit $A$ from context, type and
index bin with Laplace smoothing and threshold $1/2$. Its error branch
was active in 134/256 episodes; rollout improved every trial mean over
$A$, but fitted-minus-uniform averaged $-.25/+.13$ points in IID/Context.
The earlier cohorts' shared fitted models and different references preclude pooling with
the cohorts above.

The known-law controller is not optimal continuation: it follows its
fixed reference in hypothetical suffixes and replans at the next actual
arrival. All tables average episode ratios, not
$\mathbb E\mathrm{ALG}/\mathbb E\mathrm{OPT}$ in the certificate.
Cancellation is material: in the attribution cohort, mean revoked
length is $2.4531$ for myopic, $4.6641$ for uniform and $4.8672$ for
conditional rollout; in the learned-reference cohort it rises from
$1.70/1.85$ for $A$ to $4.39/4.16$ with fitted-law rollout.
On the original attribution cohort's fixed Context trajectories, subtracting a proportional
fee $c$ per revoked length gives conditional minus uniform mean net
reward $(31-78c)/384$. It stays positive throughout the unchanged
algorithm's fee domain $0\le c<1/(2\varphi)$; the six trial differences
remain four positive and two negative. This does not hold across all
environments: the IID rollout-minus-myopic mean is $(63-245c)/128$,
which turns negative above $c=9/35$ within that domain. This is
retrospective accounting in original length units, with no new samples
or fee-aware reoptimization.

On the new fixed Context trajectories, charging $c$ per revoked length
gives conditional minus uniform mean net reward $(5-190c)/1536$ and
uniform minus myopic $(711-2887c)/1536$. Their sign changes at $1/38$
and $711/2887$ both lie inside $0\le c<1/(2\varphi)$. Thus the original
cohort's positive conditional-minus-uniform net difference throughout that domain does not extend to this cohort.

These studies do not estimate deployment-state error in~\eqref{r32:eq:rollout-transfer},
and matched suffix counts do not equalize inference costs.

\section{Further model comparisons}\label{sec:comparison-appendix}
\subsection{Further related models and source scope}
\label{app:further-literature}

\paragraph{Guarantees in the same model.}
The extended version of \citet{karavasilis2025}, Theorem~3.5, proves
proportional-weight Naive; Appendix~\ref{new:app:filter} gives our
self-contained proof. The proceedings, Section~4, Algorithm~3 and
Lemma~4.5, give LR's maximum-conflict rule at parameter $\varphi$ with
competitive ratio at most $2\varphi+1$, attributed there to \citet{garay1997}.
Section~\ref{sec:intro} cites the stopping-prefix antecedent and
finitely robust $1$-increasing barrier. The latter uses $\alpha=1$: for $\alpha<1$,
its construction need not force rejection, since the internal request
has length $v-2\epsilon>\alpha v$ when $\epsilon<(1-\alpha)v/2$.

The matching prediction-free lower bound is Theorem~1 in Section~2 of the October 28, 1995 author manuscript of \citet{tomkins1995}. It allows preemption of infinite-duration calls on a line,
with length rewards paid only to calls retained at termination.
Represent a route from $i$ to $j$ by $[i,j)$: shared-edge conflict,
endpoint compatibility, length reward, and permanent rejection or
preemption agree with our model. Fix every prediction to zero. Any
deterministic algorithm here then induces a prediction-free algorithm
on this subclass, so $R\ge2\varphi+1$. No request expires, and arbitrary
stopping prefixes are allowed. This is a supremum lower bound, not a
claim that a finite witness attains equality. 

\paragraph{Relation to earlier charging arguments.}
For the unweighted $2k$ guarantee,
\citet[Section~3, Propositions~1--2]{borodin2023anyorder} charge optimal
intervals to holders and transfer charges through replacements, separately
bounding two direct and at most $2k-2$ transferred charges. We retain this
device, adding cost-weighted ancestry antichains and the two-cut functional
$F_J(x,y)$ to bound the remaining geometric gap and both exterior charges
jointly under special, potentially length-decreasing replacements and
nonconstant ledger costs.

\paragraph{Prediction and decision interfaces.}
Our bits refer to one fixed final optimum, with deterministic robustness
for every fixed label sequence. In contrast, \citet[Section~2]{dallot2026v2}
use answer-dependent guidance whose replacement by adversarial guidance
is selected independently at each step. TrustGreedy of \citet{boyar2026}
maximizes cardinality using an upfront prediction of the full input set;
it revises a planned solution while preserving irrevocably accepted
requests. \citet{li2026simplechain} study unit-length, unit-profit simple
chains under random and adversarial orders, including complexity with
a trustworthy advice tape.

Several adjacent guarantees allow different decisions.
\citet[Theorems~4.1 and~5.2]{daneshvaramoli2025knapsack} mix fractional
knapsack decisions from critical-value predictions; their integral
conversion assumes bounded unit values and sufficiently small item
weights. It addresses a capacity constraint, not interval conflicts.
\citet{grigorescu2025} use real
packing decisions, advice interpolation and a quantified feasibility
relaxation. Scheduling with partial size predictions permits rate sharing
and resumption \citep{benomar2024partial}; paging permits reloads and a
consumable, resettable relative prediction budget \citep{chen2026paging}.
These mechanisms differ from BANK's inherited positive mass and permanent
rejection. The request-prediction framework of
\citet[Section~1, Theorems~1.1--1.2]{ameli2026covering} uses irrevocable
covering decisions, feasibility under unions, and subadditive costs.
Feasible interval selections need not remain feasible under unions,
so that covering guarantee does not directly transfer.

\paragraph{Weights, order, and source versions.}
The technical version of \citet{berg2025faw}, Corollary~36, treats unit
profits, irrevocable decisions and bounded overlap $t$, with
$1\le\alpha\le t$ and decision bit $0$ denoting acceptance. Splitting
length-weighted intervals into unit requests can change feasibility,
prediction identities and available online decisions, so it does not
transfer that guarantee. \citet{antoniadis2026switching} use membership
bits and length weights, but require release-time order and irrevocable
selection; their guarantees involve additive terms, length ratios or
randomization. Our results instead use arbitrary order and strict ratios.
Berg and Boyar are compared through their public technical versions;
the simple-chain source is arXiv:2608.10376v1, and the paging and
concave-packing sources are arXiv:2606.01342v3 and
arXiv:2411.08332v1.

\paragraph{Advice, learning, and randomization.}
The Pareto viewpoint for trusted and untrusted advice precedes this work
\citep{angelopoulos2020}; \citet{mitzenmacher2020} survey prediction-based
algorithms. \citet{khodak2022} learn predictors through algorithm-dependent
quality measures, without directly supplying a guarantee for our
membership interface. Theorem~2.1 of \citet{bendavid1994} derandomizes
against adaptive offline adversaries in a request-answer model. Our
mixture instead fixes the complete labelled input before randomization,
with consistency and robustness suprema over different input classes;
that theorem alone does not yield an adaptive two-criterion extension.

\paragraph{Immediate service and learning.}
Capacity-constrained online learning loses delayed feedback when a
tracked round is abandoned \citep{ryabchenko2025capacity}; its regret
objective depends on which observations survive. Our duration guarantee
does not imply a learning-regret bound.
\citet{wang2026immediate} study no-buffer scheduling with abandonment
and censored reward feedback. Their general objective uses task rewards,
not occupied duration. With one machine, arrival times $b_i$, known lengths $\ell_i$,
rewards equal to $\ell_i$, and time to finish the final accepted jobs,
the intervals $[b_i,b_i+\ell_i)$ form a start-ordered subclass of our
model. Completed jobs can remain in the held-set record because later
arrivals cannot conflict with them. Our upper guarantees therefore
apply to this restricted case; our full-class lower bounds and their
bandit regret bounds do not transfer between the two models.

\paragraph{Earlier scheduling and nearby tradeoffs.}
Early interval scheduling includes \citet{lipton1994} and
\citet{woeginger1994}; the latter rewards completion of fixed-start/end jobs
started at arrival and lost if aborted, rather than terminal held-set rewards.
\citet{seiden1998} improves the randomized guarantee in this start-ordered setting,
distinguished from any-order selection by
\citet[Section~1]{borodin2023anyorder}; its guarantees do not replace LR
for arbitrary order.

\citet{fung2012} randomize initially between two deterministic schedulers;
arrival order and completion rewards prevent transferring their constants.
\citet{jin2022} prove a tight two-stage matching tradeoff where naive coin
interpolation can be suboptimal for vertex weights. Its advice/hindsight
ratios and irrevocable batches do not give an interval lower bound.
\citet{antoniadis2020} robustify metrical-task and line-matching predictions
through different states and interfaces. \citet{nair2021} learn MIP partial
assignments and branching decisions.

\end{document}